\RequirePackage{fix-cm}

\documentclass[smallextended]{svjour3}       

\smartqed  

\spnewtheorem{definitionitalic}{Definition}{\bf}{\it}
\spnewtheorem*{nonumtheorem}{Theorem}{\bf}{\it}

\usepackage{indentfirst}
\usepackage{amsmath,amssymb}
\usepackage{graphicx,color}
\usepackage[all,cmtip]{xy}
\usepackage{bbm}
\usepackage{hyperref}
\usepackage{mathtools}
\usepackage{verbatim}
\usepackage{stmaryrd}
\usepackage{slashed}
\usepackage[shortlabels]{enumitem}
\usepackage{tikz}
\usepackage{colortbl}
\usetikzlibrary{matrix}
\usetikzlibrary{backgrounds}

\usepackage{bm}
\usepackage[margin=20pt,font=footnotesize]{caption}

\definecolor{lcolor}{rgb}{0.6,0.3,0.3}
\hypersetup{colorlinks=true,linkcolor=lcolor,urlcolor=lcolor,citecolor=lcolor}

\definecolor{qcolor}{rgb}{0.19,0.55,0.91}
\definecolor{hcolor}{rgb}{0.9,0.2,0.5}
\definecolor{scolor}{rgb}{0.9,0.5,0.2}

\newcommand{\tc}[2]{\textcolor{#1}{#2}}

\newcommand{\phiel}{\tc{scolor}{\varphi}}
\newcommand{\eel}{\tc{hcolor}{\gamma}}
\newcommand{\phiconstraint}{\tc{qcolor}{\chi}}
\newcommand{\scl}{\tc{qcolor}{\zeta}}

\newcommand{\bklq}{\tc{qcolor}{p_0}}

\newcommand{\inphi}[1]{\smash{\underline{#1}}} 

\newcommand{\E}{\mathcal{E}}
\newcommand{\A}{\mathcal{A}}

\renewcommand{\P}{\mathcal{P}}
\renewcommand{\L}{\mathcal{L}}
\newcommand{\I}{\mathcal{I}}

\newcommand{\p}{\partial}
\newcommand{\N}{\mathbbm{N}}
\newcommand{\Z}{\mathbbm{Z}}
\newcommand{\R}{\mathbbm{R}}
\newcommand{\C}{\mathbbm{C}}
\newcommand{\one}{\mathbbm{1}}

\newcommand{\eps}{\epsilon}

\definecolor{cerulean}{rgb}{0.0, 0.48, 0.65}
\definecolor{brightube}{rgb}{0.82, 0.62, 0.91}

\newcommand{\step}{\vskip 3mm}

\newcommand{\anc}{\tc{qcolor}{\rho}}
\newcommand{\spar}{\tc{scolor}{s}}

\newcommand{\ff}[1]{\mathcal{#1}}

\newcommand{\gx}{{\mathfrak g}}

\renewcommand{\sl}{{\mathfrak{sl}}_{2}}

\DeclareMathOperator{\sheafHom}{\mathscr{Hom}\text{\kern -3pt {\calligra\large om}}\,}

\DeclareMathOperator{\rank}{rank}

\DeclareMathOperator{\image}{im}

\DeclareMathOperator{\End}{End}
\DeclareMathOperator{\Hom}{Hom}

\DeclareMathOperator{\MC}{MC}

\DeclareMathOperator{\RE}{Re}

\DeclareMathOperator{\CDer}{\tc{scolor}{CDer}}

\DeclareMathOperator{\Ric}{Ric}
\DeclareMathOperator{\Orth}{O}
\DeclareMathOperator{\Der}{Der}

\DeclareMathOperator{\diag}{diag}

\DeclareMathOperator{\Sym}{Sym}

\newcommand{\cinf}{{\tc{scolor}{C^\infty}}}

\newcommand{\dd}{d} 

\renewcommand{\subset}{\subseteq}

\newcommand{\sv}[1]{\tc{qcolor}{\mathbf{s}_{#1}}}

\newcommand{\dimx}{{\tc{hcolor}{N}}}
\newcommand{\dimy}{{\tc{scolor}{n}}}
\newcommand{\dimz}{{\tc{qcolor}{m}}}
\newcommand{\fld}{{\tc{qcolor}{u}}}
\newcommand{\ind}{{\tc{scolor}{\xi}}}

\newcommand{\ptwo}{Subsection \ref{section:analysis2}}

\newcommand{\keroneform}{\tc{hcolor}{\omega}}

\newcommand{\spx}{\,\,\,}

\newcommand{\amc}{\tc{hcolor}{\psi}}

\newcommand{\normp}{{\tc{hcolor}{P}}}

\newcommand{\GS}{{\tc{hcolor}{G}}}
\newcommand{\GB}{{\tc{qcolor}{X}}} 
\newcommand{\Aspec}{\A_{\textnormal{gauged}}} 
\newcommand{\Aspecaffine}{\A_{\textnormal{affine gauged}}}

\newcommand{\ginj}{{\tc{scolor}{R}}}

\newcommand{\usol}{{\tc{scolor}{\gamma}}}
\newcommand{\uzero}{\widetilde{\usol}}

\newcommand{\epoly}{\econst[\tau]}
\newcommand{\econst}{\textcolor{scolor}{\mathcal{U}}}

\newcommand{\mcs}{\smash{\MC(\A)_{\textnormal{gauged}}}}
\newcommand{\mcsp}{\smash{\MC(\A)_{\textnormal{gauged}}^+}}

\usepackage{mathptmx}

\begin{document}

\title{Semiglobal non-oscillatory
       big bang singular spacetimes\\
       for the Einstein-scalar field system}

\author{
Andrea N\"utzi, \\
Michael Reiterer, \\
Eugene Trubowitz
}

\institute{Andrea N\"utzi,
              ETH Zurich,
              \email{andrea.nuetzi@math.ethz.ch}\\
           Michael Reiterer,
              Hebrew University of Jerusalem,
              \email{michael.reiterer@protonmail.com}             \\
           Eugene Trubowitz,
              ETH Zurich,
              \email{eugene.trubowitz@math.ethz.ch}
}

\date{Received: date / Accepted: date}


\maketitle
\begin{abstract}
We construct semiglobal
singular spacetimes for the Einstein equations coupled to a massless scalar field.
Consistent with the heuristic analysis of
Belinskii, Khalatnikov, Lifshitz or BKL for this system,
there are no oscillations due to the scalar field.
(This is much simpler than the
oscillatory BKL heuristics for the Einstein vacuum equations.)
Prior results are due to Andersson and Rendall in the real analytic case,
and Rodnianski and Speck in the smooth
near-spatially-flat-FLRW case.
Similar to Andersson and Rendall we give asymptotic data at the singularity,
which we refer to as final data,
but our construction is not limited to real analytic solutions.
This paper is a test application of tools
(a graded Lie algebra formulation of the Einstein equations
 and a filtration)
intended for the more subtle vacuum case.
We use homological algebra tools to construct a formal series solution,
then symmetric hyperbolic energy estimates to construct a true solution well-approximated by truncations of the formal one.
We conjecture that the image of the
map from final data to initial data 
is an open set of 
anisotropic initial data.
\end{abstract}

\tableofcontents

\section{Introduction}

The Einstein equations for a Lorentzian metric coupled to a real scalar field are
\begin{equation}\label{eq:esf}
\begin{aligned}
\Ric_g & = 2 \dd \phi \otimes \dd \phi\\
\Box_g \phi & = 0
\end{aligned}
\end{equation}
Belinskii and Khalatnikov \cite{bkscalar} developed
the heuristics for a class of solutions with a spacelike singularity
along which curvature invariants diverge.
This requires $\phi \neq 0$.
The heuristics are much simpler than for vacuum,
$\Ric_g = 0$,
which are oscillatory according to Belinskii, Khalatnikov and Lifshitz \cite{bkl}.
\step
Andersson and Rendall \cite{ar}
rigorously constructed a large
class of real analytic solutions to \eqref{eq:esf} implementing the
non-oscillatory heuristics,
using Fuchsian systems.
By a different approach,
Rodnianski and Speck \cite{rs,rs2}
constructed smooth solutions to \eqref{eq:esf}
and in particular they showed that the class contains an open ball
of initial data in some gauge;
the results are for near-spatially-flat
Friedman-Lemaitre-Robertson-Walker or FLRW
spacetimes
and the spatial topology is a torus $\mathbbm{T}^3$.
A key difference is that Andersson and Rendall give
asymptotic data
at the singularity (final data) and
solve the equations away from the singularity,
whereas Rodnianski and Speck give initial data
along a spacelike hypersurface
and solve the equations towards the singularity.
It is possible that both approaches can be pushed further
in the context of the Einstein equations with scalar field:
\cite{ar} to include smooth spacetimes using energy estimates\footnote{%
A discussion of difficulties encountered
 when extending this to smooth solutions is in \cite{r}.
},
and \cite{rs,rs2} to construct solutions that are not necessarily near isotropic\footnote{%
Perhaps a careful examination of the constants in the estimates in \cite{rs,rs2},
and small adaptations, would show that this approach is not in a significant sense
limited to solutions that are nearly isotropic.
In \cite{rs} the deviation from being isotropic shows up, in particular,
in the traceless part of the second fundamental form of level sets of the time function,
and this is zero in the isotropic case.}\textsuperscript{,}\footnote{%
Both \cite{ar} and \cite{rs} also deal with the stiff fluid equations,
$\Ric_g = 2v \otimes v$ for a one-form $v$.
Usually $v$ is required to be future timelike.
A scalar field solution yields a stiff fluid solution
with $v = \dd \phi$.}.

\newcommand{\xa}{(-0.5,-2.) -- (-0.5975,-1.71475) -- (-0.69,-1.458) -- (-0.7775,-1.22825) -- (-0.86,-1.024) -- (-0.9375,-0.84375) -- (-1.01,-0.686) -- (-1.0775,-0.54925) -- (-1.14,-0.432) -- (-1.1975,-0.33275) -- (-1.25,-0.25) -- (-1.2975,-0.18225) -- (-1.34,-0.128) -- (-1.3775,-0.08575) -- (-1.41,-0.054) -- (-1.4375,-0.03125) -- (-1.46,-0.016) -- (-1.4775,-0.00675) -- (-1.49,-0.002) -- (-1.4975,-0.00025) -- (-1.5,0.)}%
\newcommand{\xb}{(1.5,0.) -- (1.4975,-0.00025) -- (1.49,-0.002) -- (1.4775,-0.00675) -- (1.46,-0.016) -- (1.4375,-0.03125) -- (1.41,-0.054) -- (1.3775,-0.08575) -- (1.34,-0.128) -- (1.2975,-0.18225) -- (1.25,-0.25) -- (1.1975,-0.33275) -- (1.14,-0.432) -- (1.0775,-0.54925) -- (1.01,-0.686) -- (0.9375,-0.84375) -- (0.86,-1.024) -- (0.7775,-1.22825) -- (0.69,-1.458) -- (0.5975,-1.71475) -- (0.5,-2.)}%
\begin{figure}[h]
  \centering
  \begin{tikzpicture}
  \draw (-1.5,0) -- (-1.5,0.2);
  \draw (1.5,0) -- (1.5,0.2);
  \draw (0.5,-2) -- (0.5,-2.2);
  \draw (-0.5,-2) -- (-0.5,-2.2);
  \draw [->] (-4,-1) -- (-4,-0.5) node [anchor = east, xshift=-2] {time};
  \draw [->] (-4,-1) -- (-3.5,-1) node [anchor = north, yshift=-2] {space};
  \draw [ultra thick] (-3,0) -- (3,0)
     node [text width = 75pt, anchor = west, xshift = 2]
         {$\tau = +\infty$\\ curvature singularity\\at finite proper time};
  \draw [dashed] (-3,-2) -- (3,-2)
     node [text width = 85pt, anchor = west, xshift = 2]
         {$\tau = 0$\\ initial hypersurface};
  \begin{scope}[shift={(0,0)}]
    \draw (0,0) node [anchor = south] {$A$};
    \draw (0,-1) node {future of $B$};
    \draw (0,-2) node [anchor = north] {$B$};
    \begin{scope}[on background layer]
       \draw [draw opacity=0,fill=gray!20] \xa -- \xb -- cycle;
    \end{scope}
    \draw \xa;
    \draw \xb;
  \end{scope}
  \end{tikzpicture}
\caption{%
Shown is the domain $[0,\infty) \times U$
where $U \subset M^3$ is a small patch of space.
An important property of all spacetimes constructed in this paper
is that they have particle horizons.
This is indicated by the trapezoidal domain
which is causally independent from its complement
going away from the singularity,
meaning final data along the top of the trapezoidal domain, $A$,
determines initial data along the bottom of the trapezoidal domain, $B$.
}\label{fig:causalx}
\end{figure}
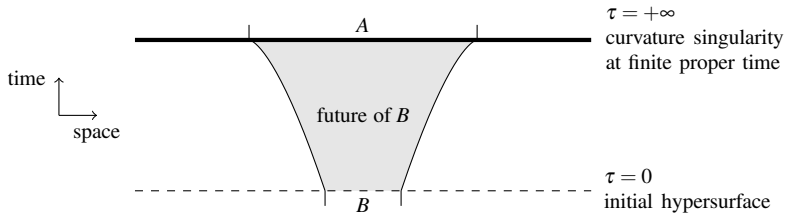

The present paper could roughly be seen as pushing the final data perspective of \cite{ar} further.
The solutions we construct are generally not real analytic or near isotropic.
Actually, our approach is more suited to anisotropic elements, see Remark \ref{remark:anisotropic}.
We construct a large class of smooth solutions on semiglobal domains
\[
[0,\infty) \times M^3
\]
where $M^3$ is a parallelizable closed 3-manifold,
for instance the sphere $S^3$ or the torus $\mathbbm{T}^3$.
The spatial topology is actually not very important
because all solutions constructed here
have particle horizons, see Figure \ref{fig:causalx}\footnote{%
 Therefore using the finite speed of propagation theorem
 for hyperbolic equations one can extract various kinds of statements
 that are local in space.}.
The time coordinate $\tau \in [0,\infty)$ is such that
there is a curvature singularity at $\tau = +\infty$,
but beware that this corresponds
to finite proper time.
Expansions near the singularity
are obtained as part of the construction of these solutions.
(A similar technique was used to prove
trapped spheres formation \cite{foc}, simplifying earlier work \cite{chr}.)
Similar to \cite{ar} in the real analytic class,
we expect that our construction yields an open set of initial data
in the smooth class.
Whereas \cite{ar} count the number of free functions,
we show completeness at a formal perturbative level
(all formal solutions are obtained up to gauge transformations).
\step
This paper 
uses the language and algebraic setup
developed in \cite{rtgla,rtfil,rtfil2}.
Briefly, our methods are geared towards
 the vacuum case
\begin{equation}\label{eq:vc}
\Ric_g = 0
\end{equation}
where the BKL heuristics predict a class of singularities with infinitely many oscillations.
From the perspective of mathematics this case is wide open.
One can view \eqref{eq:esf}
as training ground to learn from.
Another important training ground are
oscillatory spatially homogeneous solutions to \eqref{eq:vc},
briefly reviewed in Subsection \ref{subsec:heuristics}.
\step
This introduction
contains an informal conjecture
(Subsection \ref{sec:informalconjecture}),
a summary of our main results
(Subsections \ref{sec:informalmain}),
and an overview of the BKL heuristics
with and without oscillations for general orientation
(Subsection \ref{subsec:heuristics}).

\subsection{An informal conjecture}
\label{sec:informalconjecture}

To orient the reader,
we conjecture what form a strong implementation of
the BKL heuristics for \eqref{eq:esf} could take.
On the manifold $[0,\infty) \times M^3$ define:
\begin{itemize}
\item The graded Lie algebra in \cite{rtgla}
for the Einstein equations
extended to include a scalar field, see Section \ref{section:gla}.
This extended graded Lie algebra is denoted $\E_\Phi$.
It is actually a graded Lie algebroid,
which encodes its differential geometric nature.
\item
An increasing $\N^3$-indexed graded Lie algebra filtration
of $\E_\Phi$.
This filtration from \cite{rtfil}
is here extended to include a scalar field, see Section \ref{section:filtration}.
Its associated graded is denoted $\A$.
So $\A = \P/\spar\P$
where\footnote{%
Actually $\spar = (\spar_1,\spar_2,\spar_3)$,
because the filtration is indexed by $\N^3$.
We gloss over this in the introduction.
}
$\P \subset \E_\Phi[[\spar]]$ is the Rees algebra.
\end{itemize}
In general, in a graded Lie algebra, a Maurer-Cartan or MC element
is an element $x$ of degree one that solves the Maurer-Cartan equation\footnote{%
This equation is better known in another branch of mathematics,
namely the study of
deformations of algebraic structures. Quoting \cite{niri}:
`Our basic observation is that a wide class of algebraic structures 
on a vector space can be defined by essentially the same equation in 
an appropriate graded Lie algebra.
Among the structures thus obtained are Lie algebras,
associative algebras, commutative and associative algebras,
extensions of algebras (of any of the above types), 
and representations of algebras.'}
\[
[x,x] = 0
\]
In $\E_\Phi$ this is a first order partial differential equation
that is hyperbolic (up to gauge transformations) and in this sense a dynamical system.
The construction of $\E_\Phi$ is such that nondegenerate elements
of $\MC(\E_\Phi)$ are solutions to the Einstein equations \eqref{eq:esf},
see Definition \ref{def:nondeg}
and Lemmas \ref{lemma:glavac}, \ref{lemma:gsf}.
Gauge transformations are indicated by ${\sim}$.

\begin{conjecture}[Informal] \label{conj}
Scattering at $\tau \to \infty$
from the dynamical system defined by $\A$
to the dynamical system defined by $\E_\Phi$
defines a smooth map
\begin{equation}\label{eq:scx}
U \subset \frac{\MC(\A)}{\sim} \;\to\; \frac{\MC(\E_\Phi)}{\sim}
\end{equation}
that is a diffeomorphism onto its image
in suitable function spaces.
The domain of definition $U$ is open and contains
the anisotropic spatially homogeneous elements
with positive Kasner parameters (Lemma \ref{lemma:she} and Definition \ref{def:anisotropic}).
The nondegenerate elements in the image are semiglobal solutions to
\eqref{eq:esf} with a curvature singularity reached in finite proper time,
and with particle horizons. 
\end{conjecture}

Scattering means that two dynamical systems have
asymptotically the same behavior and therefore locally diffeomorphic
solution spaces, typically one simple system and one complicated system\footnote{%
For instance, in classical mechanics, a system of free particles
 and a system with short range interactions.
 For a discussion of scattering in the context of PDE,
 though not gauge theories, see \cite{tao}.}.
In Conjecture \ref{conj} the simple system is the one defined by $\A$,
because the $\MC(\A)$ equations can be solved explicitly\footnote{
The $\MC(\A)$ equations are related to equations in the BKL literature
called velocity-term-dominated or VTD.
However note that $\A$ is a more comprehensive object because it
is a graded Lie algebra.}\textsuperscript{,}\footnote{%
The nonlinear constraint equations
can be analyzed in perturbative regimes,
see Appendix \ref{app:constraints}.}.
\step
Informally, \eqref{eq:scx} is interpreted as the map
from asymptotic final data to initial data.
That is, in this paper the $\MC(\A)$ elements serve as final data.


\subsection{Summary of the main results}
\label{sec:informalmain}

Our results are pieces of Conjecture \ref{conj}. 
Define as above the graded Lie algebra and filtration
on the 4-dimensional domain $[0,\infty) \times M^3$.
We use a subspace $\Aspec \subset \A^1$
respectively an affine shift $\Aspecaffine \subset \A^1$ as a gauge,
that imposes fiberwise linear constraints on frame and connection,
see Definition \ref{def:agsub}.
Partial justification for this gauge is contained in b) below.
\begin{theorem}[Summary]\label{thm:informalintro}
Let $M^3 = \mathbbm{T}^3$.
Let
\[
\mcsp \;\subset\; \MC(\A) \cap \Aspecaffine
\]
be elements with positive Kasner parameters.
See Definitions \ref{def:mcsnew}, \ref{def:mcsp}. Then
\begin{itemize}
\item[a)]
Final data and asymptotic constraints:
The elements in $\mcsp$
are in one-to-one correspondence with
the solutions to an underdetermined
system of first order differential equations on $M^3$ 
that play the role of asymptotic constraint equations.
See Lemma \ref{lemma:mcael}.
A description of the space of solutions
near anisotropic spatially homogeneous elements is in Appendix \ref{app:constraints}.
\item[b)] Formal solutions: There is a smooth map
\[
S_{\textnormal{formal}} \;:\; \mcsp \;\to\; \MC(\P)
\]
that is right inverse to the canonical map $\pi : \MC(\P) \to \MC(\A)$.
This means that there are no obstructions to extending
any given element of $\mcsp$ to a formal power series solution.
See Theorem \ref{thm:ex}.
Formal completeness holds
for anisotropic elements $x \in \mcsp$
in the sense that the $S_{\textnormal{formal}}$-induced map
\[
\MC(\A[[\spar]]) \cap (x + \spar \Aspec[[\spar]])
\;\to\;
\frac{\pi^{-1}(\{x\})}{\exp(\spar\P^0)}
\]
is surjective. In this sense one gets all formal solutions.
See Lemmas \ref{lemma:gae} and \ref{lemma:fcomp}.
\item[c)] Actual solutions: 
There is a map 
\[
S\;:\; U \subset \mcsp \to \MC(\E_\Phi)
\]
where $U$ encodes smallness conditions
(see the hypotheses of Theorem \ref{theorem:sendg},
Remark \ref{remark:s1})
but contains all
spatially homogeneous elements
with positive Kasner parameters (Lemma \ref{lemma:she});
for every $x \in U$ the corresponding $y = S(x)$
is asymptotic, as $\tau \to \infty$, to $x$ with
$\spar_1 = \spar_2 = \spar_3 = 1$;
the element $y$
corresponds to a smooth metric and scalar field 
that solve \eqref{eq:esf}, with curvature singularity
in finite proper time, and particle horizons.
See Theorem \ref{theorem:sendg} and Remark \ref{remark:s1}.
\end{itemize}
\end{theorem}

The assumption $M^3 = \mathbbm{T}^3$ is for simplicity,
most statements in the paper are for general $M^3$.
The theorems where we analyze symmetric hyperbolic equations,
in Section \ref{sec:semi}, are for $\mathbbm{T}^3$ only
for convenience so that one has a global coordinate system.
\step
Here are further comments about the statements in the theorem:
\begin{itemize}
\item[b)]
A key step for formal completeness is the
realizability of the gauge
  -- 
every MC element in $x+\spar\A[[\spar]]$ 
has a representative in $x+\spar\Aspec[[\spar]]$.
This is a post hoc justification of the definition of $\Aspec$
in the context of formal solutions.
Note that the homological language allows
for a natural formulation of completeness.
\item[c)]
This existence statement for smooth solutions
uses symmetric hyperbolic gauges from \cite{rtgla},
extended to the scalar field in Section \ref{section:gla},
making energy estimates available
which are indispensable for hyperbolic equations\footnote{%
Note for comparison that \cite{rs,rs2}
use a gauge where the equations are not purely hyperbolic.
Their motivation is to make a construction where the singularity
ends up automatically, and not by a post hoc coordinate transformation,
along a coordinate level set of a time coordinate $t$.
The situation in this paper is different since
we only solve the equations away from the singularity.
(The discussion in Section \ref{sec:discussion}
 proposes using separate gauges for low frequencies and late times
 versus high frequencies and early times.)}.
\end{itemize}

We expect that $S$ is smooth and that $S_{\textnormal{formal}}$ is obtained
by differentiating $S$, but we do not prove this.
We expect that $S$ satisfies a completeness statement,
in the sense of an open set of initial data,
see Section \ref{sec:discussion} for comments about this.

\begin{remark}[Anisotropy assumption] \label{remark:anisotropic}
Some of our results are only for fully anisotropic elements,
see Definition \ref{def:anisotropic}.
No attempt is made to get rid of this assumption
because our interest is vacuum where the BKL heuristics are inherently anisotropic,
and the filtration that we use is designed for this.
The distinction is relevant because we use a conformally orthonormal frame,
and only the anisotropic case picks a distinguished frame.
To make this more concrete,
every real symmetric $n \times n$ matrix is diagonalizable
meaning the map
$\Orth(n) \times \R^n \to \Sym_n$,
$(R,(\lambda_1,\ldots,\lambda_n)) \mapsto R \diag(\lambda_1,\ldots,\lambda_n) R^T$
is surjective, but its Jacobian is invertible
iff the $\lambda_i$ are pairwise different.
We expect that the methods can be extended to the isotropic case,
with technical modifications.
\end{remark}

We expect that the tools used to formulate and prove this theorem,
in particular the filtration from \cite{rtfil},
are suitable also for the oscillatory vacuum case,
see \cite{rtfil,rtfil2} for results at the formal level for a single bounce.

\begin{remark}
One could ask if the formal series solutions we construct
($\MC(\P)$ elements)
actually converge for real analytic
$\mcsp$ elements, which would be an alternative to \cite{ar} for real analytic solutions.
We have not investigated this.
\end{remark}


\subsection{Review of some BKL heuristics}\label{subsec:heuristics}
Here we use the proper time coordinate $t$ 
and we say in due course how this relates
to the BKL time coordinate $\tau$.
This review is organized around Table \ref{table:22}.

\begin{table}[h!]
\renewcommand{\arraystretch}{1.3}
\centering
\begin{tabular}{r||l|l}
& equation \eqref{eq:esf}, with $\phi\neq0$ & equation \eqref{eq:vc}, vacuum\\
\hline\hline
spatially homogeneous & non-oscillatory & oscillatory \\
& discussed below & \cite{hombkl,Ri1,Ri2,Be,LHWG,Br} \\
\hline
no symmetry & non-oscillatory & \cellcolor[gray]{0.9} oscillatory\\
& \cite{ar,rs,rs2} and this paper & \cellcolor[gray]{0.9} wide open
(but see \cite{rtfil,rtfil2})
\end{tabular}
\caption{The BKL heuristics 
and a selection of rigorous work.
Spatially homogeneous refers to a metric
$-(dt)^2 + h_t$ 
where $h_t$ is a $t$-dependent
left invariant metric on a semisimple 3-dimensional Lie group;
the Einstein equations reduce to ordinary differential equations in this case.
Non-oscillatory means that singularities are Kasner-like
whereas oscillatory means that singularities
are, according to BKL heuristics,
only intermittently Kasner-like.
}\label{table:22}
\end{table}

Begin in the upper left entry of Table \ref{table:22}.
For concreteness consider metrics on
$(t_-,t_+) \times S^3$ where the 3-sphere is viewed as a Lie group.
The simplest example are isotropic metrics,
also known as Friedman-Lemaitre-Robertson-Walker metrics,
\begin{equation}\label{eq:flrw}
-(d t)^2 + a(t)^2 g_{S^3}
\end{equation}
with past and future curvature singularities at $t_{\pm}$.
The solution is in Figure \ref{fig:flrw}
and satisfies
$a(t) \sim |t-t_\pm|^\normp$ as $t \to t_{\pm}$
with $\normp = \tfrac13$.
To understand the causal structure near $t_+$
it suffices to look at the conformal class of the metric.
In this particular case one can compactify using a new time coordinate
$s = |t-t_+|^{2/3}$.
The conformal metric extends smoothly to $s=0$
hence the boundary at $t_+$ is a 3-sphere whose
points are causally disconnected, one says that there are particle
horizons\footnote{%
A more precise discussion
uses the concept of a `past indecomposable set' in the literature.
}.
\begin{figure}[h!]
  \centering
  \begin{tikzpicture}
    \pgftransformyscale{0.7}
    \draw [ultra thick] plot [smooth] coordinates {(-2.39628,0.0190578) (-2.39628,0.0232772) (-2.39628,0.0284309) (-2.39627,0.0347256) (-2.39627,0.0424139) (-2.39626,0.0518044) (-2.39624,0.0632741) (-2.3962,0.0772831) (-2.39614,0.0943938) (-2.39603,0.115293) (-2.39582,0.140819) (-2.39543,0.171996) (-2.39474,0.210074) (-2.39347,0.25658) (-2.39115,0.313375) (-2.38693,0.382722) (-2.37926,0.467361) (-2.36528,0.570576) (-2.33986,0.696196) (-2.29378,0.848422) (-2.2107,1.03112) (-2.06285,1.24576) (-1.80682,1.48632) (-1.38763,1.72939) (-0.769577,1.92355) (0.,2.) (0.769577,1.92355) (1.38763,1.72939) (1.80682,1.48632) (2.06285,1.24576) (2.2107,1.03112) (2.29378,0.848422) (2.33986,0.696196) (2.36528,0.570576) (2.37926,0.467361) (2.38693,0.382722) (2.39115,0.313375) (2.39347,0.25658) (2.39474,0.210074) (2.39543,0.171996) (2.39582,0.140819) (2.39603,0.115293) (2.39614,0.0943938) (2.3962,0.0772831) (2.39624,0.0632741) (2.39626,0.0518044) (2.39627,0.0424139) (2.39627,0.0347256) (2.39628,0.0284309) (2.39628,0.0232772) (2.39628,0.0190578)}; 
    \draw (-2.4,-0.3) node {$t_-$};
    \draw (2.4,-0.3) node {$t_+$};
    \draw [->] (0,0) -- (0,2.5) node [anchor = west, xshift=2] {$a(t)$};
    \draw [->] (-2.9,0) -- (2.9,0) node [anchor = west, xshift=2] {$t$};
  \end{tikzpicture}
\caption{%
The metric \eqref{eq:flrw}
solves the Einstein equations \eqref{eq:esf}
when, up to translations and scaling,
$a(\int_0^\tau \cosh(s)^{-3/2} \dd s) = \cosh(\tau)^{-1/2}$.
This function satisfies $aa'' + 2(a')^2 = -\text{const}$.
}\label{fig:flrw}
\end{figure}
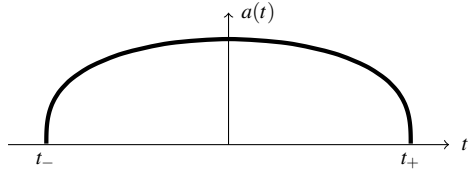

Still in the upper left entry of 
Table \ref{table:22}, consider now anisotropic metrics
\begin{equation}\label{eq:anisotropic}
- (dt)^2 + \textstyle\sum_{i=1}^3 a_i(t)^2 \omega_i \otimes \omega_i
\end{equation}
where $\omega_i$ is a standard 
frame of left-invariant one-forms on $S^3$,
so $\dd \omega_i = \omega_{i+1} \wedge \omega_{i+2}$ with $i \in \Z/3\Z$.
The Einstein equations \eqref{eq:esf} are
ordinary differential equations for
$a_1, a_2, a_3$ and their derivatives\footnote{%
There is no constraint equation but a constraint inequality.
Note that the effect of the scalar field in \eqref{eq:esf},
compared to vacuum \eqref{eq:vc},
is that the Ricci curvature
component $\Ric_g(\p_t,\p_t)$ needs only be non negative.
}.
A typical solution is in Figure \ref{fig:anisotropic}
with some oscillations in the bulk
-- note that we do not refer to such solutions as oscillatory
   because there are only a finite number of oscillations --
but simple limiting behavior
\begin{equation}\label{eq:lim}
a_i(t) \sim |t-t_\pm|^{\normp_i^\pm}
\end{equation}
as $t \to t_{\pm}$ for some $0 < \normp_i^\pm < 1$
with $\normp_1^\pm + \normp_2^\pm + \normp_3^\pm = 1$.
The analysis in this paper includes such solutions as a special case,
but only intervals near $t_{+}$, not the oscillations in the bulk
($t_-$ is analogous).
Near $t_{+}$ and near a point of $S^3$,
they
are well approximated, in local coordinates,
by the Kasner metric, have particle horizons and a curvature singularity.
\begin{figure}[h!]
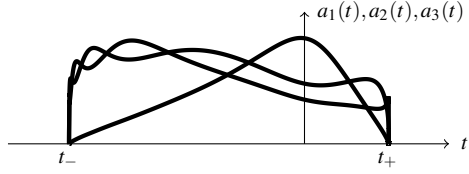

  \centering

\caption{%
A typical anisotropic spatially homogeneous
solution \eqref{eq:anisotropic} of \eqref{eq:esf} in proper time $t$,
with finitely many oscillations (this is clearer in Figure \ref{fig:anisotropic2})
but Kasner-like behavior near $t_{\pm}$.
The image was generated using a numerical solver for ordinary differential equations.
It would not be difficult to prove the behavior seen in such a numerical solution.
}\label{fig:anisotropic}
\end{figure}
\begin{remark}[Kasner metric] \label{remark:kasner1}
The Kasner metric on $(-\infty,t_+) \times \R^3$
(or $\times \mathbbm{T}^3$) is
\[
-(\dd t)^2 + \textstyle\sum_{i=1}^3 |t-t_+|^{2 \normp_i} (\dd x^i)^2
\]
with $\normp_i \in \R$ and $\normp_1 + \normp_2 + \normp_3 = 1$.
It solves \eqref{eq:esf} for some real scalar field $\phi$
if and only if $\normp_2 \normp_3 + \normp_3 \normp_1 + \normp_1 \normp_2 \geq 0$,
and \eqref{eq:vc} if and only if equality holds.
The scalar field crucially allows having $\normp_1,\normp_2,\normp_3 > 0$.
We are interested in the curvature singularity at $t \uparrow t_+$\footnote{%
The Kretschmann curvature invariant is
a constant times $|t-t_+|^{-4}$.
There is a transitive action of $\R \ltimes \R^3$
as a group of global conformal isometries,
so the Kretschmann scalar and all other curvature scalars
are determined up to an overall constant by this action.
}.
If $\normp_1,\normp_2,\normp_3 < 1$ then the boundary at $t_+$
is an $\R^3$ (or $\mathbbm{T}^3$) whose points are causally disconnected.
The same solution is given, using another time coordinate, by
\[
  -e^{-2(p_1+p_2+p_3)\tau} (\dd \tau)^2
   + \textstyle\sum_{i=1}^3 e^{-2p_i\tau} (\dd x^i)^2
\]
where $\normp_i = p_i/(p_1+p_2+p_3)$.
The Kasner spacetime can alternatively be presented as a Maurer-Cartan element in
the graded Lie algebra $\E_\Phi$, see Remark \ref{remark:kasner2}.
\end{remark}

It is very useful to change coordinates to push $t = t_{\pm}$
to points at infinity, $\tau = \pm \infty$.
Among the infinitude of choices, a distinguished choice is
\[
dt/d\tau = a_1a_2a_3
\]
The superficial reason is that instead of Figure \ref{fig:anisotropic}
we get Figure \ref{fig:anisotropic2}. 
More importantly, 
as BKL noticed,
various leading terms become hyperbolic trigonometric functions
of the time coordinate $\tau$. 
With reference to Figure \ref{fig:anisotropic2},
the large positive $\tau$ tail is given by a converging power series
\begin{equation}\label{eq:expa}
a_1, a_2, a_3 \;\in \; \R[[e^{-p_1 \tau}, e^{-p_2 \tau}, e^{-p_3 \tau}]]
\end{equation}
We sometimes refer to $\tau$ as BKL time.
The time function $\tau$ used in this paper reduces to this one
in the spatially homogeneous case.

\begin{figure}[h!]
  \centering
  \begin{tikzpicture}
     \pgftransformyscale{0.7}
     \draw [ultra thick] plot [smooth] coordinates {(-5.,0.126476) (-4.97,0.128249) (-4.94,0.130048) (-4.91,0.131873) (-4.88,0.133724) (-4.85,0.135602) (-4.82,0.137507) (-4.79,0.13944) (-4.76,0.141401) (-4.73,0.14339) (-4.7,0.14541) (-4.67,0.147458) (-4.64,0.149538) (-4.61,0.151648) (-4.58,0.153789) (-4.55,0.155963) (-4.52,0.15817) (-4.49,0.16041) (-4.46,0.162684) (-4.43,0.164994) (-4.4,0.167339) (-4.37,0.169721) (-4.34,0.17214) (-4.31,0.174597) (-4.28,0.177094) (-4.25,0.179631) (-4.22,0.18221) (-4.19,0.18483) (-4.16,0.187495) (-4.13,0.190204) (-4.1,0.19296) (-4.07,0.195763) (-4.04,0.198615) (-4.01,0.201517) (-3.98,0.204472) (-3.95,0.207481) (-3.92,0.210546) (-3.89,0.213669) (-3.86,0.216852) (-3.83,0.220098) (-3.8,0.223409) (-3.77,0.226787) (-3.74,0.230235) (-3.71,0.233757) (-3.68,0.237356) (-3.65,0.241035) (-3.62,0.244798) (-3.59,0.248648) (-3.56,0.252591) (-3.53,0.25663) (-3.5,0.260772) (-3.47,0.26502) (-3.44,0.269381) (-3.41,0.273861) (-3.38,0.278466) (-3.35,0.283204) (-3.32,0.288082) (-3.29,0.293108) (-3.26,0.29829) (-3.23,0.303639) (-3.2,0.309164) (-3.17,0.314876) (-3.14,0.320786) (-3.11,0.326906) (-3.08,0.33325) (-3.05,0.339831) (-3.02,0.346664) (-2.99,0.353765) (-2.96,0.361151) (-2.93,0.368839) (-2.9,0.376849) (-2.87,0.3852) (-2.84,0.393914) (-2.81,0.403012) (-2.78,0.41252) (-2.75,0.422461) (-2.72,0.43286) (-2.69,0.443747) (-2.66,0.455148) (-2.63,0.467092) (-2.6,0.479612) (-2.57,0.492737) (-2.54,0.506499) (-2.51,0.520933) (-2.48,0.53607) (-2.45,0.551945) (-2.42,0.568589) (-2.39,0.586036) (-2.36,0.604316) (-2.33,0.623459) (-2.3,0.64349) (-2.27,0.664433) (-2.24,0.686304) (-2.21,0.709116) (-2.18,0.732871) (-2.15,0.757563) (-2.12,0.78317) (-2.09,0.80966) (-2.06,0.836979) (-2.03,0.865051) (-2.,0.893778) (-1.97,0.923031) (-1.94,0.952647) (-1.91,0.982429) (-1.88,1.01214) (-1.85,1.0415) (-1.82,1.07021) (-1.79,1.09789) (-1.76,1.12418) (-1.73,1.14866) (-1.7,1.17094) (-1.67,1.19061) (-1.64,1.20731) (-1.61,1.22074) (-1.58,1.23068) (-1.55,1.23699) (-1.52,1.23969) (-1.49,1.2389) (-1.46,1.23487) (-1.43,1.22803) (-1.4,1.21888) (-1.37,1.20807) (-1.34,1.19636) (-1.31,1.18458) (-1.28,1.17365) (-1.25,1.16457) (-1.22,1.15838) (-1.19,1.15617) (-1.16,1.15906) (-1.13,1.16812) (-1.1,1.18438) (-1.07,1.20873) (-1.04,1.24183) (-1.01,1.28403) (-0.98,1.33523) (-0.95,1.39473) (-0.92,1.4611) (-0.89,1.53198) (-0.86,1.60398) (-0.83,1.6727) (-0.8,1.73299) (-0.77,1.77958) (-0.74,1.80808) (-0.71,1.81604) (-0.68,1.80382) (-0.65,1.7747) (-0.62,1.73431) (-0.59,1.68967) (-0.56,1.64806) (-0.53,1.61608) (-0.5,1.59879) (-0.47,1.59908) (-0.44,1.61703) (-0.41,1.64962) (-0.38,1.69082) (-0.35,1.73224) (-0.32,1.76433) (-0.29,1.77813) (-0.26,1.76716) (-0.23,1.72895) (-0.2,1.66546) (-0.17,1.58246) (-0.14,1.48814) (-0.11,1.3916) (-0.08,1.30158) (-0.05,1.22528) (-0.02,1.16747) (0.01,1.12987) (0.04,1.11127) (0.07,1.10834) (0.1,1.11677) (0.13,1.13232) (0.16,1.15135) (0.19,1.17107) (0.22,1.18938) (0.25,1.20484) (0.28,1.21646) (0.31,1.22358) (0.34,1.22587) (0.37,1.2232) (0.4,1.21563) (0.43,1.20341) (0.46,1.18689) (0.49,1.16652) (0.52,1.14283) (0.55,1.11635) (0.58,1.08763) (0.61,1.0572) (0.64,1.02557) (0.67,0.993154) (0.7,0.960362) (0.73,0.927524) (0.76,0.894925) (0.79,0.862794) (0.82,0.83132) (0.85,0.800648) (0.88,0.77089) (0.91,0.742129) (0.94,0.714422) (0.97,0.687805) (1.,0.662299) (1.03,0.637912) (1.06,0.614637) (1.09,0.592464) (1.12,0.571373) (1.15,0.551339) (1.18,0.532333) (1.21,0.514325) (1.24,0.497282) (1.27,0.481168) (1.3,0.465947) (1.33,0.451584) (1.36,0.438041) (1.39,0.425284) (1.42,0.413275) (1.45,0.401978) (1.48,0.39136) (1.51,0.381385) (1.54,0.37202) (1.57,0.363232) (1.6,0.35499) (1.63,0.347262) (1.66,0.340019) (1.69,0.333232) (1.72,0.326874) (1.75,0.320917) (1.78,0.315337) (1.81,0.310109) (1.84,0.30521) (1.87,0.300617) (1.9,0.296311) (1.93,0.29227) (1.96,0.288477) (1.99,0.284914) (2.02,0.281564) (2.05,0.278411) (2.08,0.27544) (2.11,0.272639) (2.14,0.269994) (2.17,0.267493) (2.2,0.265125) (2.23,0.26288) (2.26,0.260749) (2.29,0.258721) (2.32,0.25679) (2.35,0.254946) (2.38,0.253185) (2.41,0.251497) (2.44,0.249879) (2.47,0.248324) (2.5,0.246827) (2.53,0.245384) (2.56,0.243989) (2.59,0.24264) (2.62,0.241333) (2.65,0.240063) (2.68,0.238829) (2.71,0.237627) (2.74,0.236455) (2.77,0.23531) (2.8,0.234191) (2.83,0.233094) (2.86,0.23202) (2.89,0.230965) (2.92,0.229928) (2.95,0.228909) (2.98,0.227905) (3.01,0.226915) (3.04,0.225939) (3.07,0.224976) (3.1,0.224024) (3.13,0.223083) (3.16,0.222152) (3.19,0.221231) (3.22,0.220318) (3.25,0.219414) (3.28,0.218517) (3.31,0.217628) (3.34,0.216746) (3.37,0.21587) (3.4,0.215) (3.43,0.214136) (3.46,0.213278) (3.49,0.212425) (3.52,0.211576) (3.55,0.210733) (3.58,0.209895) (3.61,0.20906) (3.64,0.20823) (3.67,0.207404) (3.7,0.206583) (3.73,0.205765) (3.76,0.20495) (3.79,0.20414) (3.82,0.203332) (3.85,0.202529) (3.88,0.201728) (3.91,0.200931) (3.94,0.200138) (3.97,0.199347) (4.,0.198559) (4.03,0.197775) (4.06,0.196993) (4.09,0.196215) (4.12,0.195439) (4.15,0.194667) (4.18,0.193897) (4.21,0.19313) (4.24,0.192366) (4.27,0.191605) (4.3,0.190846) (4.33,0.19009) (4.36,0.189337) (4.39,0.188587) (4.42,0.187839) (4.45,0.187094) (4.48,0.186352) (4.51,0.185613) (4.54,0.184876) (4.57,0.184141) (4.6,0.18341) (4.63,0.18268) (4.66,0.181954) (4.69,0.18123) (4.72,0.180509) (4.75,0.17979) (4.78,0.179074) (4.81,0.17836) (4.84,0.177649) (4.87,0.176941) (4.9,0.176235) (4.93,0.175532) (4.96,0.174831) (4.99,0.174133)};
     \draw [ultra thick] plot [smooth] coordinates {(-5.,0.000036748) (-4.97,0.0000387306) (-4.94,0.0000408207) (-4.91,0.0000430227) (-4.88,0.000045344) (-4.85,0.0000477917) (-4.82,0.0000503705) (-4.79,0.0000530874) (-4.76,0.0000559524) (-4.73,0.0000589731) (-4.7,0.0000621554) (-4.67,0.0000655106) (-4.64,0.0000690489) (-4.61,0.0000727772) (-4.58,0.000076706) (-4.55,0.0000808495) (-4.52,0.0000852179) (-4.49,0.00008982) (-4.46,0.0000946725) (-4.43,0.0000997906) (-4.4,0.000105185) (-4.37,0.000110871) (-4.34,0.000116868) (-4.31,0.000123189) (-4.28,0.000129855) (-4.25,0.000136884) (-4.22,0.000144294) (-4.19,0.000152109) (-4.16,0.000160351) (-4.13,0.000169041) (-4.1,0.000178208) (-4.07,0.000187876) (-4.04,0.000198071) (-4.01,0.000208829) (-3.98,0.000220175) (-3.95,0.000232146) (-3.92,0.000244778) (-3.89,0.000258102) (-3.86,0.000272166) (-3.83,0.000287009) (-3.8,0.00030267) (-3.77,0.000319207) (-3.74,0.000336664) (-3.71,0.000355094) (-3.68,0.000374557) (-3.65,0.000395109) (-3.62,0.000416822) (-3.59,0.000439756) (-3.56,0.000463992) (-3.53,0.000489601) (-3.5,0.000516674) (-3.47,0.000545292) (-3.44,0.000575557) (-3.41,0.000607565) (-3.38,0.000641428) (-3.35,0.00067726) (-3.32,0.000715185) (-3.29,0.00075534) (-3.26,0.00079786) (-3.23,0.000842908) (-3.2,0.000890636) (-3.17,0.000941231) (-3.14,0.00099487) (-3.11,0.00105177) (-3.08,0.00111213) (-3.05,0.00117619) (-3.02,0.00124421) (-2.99,0.00131645) (-2.96,0.00139319) (-2.93,0.00147476) (-2.9,0.00156147) (-2.87,0.00165369) (-2.84,0.0017518) (-2.81,0.0018562) (-2.78,0.00196733) (-2.75,0.00208566) (-2.72,0.0022117) (-2.69,0.00234596) (-2.66,0.00248901) (-2.63,0.00264147) (-2.6,0.00280397) (-2.57,0.00297718) (-2.54,0.00316185) (-2.51,0.00335872) (-2.48,0.0035686) (-2.45,0.00379234) (-2.42,0.00403083) (-2.39,0.00428503) (-2.36,0.00455591) (-2.33,0.00484451) (-2.3,0.00515193) (-2.27,0.00547931) (-2.24,0.00582787) (-2.21,0.00619888) (-2.18,0.00659369) (-2.15,0.00701374) (-2.12,0.0074606) (-2.09,0.00793594) (-2.06,0.00844158) (-2.03,0.00897955) (-2.,0.0095521) (-1.97,0.0101618) (-1.94,0.0108114) (-1.91,0.0115044) (-1.88,0.0122444) (-1.85,0.0130359) (-1.82,0.0138838) (-1.79,0.0147939) (-1.76,0.0157727) (-1.73,0.0168277) (-1.7,0.0179671) (-1.67,0.0192001) (-1.64,0.0205366) (-1.61,0.0219872) (-1.58,0.0235628) (-1.55,0.0252748) (-1.52,0.0271343) (-1.49,0.0291525) (-1.46,0.0313401) (-1.43,0.0337073) (-1.4,0.0362642) (-1.37,0.0390214) (-1.34,0.0419906) (-1.31,0.0451861) (-1.28,0.0486265) (-1.25,0.0523363) (-1.22,0.0563486) (-1.19,0.0607061) (-1.16,0.0654625) (-1.13,0.0706812) (-1.1,0.0764333) (-1.07,0.0827927) (-1.04,0.0898307) (-1.01,0.0976098) (-0.98,0.106181) (-0.95,0.115583) (-0.92,0.125855) (-0.89,0.137047) (-0.86,0.149245) (-0.83,0.162591) (-0.8,0.177301) (-0.77,0.193657) (-0.74,0.211986) (-0.71,0.232606) (-0.68,0.25577) (-0.65,0.281646) (-0.62,0.31034) (-0.59,0.341995) (-0.56,0.376924) (-0.53,0.415731) (-0.5,0.459323) (-0.47,0.508777) (-0.44,0.56509) (-0.41,0.628967) (-0.38,0.700779) (-0.35,0.780766) (-0.32,0.869341) (-0.29,0.967294) (-0.26,1.07568) (-0.23,1.19535) (-0.2,1.32621) (-0.17,1.46629) (-0.14,1.61069) (-0.11,1.75063) (-0.08,1.87308) (-0.05,1.9622) (-0.02,2.00341) (0.01,1.98892) (0.04,1.92139) (0.07,1.81258) (0.1,1.67822) (0.13,1.5329) (0.16,1.38749) (0.19,1.24876) (0.22,1.12033) (0.25,1.00365) (0.28,0.898885) (0.31,0.805494) (0.34,0.722582) (0.37,0.649119) (0.4,0.584057) (0.43,0.526396) (0.46,0.475216) (0.49,0.429693) (0.52,0.3891) (0.55,0.352804) (0.58,0.320258) (0.61,0.290997) (0.64,0.26462) (0.67,0.240788) (0.7,0.21921) (0.73,0.19964) (0.76,0.181864) (0.79,0.1657) (0.82,0.15099) (0.85,0.137593) (0.88,0.125388) (0.91,0.114267) (0.94,0.104132) (0.97,0.094896) (1.,0.0864813) (1.03,0.0788161) (1.06,0.0718357) (1.09,0.0654805) (1.12,0.0596961) (1.15,0.0544329) (1.18,0.049645) (1.21,0.0452907) (1.24,0.0413316) (1.27,0.0377322) (1.3,0.0344602) (1.33,0.0314861) (1.36,0.0287824) (1.39,0.0263248) (1.42,0.0240904) (1.45,0.0220583) (1.48,0.02021) (1.51,0.0185283) (1.54,0.0169973) (1.57,0.0156031) (1.6,0.0143328) (1.63,0.0131747) (1.66,0.0121183) (1.69,0.0111539) (1.72,0.0102731) (1.75,0.0094679) (1.78,0.00873132) (1.81,0.00805698) (1.84,0.00743914) (1.87,0.00687264) (1.9,0.00635279) (1.93,0.00587536) (1.96,0.00543653) (1.99,0.0050329) (2.02,0.00466136) (2.05,0.00431909) (2.08,0.00400356) (2.11,0.0037125) (2.14,0.00344382) (2.17,0.00319563) (2.2,0.00296624) (2.23,0.00275412) (2.26,0.00255784) (2.29,0.00237611) (2.32,0.00220783) (2.35,0.00205187) (2.38,0.00190732) (2.41,0.00177324) (2.44,0.00164887) (2.47,0.00153344) (2.5,0.00142629) (2.53,0.00132679) (2.56,0.00123436) (2.59,0.0011485) (2.62,0.00106871) (2.65,0.000994548) (2.68,0.000925593) (2.71,0.00086149) (2.74,0.000801862) (2.77,0.000746425) (2.8,0.000694836) (2.83,0.000646853) (2.86,0.000602204) (2.89,0.000560655) (2.92,0.000521998) (2.95,0.000486011) (2.98,0.000452529) (3.01,0.000421349) (3.04,0.000392334) (3.07,0.00036532) (3.1,0.000340175) (3.13,0.000316767) (3.16,0.000294969) (3.19,0.00027468) (3.22,0.000255778) (3.25,0.000238186) (3.28,0.000221801) (3.31,0.000206548) (3.34,0.000192356) (3.37,0.000179129) (3.4,0.000166815) (3.43,0.000155352) (3.46,0.000144662) (3.49,0.000134707) (3.52,0.000125445) (3.55,0.00011682) (3.58,0.000108798) (3.61,0.00010133) (3.64,0.0000943613) (3.67,0.0000878704) (3.7,0.0000818244) (3.73,0.0000761852) (3.76,0.0000709375) (3.79,0.0000660629) (3.82,0.0000615403) (3.85,0.0000573387) (3.88,0.0000534085) (3.91,0.0000497169) (3.94,0.0000462683) (3.97,0.000043064) (4.,0.0000400735) (4.03,0.0000372994) (4.06,0.0000347548) (4.09,0.0000324512) (4.12,0.0000303559) (4.15,0.0000283952) (4.18,0.0000264749) (4.21,0.0000245666) (4.24,0.0000227234) (4.27,0.0000210195) (4.3,0.0000195018) (4.33,0.0000181602) (4.36,0.0000169928) (4.39,0.0000160073) (4.42,0.0000152085) (4.45,0.0000144429) (4.48,0.0000135037) (4.51,0.0000122024) (4.54,0.0000106844) (4.57,0) (4.6,0) (4.63,0) (4.66,0) (4.69,0) (4.72,0) (4.75,0) (4.78,0) (4.81,0) (4.84,0) (4.87,0) (4.9,0) (4.93,0) (4.96,0) (4.99,0)};
     \draw [ultra thick] plot [smooth] coordinates {(-5.,0.228274) (-4.97,0.232418) (-4.94,0.236637) (-4.91,0.24093) (-4.88,0.2453) (-4.85,0.249748) (-4.82,0.254274) (-4.79,0.258881) (-4.76,0.26357) (-4.73,0.268341) (-4.7,0.273196) (-4.67,0.278137) (-4.64,0.283164) (-4.61,0.288279) (-4.58,0.293483) (-4.55,0.298777) (-4.52,0.304163) (-4.49,0.309641) (-4.46,0.315213) (-4.43,0.32088) (-4.4,0.326643) (-4.37,0.332504) (-4.34,0.338462) (-4.31,0.34452) (-4.28,0.350678) (-4.25,0.356937) (-4.22,0.363297) (-4.19,0.36976) (-4.16,0.376326) (-4.13,0.382996) (-4.1,0.38977) (-4.07,0.396648) (-4.04,0.403631) (-4.01,0.410718) (-3.98,0.417909) (-3.95,0.425204) (-3.92,0.432603) (-3.89,0.440104) (-3.86,0.447706) (-3.83,0.455409) (-3.8,0.46321) (-3.77,0.471108) (-3.74,0.479101) (-3.71,0.487185) (-3.68,0.495358) (-3.65,0.503617) (-3.62,0.511957) (-3.59,0.520374) (-3.56,0.528863) (-3.53,0.537419) (-3.5,0.546036) (-3.47,0.554707) (-3.44,0.563424) (-3.41,0.57218) (-3.38,0.580966) (-3.35,0.589772) (-3.32,0.598588) (-3.29,0.607403) (-3.26,0.616205) (-3.23,0.624981) (-3.2,0.633718) (-3.17,0.642401) (-3.14,0.651015) (-3.11,0.659543) (-3.08,0.66797) (-3.05,0.676276) (-3.02,0.684445) (-2.99,0.692457) (-2.96,0.700293) (-2.93,0.707932) (-2.9,0.715355) (-2.87,0.722541) (-2.84,0.729471) (-2.81,0.736124) (-2.78,0.74248) (-2.75,0.748521) (-2.72,0.754228) (-2.69,0.759584) (-2.66,0.764573) (-2.63,0.769181) (-2.6,0.773396) (-2.57,0.777208) (-2.54,0.780608) (-2.51,0.783593) (-2.48,0.78616) (-2.45,0.78831) (-2.42,0.79005) (-2.39,0.791387) (-2.36,0.792336) (-2.33,0.792914) (-2.3,0.793143) (-2.27,0.793054) (-2.24,0.792678) (-2.21,0.792056) (-2.18,0.791235) (-2.15,0.790268) (-2.12,0.789217) (-2.09,0.788151) (-2.06,0.78715) (-2.03,0.786302) (-2.,0.785707) (-1.97,0.785476) (-1.94,0.785735) (-1.91,0.78662) (-1.88,0.788286) (-1.85,0.790899) (-1.82,0.794643) (-1.79,0.799718) (-1.76,0.806339) (-1.73,0.814733) (-1.7,0.825141) (-1.67,0.837812) (-1.64,0.852999) (-1.61,0.870953) (-1.58,0.891915) (-1.55,0.916109) (-1.52,0.943726) (-1.49,0.974913) (-1.46,1.00976) (-1.43,1.04825) (-1.4,1.09029) (-1.37,1.13562) (-1.34,1.18382) (-1.31,1.23422) (-1.28,1.28593) (-1.25,1.33777) (-1.22,1.38826) (-1.19,1.43568) (-1.16,1.47814) (-1.13,1.51374) (-1.1,1.54076) (-1.07,1.55794) (-1.04,1.56469) (-1.01,1.56127) (-0.98,1.54881) (-0.95,1.52929) (-0.92,1.50535) (-0.89,1.4801) (-0.86,1.45693) (-0.83,1.4393) (-0.8,1.43055) (-0.77,1.43372) (-0.74,1.45126) (-0.71,1.4847) (-0.68,1.5343) (-0.65,1.59864) (-0.62,1.67424) (-0.59,1.75529) (-0.56,1.83364) (-0.53,1.89951) (-0.5,1.94306) (-0.47,1.95685) (-0.44,1.93808) (-0.41,1.88947) (-0.38,1.81805) (-0.35,1.73267) (-0.32,1.64142) (-0.29,1.55004) (-0.26,1.46144) (-0.23,1.37604) (-0.2,1.29293) (-0.17,1.21105) (-0.14,1.13036) (-0.11,1.05239) (-0.08,0.979994) (-0.05,0.916499) (-0.02,0.86447) (0.01,0.824763) (0.04,0.796288) (0.07,0.776551) (0.1,0.762599) (0.13,0.751818) (0.16,0.742346) (0.19,0.733117) (0.22,0.723707) (0.25,0.714117) (0.28,0.704582) (0.31,0.695433) (0.34,0.687) (0.37,0.679567) (0.4,0.67335) (0.43,0.668489) (0.46,0.66506) (0.49,0.663078) (0.52,0.662517) (0.55,0.663313) (0.58,0.665383) (0.61,0.668626) (0.64,0.672935) (0.67,0.678196) (0.7,0.6843) (0.73,0.691137) (0.76,0.698602) (0.79,0.706595) (0.82,0.715021) (0.85,0.723789) (0.88,0.732815) (0.91,0.742015) (0.94,0.751313) (0.97,0.760633) (1.,0.769903) (1.03,0.779055) (1.06,0.788022) (1.09,0.796738) (1.12,0.805141) (1.15,0.81317) (1.18,0.820767) (1.21,0.827877) (1.24,0.834446) (1.27,0.840425) (1.3,0.845769) (1.33,0.850434) (1.36,0.854384) (1.39,0.857586) (1.42,0.860013) (1.45,0.861643) (1.48,0.862461) (1.51,0.862457) (1.54,0.861627) (1.57,0.859974) (1.6,0.857507) (1.63,0.85424) (1.66,0.850192) (1.69,0.845389) (1.72,0.839859) (1.75,0.833636) (1.78,0.826756) (1.81,0.819258) (1.84,0.811184) (1.87,0.802576) (1.9,0.793478) (1.93,0.783935) (1.96,0.77399) (1.99,0.763687) (2.02,0.75307) (2.05,0.742178) (2.08,0.731053) (2.11,0.719732) (2.14,0.70825) (2.17,0.696643) (2.2,0.684942) (2.23,0.673177) (2.26,0.661375) (2.29,0.649561) (2.32,0.637759) (2.35,0.62599) (2.38,0.614273) (2.41,0.602625) (2.44,0.591062) (2.47,0.579597) (2.5,0.568243) (2.53,0.557011) (2.56,0.54591) (2.59,0.534947) (2.62,0.524131) (2.65,0.513467) (2.68,0.50296) (2.71,0.492614) (2.74,0.482433) (2.77,0.472418) (2.8,0.462572) (2.83,0.452897) (2.86,0.443393) (2.89,0.43406) (2.92,0.424899) (2.95,0.415909) (2.98,0.407089) (3.01,0.398438) (3.04,0.389955) (3.07,0.381639) (3.1,0.373488) (3.13,0.365499) (3.16,0.357672) (3.19,0.350003) (3.22,0.342491) (3.25,0.335134) (3.28,0.327928) (3.31,0.320872) (3.34,0.313963) (3.37,0.307198) (3.4,0.300576) (3.43,0.294094) (3.46,0.287748) (3.49,0.281537) (3.52,0.275457) (3.55,0.269508) (3.58,0.263685) (3.61,0.257986) (3.64,0.25241) (3.67,0.246953) (3.7,0.241613) (3.73,0.236388) (3.76,0.231276) (3.79,0.226274) (3.82,0.221379) (3.85,0.21659) (3.88,0.211905) (3.91,0.207321) (3.94,0.202835) (3.97,0.198447) (4.,0.194154) (4.03,0.189954) (4.06,0.185844) (4.09,0.181824) (4.12,0.177891) (4.15,0.174043) (4.18,0.170278) (4.21,0.166595) (4.24,0.162992) (4.27,0.159467) (4.3,0.156018) (4.33,0.152644) (4.36,0.149344) (4.39,0.146115) (4.42,0.142955) (4.45,0.139865) (4.48,0.136841) (4.51,0.133883) (4.54,0.130989) (4.57,0.128158) (4.6,0.125389) (4.63,0.122679) (4.66,0.120028) (4.69,0.117435) (4.72,0.114898) (4.75,0.112415) (4.78,0.109987) (4.81,0.107611) (4.84,0.105287) (4.87,0.103013) (4.9,0.100788) (4.93,0.0986118) (4.96,0.0964825) (4.99,0.0943993)};
    \draw [->] (0,0) -- (0,2.5) node [anchor = west, xshift=2] {$a_1$, $a_2$, $a_3$};
    \draw [->] (-5.5,0) -- (5.5,0) node [anchor = west, xshift=2] {$\tau$};
  \end{tikzpicture}
\caption{%
The solution in Figure \ref{fig:anisotropic}
relative to BKL time $\tau$.
}\label{fig:anisotropic2}
\end{figure}
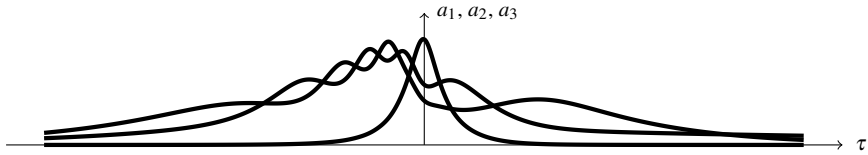

The upper right entry in Table \ref{table:22}
refers to the vacuum equations \eqref{eq:vc}.
There are no isotropic solutions \eqref{eq:flrw} on an interval times $S^3$.
Consider instead anisotropic metrics \eqref{eq:anisotropic}.
The heuristics of BKL and Misner
\cite{bkl,misner}
predict complicated oscillations that never stop;
the spacetime becomes singular at a finite time $t_{+}$
but only after infinitely many oscillations.
The Kasner metric is only intermittently a good approximation:
the only-approximately-defined parameters $\normp_1$, $\normp_2$, $\normp_3$
jump countably many times\footnote{%
Two positive and one negative.
Note that there are Kasner metrics for vacuum
with two positive and one negative $\normp_i$,
but this situation is unstable under perturbing $\R^3$ as a Lie group into $S^3$.
The instability leads to a transition (called BKL bounce)
from one Kasner-like interval to the next.}.
A substantial portion of these heuristics were made rigorous in \cite{hombkl}:
\begin{itemize}
\item A large class of semiglobal oscillatory solutions was constructed, up to $t_+$.
\item It was shown that these solutions have particle horizons.
\end{itemize}
A key tool in \cite{hombkl} are Poincare sections to construct a
discrete dynamical system of partially defined Poincare maps.
The Poincare map is constructed by shadowing
explicit approximate solutions.
Iterating the partially defined Poincare maps requires care
because part of this discrete dynamical system behaves like
the Gauss map for continued fractions\footnote{%
This is closely related to the appearance of
 denominators $1/(n_1p_1 + n_2p_2 + n_3p_3)$
 in formal series solutions in the present paper, with $n_i \in \Z_{>0}$.
 In the present paper we can take $p_1,p_2,p_3 > 0$ due to the scalar field,
 but for vacuum one $p_i$ is negative and
 one is faced with a small denominator problem.
 Actually there are even zero denominators at points where there
 are linear dependencies over the positive rationals
 among the $p_i$.
},
and only succeeds because there is another emerging small parameter
-- the bounces become more and more pronounced.
Other interesting rigorous work related to this case
includes \cite{lw1,lw2,Ri1,Ri2,Be,LHWG,Br}.
\step
The lower left entry of Table \ref{table:22}
is the topic of this paper
and requires solving the partial differential equations \eqref{eq:esf}.
The BKL heuristics are non-oscillatory and
qualitatively like in the upper left entry of Table \ref{table:22}.
We employ expansions similar to \eqref{eq:expa}
to implement Theorem \ref{thm:informalintro} b)
but $p_1,p_2,p_3$ are now functions of space and
the coefficients are functions of space and polynomial in $\tau$,
roughly corresponding to logarithms for proper time.
Such expansions diverge in general,
and we make energy estimates to construct actual solutions,
implementing Theorem \ref{thm:informalintro} c).
\step
The lower right entry of
Table \ref{table:22} for vacuum \eqref{eq:vc}
motivates this endeavor.
It remains to be seen if the oscillatory BKL heuristics will hold up.
If they do, the formulation of a theorem is likely to be subtle.
It seems essential that there be particle horizons
to allow a local-in-space analysis,
and it is encouraging that particle horizons appear in all cases discussed above
 including the oscillatory spatially homogeneous case \cite{hombkl}.
We speculate that a successful approach to this problem
will involve:
\begin{itemize}
\item The construction of 
a partially defined infinite-dimensional discrete dynamical system,
using shadowing.
This will require energy estimates.
\item For shadowing,
the construction of explicit leading terms for one BKL bounce.
\item
Control of the first derivative of the partially
defined discrete dynamical system,
because the inverse function theorem requires this.
\item An algebraic framework that is flexible and provides enough language
to be able to change between different gauges,
because a single gauge will not do.
\end{itemize}
Shadowing is used, in a simpler context, in \cite{hombkl}.
A leading term for one BKL bounce is in \cite{rtfil},
and this extends to formal power series solutions \cite{rtfil2}.
The MC formulation of the Einstein equations in \cite{rtgla} 
makes homological algebra language available that is very suited
to discussing gauge freedom, and quadratic equations.

\newcommand{\trap}{\textnormal{trapezoid}}

\section{Preliminaries} \label{sec:preliminaries}
\emph{Notation.}
Let $\N = \{0,1,2,\ldots\}$ be the non-negative integers;
$\End(\R^\dimy)$
the vector space of $\dimy \times \dimy$ matrices;
$\Sym_\dimy$ the vector space of real symmetric matrices.
\step
\emph{Domain.}
Sometimes we write $\dimx$ for the spatial dimension,
though the application will be to $\dimx = 3$.
Instead of $[0,\infty) \times M^\dimx$ 
we sometimes work on $[0,\infty) \times \mathbbm{T}^\dimx$
for simplicity where the notation for partial derivatives is more straightforward.
\step
\emph{Graded Lie algebra.}
A real gLa \cite{niri} is a real vector space $\gx = \bigoplus_{i \in \Z} \gx^i$
with an $\R$-bilinear bracket $[-,-] : \gx^i \times \gx^j \to \gx^{i+j}$
that satisfies
$[x,y] = -(-1)^{xy}[y,x]$ and
\begin{equation}\label{eq:jacobi}
[x,[y,z]]
+ (-1)^{x(y+z)}[y,[z,x]]
+ (-1)^{z(x+y)}[z,[x,y]] = 0
\end{equation}
\step
\emph{Maurer-Cartan elements.} The set of Maurer-Cartan elements in a gLa $\gx$ is
\[
\MC(\gx) = \{u \in \gx^1 \mid [u,u] = 0\}
\]
Associated to every MC element is a differential
$d = [u,-] \in \End^1(\gx)$, so $d^2 = 0$.
\step
\emph{Cyclic summation.} We say a triple $(i,j,k)$ is cyclic if it is in the set
\begin{equation}\label{eq:cycl}
C = \{(1,2,3), (2,3,1), (3,1,2)\}
\end{equation}


\newcommand{\plat}{\p_{\textnormal{latREPLACE}}}

\section{Two semiglobal theorems for symmetric hyperbolic systems} \label{sec:semi}
This section is logically self-contained
and not specific to general relativity.
It contains two abstract theorems
for semiglobal solutions to quasilinear
partial differential equations.
They are only applied in Section \ref{section:maintruesolution}.

\subsection{Preliminaries}
Let $\dimx \geq 1$ be an integer.
Throughout Section \ref{sec:semi} we are on
\[
\Omega = [0,\infty) \times \mathbbm{T}^\dimx
\qquad
\text{with coordinates denoted}
\qquad
\ind = (\tau,x)
\]
Standard partial derivatives are denoted
$\p_{\mu}$, where $\p_0$ and $\p_1,\dots,\p_\dimx$
are derivatives with respect to $\tau$ and $x$ respectively.
If $\alpha \in \N^{1+\dimx}$ then $\p^\alpha$ is the corresponding
higher partial derivative; this includes time derivatives.
If $\fld = \fld(\ind)$ then
$\|\fld\|_{L^p}$ is the $L^p$-norm over $\Omega$,
whereas
$\|\fld\|_{L^p_x}$ is the $L^p$-norm over $\mathbbm{T}^\dimx$
and the result still depends on $\tau$.
It is always understood that $\mu = 0\ldots \dimx$
whereas $\alpha \in \smash{\N^{1+\dimx}}$.

\newcommand{\vst}{\rule{0pt}{9pt}}
\subsection{A square system} \label{section:analysis1}

A square system is one with 
as many equations as unknowns.
The system discussed here, \eqref{pde1},
is chosen to suit the
non-square system in {\ptwo}.
\begin{theorem}[Square system on a semiglobal domain] \label{theorem:semia}
Given five real constants
\begin{equation}\label{csx1}
1 < q < Q
\qquad
0 < Qz < Z
\qquad
b > 0
\qquad 
\end{equation}
Given a spatial dimension $\dimx \geq 1$
and another integer $\dimy \geq 1$. Given
\begin{itemize}
\item[]
$A^\mu \in \Hom(\R^\dimy,\Sym_\dimy)$
and $B \in \Hom(\R^\dimy \otimes \R^\dimy, \R^\dimy)$ with $B$ symmetric\footnote{%
For all
$u,v \in \R^n$ we abbreviate $B(u,v) = B(u \otimes v) \in \R^n$.
Symmetry means $B(u,v) = B(v,u)$.}.
\end{itemize}
Then there exist $\eps > 0$ and $C > 0$\footnote{%
It is understood from the formulation of the theorem,
in particular the order in which logical quantifiers appear,
that $\eps$ and $C$ can only depend on the data already mentioned.
} such that for all
$0 < \delta \leq \eps$ and all
\begin{itemize}
\item[]
$a^\mu \in C^\infty(\Omega,\Sym_\dimy)$
and $L \in C^\infty(\Omega,\End(\R^{\dimy}))$
and $F \in C^\infty(\Omega,\R^\dimy)$
\end{itemize}
if\footnote{%
Statements \ref{acs1}, \ref{acs2}, \ref{afderyy}, \ref{aLder}, \ref{aFder}, \ref{affdec1}
require a component-wise norm. It is understood that
this must be fixed before the theorem.
The theorem is for all choices.
We are more specific in \ref{alop}.
}
\begin{enumerate}[({a}1),leftmargin=10mm]
\item\label{abdd} 
$a^\mu$, $L$, $F$
and all their derivatives of all orders are bounded on $\Omega$.
\item\label{ak2k}
$q^{-1} \one \leq a^0 \leq q \one$ on $\Omega$.
\item\label{alop}
$\|L\|_{L^\infty} \leq z$.
The pointwise norm is the operator norm on $\End(\R^\dimy)$.
\item\label{acs1}
$\|\p a\|_{L^\infty} \leq \eps$.
\item\label{acs2}
$\lim_{\tau \to \infty} \|\p a\|_{L^\infty_x} = 0$.
\item\label{afderyy}
$\|\p^\alpha a \|_{L^\infty} \leq b$
for $2 \leq |\alpha| \leq \dimx+3$.
\item\label{aLder}
$\|\p^\alpha L \|_{L^\infty} \leq b$
for $1 \leq |\alpha| \leq \dimx+3$.
\item\label{aFder}
$\|\p^\alpha F \|_{L^\infty_x} \leq \delta e^{-Z \tau}$
for $|\alpha| \leq \dimx+3$.
\item \label{affdec1}
$\| \p^\alpha F \|_{L^\infty_x} = \mathcal{O}(e^{-Z \tau})$ as $\tau \to \infty$ for all $\alpha$.
\end{enumerate}
then there exists a unique
$\fld \in C^\infty(\Omega,\R^\dimy)$ that satisfies,
abbreviating $\fld = \fld(\ind)$,
\begin{subequations} \label{sys1}
\begin{align}
\label{pde1}
(a^\mu(\ind) + A^\mu(\fld))\p_\mu \fld + L(\ind)\fld + 
\tfrac{1}{2} B(\fld,\fld) + F(\ind) & = 0\\
\label{acomp1}
\vst
Q^{-1} \one \leq a^0(\ind) + A^0(\fld) & \leq Q \one
\;\;\textnormal{on $\Omega$}\\ \displaybreak[0]
\label{adec11}
\vst
\textstyle\sup_{|\alpha| \leq 1} \|\p^\alpha \fld\|_{L^\infty_x} & \leq
\delta C e^{-Z \tau}
\\
\label{adec1}
\vst
\|\p^\alpha \fld \|_{L^\infty_x}
& = \mathcal{O}(e^{- Z \tau})
\;\;\text{as $\tau \to \infty$}
\end{align}
\end{subequations}
For uniqueness,
one does not need
$q$, $b$, $\delta$
and one can drop
\ref{ak2k},
\ref{acs1},
\ref{afderyy}, \ref{aLder}, \ref{aFder}, \ref{affdec1}
and \eqref{adec11}.
No constants $\eps$ and $C$ are produced in this case.
\end{theorem}
Equation \eqref{pde1} being quasilinear,
we make assumptions and draw conclusions
to keep the causal structure intact,
see \ref{ak2k} and \eqref{acomp1}.
The theorem says that, modulo all other assumptions,
the size of the solution $\fld$ is linear in the size of the inhomogeneity $F$,
see \ref{aFder} and \eqref{adec11},
and this is the purpose of $\delta$.
The uniqueness claim is not optimal;
one can state weaker conditions that single out this solution.
Another improvement would be to state smoothness properties
of the map $(a,A,L,B,F) \mapsto \fld$.

\begin{proof}[Uniqueness]
\newcommand{\diffu}{U}
This is an opportunity to recall  energy estimates.
Suppose $\fld_1$, $\fld_2$ are two solutions.
Then $\diffu = \fld_1-\fld_2$ satisfies the linear homogeneous
\begin{equation}\label{disceq1}
(a^\mu(\ind) + A^\mu(\fld_1))\p_\mu \diffu
= -L(\ind)\diffu - \tfrac12 B(\fld_1+\fld_2,\diffu)
- A^\mu(\diffu) \p_\mu \fld_2
\end{equation}
and $\|\diffu\|_{L^\infty_x} = \mathcal{O}(e^{-Z\tau})$.
We must show $\diffu = 0$.
Parse \eqref{disceq1} as
$w^\mu(\ind) \p_\mu \diffu = l(\ind) \diffu$
with $w^\mu \in C^\infty(\Omega,\Sym_\dimy)$,
$l \in C^\infty(\Omega,\End(\R^{\dimy}))$.
Set $j^\mu = \diffu^T w^\mu \diffu$. Then
$\p_\mu j^\mu = \diffu^T J \diffu$ where $J = \p_\mu w^\mu + l + l^T$,
and this crucially relies on the $w^\mu$ being symmetric.
Apply the divergence theorem
to $j^\mu$ on the portion of $\Omega$
that lies between $\tau \in [\tau_0,\tau_1]$ for some $0 \leq \tau_0 < \tau_1 < \infty$.
Abbreviating $E(\tau) = \int_{\mathbbm{T}^\dimx} dx\, j^0$ we obtain
\[
-E(\tau_0)
+ E(\tau_1)
 \;=\; \textstyle\int_{\tau_0}^{\tau_1} d\tau \int_{\mathbbm{T}^\dimx} dx\; \diffu^T J \diffu
\]
Differentiate this identity to get
$-\tfrac{\dd}{\dd \tau} E \leq \int_{\mathbbm{T}^\dimx} dx\, |\diffu^T J \diffu|$.
Therefore, using \eqref{acomp1},
\begin{equation}\label{agw1}
E(\tau_0) \leq E(\tau_1)\,
\exp(\textstyle\int_{\tau_0}^{\tau_1} \dd \tau\, Q \|J\|_{L^\infty_x})
\end{equation}
using the operator norm for $J$ as in \ref{alop}.
Tracking back the definitions of $J$, $w$, $l$ and using
\ref{alop} and $\|L\|_{L^\infty} = \|L^T\|_{L^\infty}$
and \ref{acs2} and $Qz < Z$ we get
$\limsup_{\tau \to \infty} Q \|J\|_{L^\infty_x} < 2Z$.
The decay of $\|\diffu\|_{L^\infty_x}$
implies $E(\tau) = \mathcal{O}(e^{-2Z\tau})$.
Letting $\tau_1 \to \infty$ in \eqref{agw1},
we get $E(\tau_0) = 0$ for all $\tau_0$, therefore $\diffu = 0$.
\qed\end{proof}
\begin{proof}[Existence]
Rather than providing a formula for $\eps$ upfront,
we make smallness assumptions as we go,
only finitely many, and only admissible ones, so depending only on things
that $\eps$ is allowed to depend on.

It suffices to prove the theorem under the additional hypothesis
that $F$ has compact support, meaning it vanishes for large $\tau$,
if we also prove in this case that, everything else being the same,
the map $F \mapsto \fld$
satisfies the following supplement:
For all series $(\delta_K)$ there exists a series $(C_K)$
such that for all $F$ if
$\|\p^\alpha F\|_{L^\infty_x} \leq \delta_{|\alpha|} \smash{e^{-Z \tau}}$
then $\|\p^\alpha \fld\|_{L^\infty_x} \leq C_{|\alpha|} \smash{e^{-Z \tau}}$.
These are variants of \ref{affdec1}, \eqref{adec1}.

The reduction to this case goes as follows.
Given $F$ without support condition,
set $F_s = \chi(\tau-s) F$, for some universally fixed $\chi \in C^\infty(\R)$
equal to one on $(-\infty,-1]$ and equal to zero on $[1,\infty)$.
The reduction requires making $\eps$ smaller once and $C$ bigger once,
by some fixed factor independent of $s$,
to absorb errors introduced by $\chi$.
(For example, $F_s$ satisfies \ref{aFder} for some $\delta$
if $F$ satisfies \ref{aFder} for some smaller $\delta$.)
The $F_s$ do have compact support and so we get solutions $\fld_s$.
There is a sequence $\lim_{i \to \infty} s_i =\infty$ such that
$\fld = \lim_{i \to \infty} \fld_{s_i}$ exists
with uniform convergence of all derivatives on all compact subsets of $\Omega$,
and then has all the desired properties.
The $(s_i)$ are constructed using a countable exhaustion of $\Omega$
by compact sets $X$
(obtained by truncating $\Omega$ at different times);
the compactness of $C^{k+1}(X) \hookrightarrow C^k(X)$
for all $k$;
the uniform boundedness of $(\fld_s)$ in $C^{k+1}(X)$ by \eqref{adec1}
and crucially the supplement;
and a diagonalization argument.
This concludes the reduction.

From now on, $F$ has compact support,
so vanishes for $\tau \geq T'$ for some $T' > 0$.
By local well-posedness for quasilinear symmetric hyperbolic systems \cite{tay},
with trivial data at time $\tau = T'$,
there exists a $0 \leq T < T'$
and a smooth $\fld$ on the $\tau \in (T,\infty)$ portion of $\Omega$,
satisfying \eqref{pde1}, \eqref{acomp1}, 
with $\fld$ identically zero for $\tau \geq T'$, such that:
\begin{quote}
If $\fld$ and all its derivatives of all 
orders are bounded on $\tau \in (T,\infty)$,
so $\fld$ extends smoothly to $\tau \in [T,\infty)$,
and if \eqref{acomp1} 
is not exhausted on $\tau \in [T,\infty)$,
then $T = 0$.
\end{quote}
This follows from symmetric hyperbolicity,
$a^\mu(\ind) + A^\mu(\fld) \in \Sym_\dimy$,
since \eqref{acomp1} 
keeps the causal structure intact.
One can conclude $T=0$ under these conditions,
using an open-closed argument (bootstrap),
since a solution can only
fail to admit continuation if a $C^1$-norm diverges \cite{tay}.

We must prove $T = 0$ and the supplement estimates.
We derive a-priori estimates,
using energy estimates for $\fld$ and for
its derivatives $\fld_\alpha = \p^\alpha \fld$.
Differentiating \eqref{pde1},
\[
	w^\mu(\ind,\fld)\p_\mu \fld_\alpha = l_{\alpha}(\ind,\fld_{\leq |\alpha|})
\]
where $\fld_{\leq K} = (\fld_{\beta})_{|\beta|\leq K}$ and where
\begin{align*}
w^\mu & = a^\mu(\ind) + A^\mu(\fld)\\
l_{\alpha} & = 
-L(\ind)\fld_\alpha
- \textstyle\sum_{|\beta|\leq|\alpha|} L_{\alpha\beta}(\ind) \fld_\beta
- \tfrac{1}{2} \textstyle\sum_{\substack{|\beta|+|\gamma|\leq |\alpha|+1\\
            |\beta|,|\gamma| \leq |\alpha|}}
  B_{\alpha\beta\gamma}(\fld_\beta,\fld_\gamma) - F_\alpha(\ind)
\end{align*}
By definition, and none of these objects depend on $\fld$:
\begin{itemize}
\item $F_\alpha = \p^\alpha F \in C^\infty(\Omega,\R^\dimy)$.
\item $L_{\alpha\beta} \in C^\infty(\Omega,\End(\R^{\dimy}))$
comes from $\p^\alpha$-differentiating 
$a^\mu(\ind) \p_\mu \fld + L(\ind) \fld$ and then subtracting
$a^\mu(\ind) \p_\mu \fld_\alpha + L(\ind) \fld_\alpha$.
By the Leibniz rule, 
$L_{\alpha\beta}$ is given by:
if $|\beta| = |\alpha|$, a sum of first derivatives of the $a^\mu$;
if $|\beta| < |\alpha|$, a sum of
derivatives of $a^\mu$ of order
in the interval $[2,|\alpha|]$
and of derivatives of $L$ of order in $[1,|\alpha|]$.
\item $B_{\alpha\beta\gamma} \in \Hom(\R^\dimy \otimes \R^\dimy, \R^\dimy)$
comes from $\p^\alpha$-differentiating $A^\mu(\fld)\p_\mu \fld + \tfrac{1}{2}B(\fld,\fld)$
and then subtracting $A^\mu(\fld)\p_\mu \fld_\alpha$.
\end{itemize}
Set $j_\alpha^\mu = \fld_\alpha^T w^\mu \fld_\alpha$. Then
$\p_\mu j^\mu_\alpha = \fld_\alpha^T (\p_\mu w^\mu)
\fld_\alpha + \fld_\alpha^T l_\alpha + l_\alpha^T \fld_\alpha$
by the symmetry of $w^\mu$.
Set $E_\alpha(\tau) = \int_{\mathbbm{T}^\dimx} \dd x\, j^0_\alpha$
which is comparable to the square of the $L^2_x$-norm of
$\fld_\alpha$ by \eqref{acomp1}.
By the divergence theorem, 
\[
-\tfrac{\dd}{\dd \tau} E_\alpha
\;\leq\;
\textstyle\int_{\mathbbm{T}^\dimx} \dd x\, |\p_{\mu} j_{\alpha}^\mu|
\;\leq\;
\big(
\| (\p_\mu w^\mu) \fld_\alpha\|_{L^2_x}
+ 2\|l_\alpha\|_{L^2_x}
\big)
Q^{1/2} E_\alpha^{1/2}
\]
using the $L^2_x$-Cauchy-Schwarz inequality. This implies,
denoting $e_\alpha = Q^{-1/2} E_\alpha^{1/2}$,
\begin{subequations} \label{eq:bigineq}
\begin{equation} \label{iq1}
-\tfrac{\dd}{\dd \tau} e_\alpha 
\leq 
\tfrac12 \| (\p_\mu w^\mu) \fld_\alpha\|_{L^2_x}
+ \|l_\alpha\|_{L^2_x}
\end{equation}
Abbreviating $K = |\alpha|$,
and denoting by $c_K>0$ a constant
that in addition to $K$ depends only
on the things that also $\eps$, $C$ are allowed to depend on
(which includes $Q$, $A$, $B$):
\begin{align}
\notag
\|(\p_\mu w^\mu)\fld_\alpha\|_{L^2_x}
& \leq
  c_K \|\p a\|_{L^\infty}e_\alpha
+ c_K S_1 e_\alpha
\\
\notag
\|l_\alpha\|_{L^2_x}
& \leq
  Q \|L\|_{L^\infty} e_\alpha
+ c_K \|\p a\|_{L^\infty} e_K
+ c_K U_K e_{<K}
+ c_K S_{\lfloor \frac{K+1}{2}\rfloor} e_{\leq K}
+ c_K \|F_\alpha\|_{L^\infty_x}\\
\notag 
S_1 & \leq c_K e_{\leq N+1}\\
\label{iq2}
S_{\lfloor \frac{K+1}{2}\rfloor}
& \leq c_K e_{\leq \lceil \frac{K+\dimx}{2} \rceil +1}
\end{align}
where $e_K = \max_{|\beta| = K} e_\beta$
and similarly $e_{<K}$ and $e_{\leq K}$;
and $S_K = \max_{|\beta| \leq K} \|\fld_\beta\|_{L^\infty_x}$;
and $U_K = \max\{\max_{2\leq|\beta|\leq K}\|\p^\beta a\|_{L^\infty},
                 \max_{1\leq |\beta| \leq K}\|\p^\beta L\|_{L^\infty}\}$.
In $\|F_\alpha\|_{L^2_x} \lesssim \|F_\alpha\|_{L^\infty_x}$
and the Sobolev inequalities in $x$
we have exploited boundedness of $\mathbbm{T}^\dimx$.
Note that
\[
\textstyle
\lceil\frac{K+\dimx}{2}\rceil+1 > \lfloor \frac{K+1}{2} \rfloor + \frac{\dimx}{2}
\]

The assumptions \ref{abdd},
\ref{alop},
\ref{acs1},
\ref{acs2},
\ref{afderyy},
\ref{aLder},
\ref{aFder},
\ref{affdec1}
imply:
\begin{align}
\label{iq3}
&\begin{aligned}
\|L\|_{L^\infty} & \leq z&
\|\p a\|_{L^\infty} & \leq \eps&
U_{\dimx + 3} & \leq b&
\textstyle\sup_{|\alpha| \leq \dimx+3}
\|F_\alpha\|_{L^\infty_x} & \leq \delta e^{-Z\tau}
\end{aligned}\\
&\begin{aligned}\label{iq4}
\textstyle\lim_{\tau \to \infty} \|\p a\|_{L^\infty_x} & = 0 &
U_K & < \infty &
\|F_\alpha\|_{L^\infty_x} & = \mathcal{O}(e^{-Z\tau})
\end{aligned}
\end{align}
\end{subequations}

\renewcommand{\ff}{\smash{f}}
View \eqref{eq:bigineq} as a system of differential inequalities for
the vector $(e_{\alpha})_{|\alpha| \leq \dimx + 3}$.
Namely insert \eqref{iq2}, \eqref{iq3} into the right hand side of \eqref{iq1}.
We claim that it implies
\begin{equation}\label{eq:eest}
e_\alpha \leq \delta R^{|\alpha|+1} e^{-Z\tau}
\qquad
\text{for
$|\alpha| \leq \dimx + 3$}
\end{equation}
for $\tau \in (T,\infty)$ provided the constant $R \geq 1$ satisfies,
for $K = 0\ldots \dimx + 3$,
\begin{equation}\label{eq:kin}
-ZR^{K+1}
+
\tfrac32 c (\eps + c \delta R^{\dimx + 4} e^{-Z\tau} ) R^{K+1}
+
Qz R^{K+1}
+
c b R^K
+
c
\leq
0
\end{equation}
where $c = \max\{c_0,\ldots,c_{\dimx+3}\} > 0$.
To see this, define $\ff_\alpha$ by 
$e_\alpha = \ff_\alpha e^{-Z\tau}$.
The claim is clearly that the $\tau$-dependent
vector $\ff = \smash{(\ff_{\alpha})_{|\alpha| \leq \dimx+3}}$
is always in the hypercube
$\smash{\prod_{|\alpha| \leq \dimx+3} [0,\delta R^{|\alpha|+1}]}$.
Note that $\ff = 0$ for large $\tau$,
and the system \eqref{eq:bigineq}
implies that $\ff$ cannot leave the hypercube
going backwards in $\tau$,
since by \eqref{eq:kin},
if $f$ ever reaches the boundary, then 
$-\tfrac{\dd}{\dd \tau}f$ cannot point out of the hypercube.

It suffices that \eqref{eq:kin} holds
in the special case when
$\tau=0$ and $K=0$,
and with $\delta$ replaced by $\eps \geq \delta$.
Therefore
$\tfrac32 \eps (1 + c R^{\dimx + 4} )
+
\tfrac1R (b+1)
\leq
B$ suffices, where $B = \tfrac1c (Z-Qz) > 0$.
This holds with
$R = \max\{1,\tfrac2B (b+1)\}$
and
$\eps \leq \tfrac13 B/(1+c R^{\dimx + 4})$.
They depend only on things that also
$\eps$ and $C$ may depend on.

Estimates for $(e_\alpha)_{|\alpha| = K}$ for $K > \dimx +3$
are likewise derived using \eqref{eq:bigineq} but now using \eqref{iq4}.
The estimates are obtained by induction on $K$.
For every $K$, the system of differential inequalities is linear\footnote{%
For instance $S_{\lfloor \frac{K+1}{2}\rfloor}$
is bounded in terms of $e_{< K}$, since we have $K > \dimx + 3$.},
of the schematic form
\[
-\tfrac{d}{d \tau} e_\alpha
\leq Qz e_\alpha
+ o(1) e_K
+ \mathcal{O}(e^{-Z\tau})
\]
Here $o(1)$ are terms decaying as $\tau \to \infty$;
they come from both $\|\p a\|_{L^\infty_x}$ and $e_{<K}$.
The inhomogeneity comes from $\|F_\alpha\|_{L^\infty_x}$ 
and terms controlled inductively
and it actually has compact support in $\tau$,
however we have written $\mathcal{O}(e^{-Z\tau})$
because one has control over the constant,
in the context of the supplement.
Using $Qz < Z$ once again, one obtains $e_K = \mathcal{O}(e^{-Z\tau})$
with control over the constant,
as required by the supplement.
One concludes that all derivatives of $\fld$ of all orders
are bounded when $\tau \in (T,\infty)$.
In particular, $\fld$ extends smoothly to $\tau \in [T,\infty)$.

We have established \eqref{adec11} on $\tau \in [T,\infty)$, using \eqref{iq2}, \eqref{eq:eest}
and the admissible $C = c R^{\dimx + 2}$.
Now \ref{ak2k}, $1 < q < Q$, \eqref{adec11},
 $\delta \leq \eps$ and an admissible smallness condition on $\eps$
imply \eqref{acomp1} without exhaustion on $\tau \in [T,\infty)$.
Hence $T=0$ as desired.
\qed\end{proof}
\subsection{A non-square system, in Maurer-Cartan form}\label{section:analysis2}

A non-square
is one with not as many equations as unknowns.
The system \eqref{pde2}
studied here is the MC equation in a graded Lie algebra,
that encodes symmetries and differential identities.
The Einstein equations admit such a formulation (Section \ref{section:gla}).
\begin{definition}[Axioms for a gLa $\gx$ with symmetric hyperbolic gauge I] \label{def:gwfobag}
Here $\Omega \subset \R^{1+\dimx}$
with partial derivatives $\p_\mu$, $\mu = 0\ldots \dimx$.
Suppose $\gx = \bigoplus_{i \in \Z} \gx^i$
is a gLa\footnote{%
The assumption that $\gx$ is a graded Lie algebra
entails that the graded Jacobi identities \eqref{eq:jacobi} must hold.
Here they amount to quadratic identities among the $A$ and $B$.}
with $\gx^i = C^\infty(\Omega,\R^{\dimy_i})$
where $\dimy_i \in \N$ and $\dimy_1 \geq 1$,
and the bracket of $u \in \gx^i$ and $v \in \gx^j$ is
\[
[u,v] \;=\;
   A_{ij}^\mu(u) \p_\mu v
   - (-1)^{ij} A_{ji}^\mu(v) \p_\mu u + B_{ij}(u,v)
\]
where
\begin{itemize}
\item $A_{ij}^\mu \in \Hom(\R^{\dimy_i},\Hom(\R^{\dimy_j},\R^{\dimy_{i+j}}))$.
\item $B_{ij} \in \Hom(\R^{\dimy_i} \otimes \R^{\dimy_j},\R^{\dimy_{i+j}})$
with $B_{ij}(u,v) = -(-1)^{ij} B_{ji}(v,u)$.
\end{itemize}
Additionally, as distinguished structure,
we require $\dimz_i \in \N$ and
$\ginj_i \in \Hom(\R^{\dimz_i},\R^{\dimy_i})$
and $S_i \in \Hom(\R^{\dimy_{i+1}},\R^{\dimz_i})$
such that the following complexes are exact\footnote{%
In particular
$S_i \ginj_{i+1} = 0$ and $\dimy_{i+1} = \dimz_i + \dimz_{i+1}$.}
\begin{equation}\label{exGSseq}
0 \to \R^{\dimz_{i+1}} \xrightarrow{\ginj_{i+1}} \R^{\dimy_{i+1}}
                        \xrightarrow{S_i} \R^{\dimz_i} \to 0
\end{equation}
and such that $A_i^\mu(-) = S_i A_{1i}^\mu(-) \ginj_i$
satisfies
$A_i^\mu \in \Hom(\R^{\dimy_1},\Sym_{\dimz_i})$.
That is, this map must take values in the vector space of symmetric
$\dimz_i \times \dimz_i$ matrices.
\end{definition}
The injection $\ginj_i$ will be used to fix a gauge,
and the surjection $S_i$ will be used to construct square subsystems.
It is a slight misnomer to call this a symmetric hyperbolic gauge,
since Definition \ref{def:gwfobag} contains no trace
of the positivity condition for a symmetric hyperbolic system;
this is an additional assumption in Theorem \ref{theorem:asyts}.

The next theorem shows that given an approximate MC element,
$[\amc,\amc] \approx 0$, then nearby there is a unique true solution.
Nothing is said about the construction
of the approximate MC element $\amc$,
which is the topic of later sections.
\begin{theorem}[Non-square system on a semiglobal domain] \label{theorem:asyts}
Given five constants
\begin{equation}\label{csx2}
1 < q < Q
\qquad
0 < Qz < Z
\qquad
b > 0
\qquad 
\end{equation}
Given a graded Lie algebra $\gx$ as in Definition \ref{def:gwfobag}
on $\Omega = [0,\infty) \times \mathbbm{T}^\dimx$.
Then there exist $\eps > 0$ and $C>0$\footnote{%
They depend only on the data already mentioned.
For example, they can depend on all the data in $\gx$.}
such that for all $0 < \delta \leq \eps$
and all
\begin{itemize}
\item[] $\amc \in \gx^1 = C^\infty(\Omega,\R^{\dimy_1})$
\end{itemize}
if
\begin{enumerate}[({b}1),leftmargin=10mm]
\item \label{bbounded} $\amc$
and all its derivatives of all orders are bounded on $\Omega$.
\item \label{bk2k}
$q^{-1} \one \leq A_i^0(\amc) \leq q \one$ on $\Omega$ for $i=1,2$.
Notation as in Definition \ref{def:gwfobag}.
\item \label{blop}
$\|L_i(\amc)\|_{L^\infty} \leq z$ for $i=1,2$ where, by definition,
\begin{equation}\label{eq:defli}
L_i(\amc) = S_i A_{i1}^\mu(\ginj_i -) \p_\mu \amc + S_i B_{i1}(\ginj_i -,\amc)
\quad
\in
\quad
C^\infty(\Omega,\End(\R^{\dimz_i}))
\end{equation}
The pointwise norm is the operator norm on $\End(\R^{\dimz_i})$.
\item \label{bcs1}
$\|\p (A_i(\amc))\|_{L^\infty} \leq \eps$ for $i=1,2$
where $A_i = (A_i^\mu)_{\mu = 0\ldots \dimx}$.
\item \label{bcs2}
$\lim_{\tau \to \infty} \|\p (A_i(\amc))\|_{L^\infty_x} = 0$
for $i=1,2$.
\item \label{jhgjdffd1}
$\|\p^\alpha (A_1(\amc))\|_{L^\infty} \leq b$ for $2 \leq |\alpha| \leq \dimx + 3$.
\item \label{jhgjdffd2}
$\|\p^\alpha (L_1(\amc))\|_{L^\infty} \leq b$ for $1 \leq |\alpha| \leq \dimx + 3$.
\item \label{hifhufn1}
$\| \p^\alpha [\amc,\amc] \|_{L^\infty_x} \leq \delta e^{-Z\tau}$ for $|\alpha| \leq \dimx + 3$.
\item \label{hifhufn2}
$\| \p^\alpha [\amc,\amc] \|_{L^\infty_x} = \mathcal{O}(e^{-Z\tau})$
as $\tau \to \infty$ for all $\alpha$.
\end{enumerate}
then there exists a unique
$\fld \in C^\infty(\Omega, \R^{\dimz_1})$
that satisfies,
abbreviating $\fld = \fld(\ind)$,
\begin{subequations} \label{sys2}
\begin{align}
\label{pde2}
\vst
[\amc+\ginj_1\fld,\amc+\ginj_1\fld] & = 0\\
\label{acomp2}
\vst
Q^{-1} \one \leq A^0_1(\amc + \ginj_1\fld) & \leq Q \one
\;\;\textnormal{on $\Omega$}\\\displaybreak[0]
\label{adec22}
\vst
\textstyle\sup_{|\alpha| \leq 1}
\|\p^\alpha \fld\|_{L^\infty_x} & \leq \delta C e^{-Z\tau}
\\
\label{adec2}
\vst
\|\p^\alpha \fld \|_{L^\infty_x}
& = \mathcal{O}(e^{-Z \tau})
\;\;\text{as $\tau \to \infty$}
\end{align}
\end{subequations}
\end{theorem}
\begin{proof}
We use Theorem \ref{theorem:semia} twice,
with the parameters in Table \ref{table:parthm}.
The constants produced by the $i=1$ invocation are denoted
$\eps_1$ and $C_1$
(and a look at Table \ref{table:parthm}
 shows that they are allowed to influence our choice of $\eps$ 
 and $C$ in Theorem \ref{theorem:asyts});
no such constants are produced by the $i=2$ invocation
because it only uses uniqueness in Theorem \ref{theorem:semia}.

\begin{table}[h]
\centering
\begin{tabular}{c||c|c}
Theorem \ref{theorem:semia} & $i=1$ & $i=2$  (uniqueness only)\\
\hline\hline
$q$, $Q$ & $q$, $Q$ & not needed, $Q$\\
$z$, $Z$
    & $z$, $Z$
    & $\tfrac12 (z+Z/Q)$, $Z$\\
$b$ & $b$ & not needed\\
\hline
$\dimy$ & $\dimz_1$ & $\dimz_2$ \\
$\Omega$ & $\Omega$ & $\Omega$ \\
$A^\mu$ & $A_1^\mu(\ginj_1-)$ & $0$ \\
$B$ & $S_1B_{11}(\ginj_1-,\ginj_1-)$ & $0$ \\
$\delta$ & $\delta_1$ & not needed \\
\hline
$a^\mu$ & $A_1^\mu(\amc)$ & $A_2^\mu(\amc + \ginj_1\fld)$\\
$L$ & $L_1(\amc)$ &
$-L_2(\amc) - L_2'$, see \eqref{eql2st}\\
$F$ & $\tfrac12 S_1[\amc,\amc]$ & $0$\\
\hline
\hline
$\fld$ & $\fld \in C^\infty(\Omega,\R^{\dimz_1})$
 & $\fld' \in C^\infty(\Omega,\R^{\dimz_2})$
defined by \eqref{eq:defup}
\\
\eqref{pde1}
& 
$S_1 [\amc + \ginj_1\fld,\amc + \ginj_1\fld] = 0$
&
$S_2 [\amc + \ginj_1\fld,\ginj_2\fld'] = 0$
\end{tabular}
\caption{The rows up to and including $F$
are parameters used to invoke Theorem \ref{theorem:semia}.
As part of this proof one must
check that the partial differential equation \eqref{pde1}
is indeed that in the last row of the table.
Of course, this is by design.
}
\label{table:parthm}
\end{table}

Uniqueness follows from uniqueness in Theorem \ref{theorem:semia}
using column $i=1$ of Table \ref{table:parthm},
since then \eqref{pde1} is necessary for \eqref{pde2}.

Existence is shown in steps.
Apply Theorem \ref{theorem:semia}
using column $i=1$ of Table \ref{table:parthm}.
This yields a $\fld$ solving \eqref{sys2}
except that only the $S_1$-part of \eqref{pde2} is known to hold.
To make sure the assumptions of Theorem \ref{theorem:semia} hold,
one makes the following admissible choices.
Require $\eps \leq \eps_1$, so \ref{acs1} holds.
Require $\eps \leq \eps_1/J$,
and set $C = J C_1$ and $\delta_1 = J\delta$
for some $J \geq 1$ that depends only on a norm of $S_1$,
chosen such that \ref{hifhufn1} implies \ref{aFder}.
Since $0 < \delta \leq \eps$ we have
$0 < \delta_1 \leq \eps_1$ as required.
And \eqref{adec11} now implies \eqref{adec22}.

The exactness of \eqref{exGSseq}, for $i=1$, implies
\begin{equation}\label{eq:defup}
[\amc + \ginj_1\fld,\amc + \ginj_1\fld] = \ginj_2 \fld'
\end{equation}
for a unique $\fld' \in C^\infty(\Omega,\R^{\dimz_2})$.
The only thing left to show is that $\fld' = 0$.
By the
$(\gx^1)^{\otimes 3} \to \gx^3$
Jacobi identity, $[\amc + \ginj_1\fld,\ginj_2\fld'] = 0$.
In particular $S_2 [\amc + \ginj_1\fld,\ginj_2\fld'] = 0$.
Uniqueness in Theorem \ref{theorem:semia},
using column $i=2$ of Table \ref{table:parthm},
implies $\fld' = 0$ since $0$ is a solution of this linear homogeneous system
for $u'$.

But we must check that all assumptions of the $i=2$ invocation hold:
\ref{abdd}, \ref{alop}, \ref{acs2}
but also \eqref{acomp1}, 
\eqref{adec1} since,
logically, they are also assumptions for uniqueness.
Here \ref{abdd}, \ref{acs2} and \eqref{adec1} are clear.
For the others we use the
$i=2$ versions of \ref{bk2k}, 
\ref{blop}.
For \ref{alop} it suffices that
$\|L_2'\|_{L^\infty} \leq \tfrac12 (Z/Q-z)$ where, by definition,
\begin{equation}\label{eql2st}
\textstyle
L_2' = S_2 A_{21}^\mu(\ginj_2-) \p_\mu \ginj_1\fld +S_2 B_{21}(\ginj_2-,\ginj_1\fld)
\end{equation}
This and \eqref{acomp1} 
follow from the already established
\eqref{adec22} and
$\delta C \leq \eps C$,
by making another admissible smallness assumption about $\eps$.
\qed\end{proof}
\section{The Maurer-Cartan formulation of the Einstein equations} \label{section:gla}
This section contains a reformulation
of \eqref{eq:esf} and \eqref{eq:vc}
that are in no way specific to the BKL problem.
We define two real gLa
with symmetric hyperbolic gauges:
\begingroup
\renewcommand{\arraystretch}{1.25}
\[
\begin{tabular}{r|l}
graded Lie algebra & Maurer-Cartan equations\\
\hline
$\E$ & Einstein vacuum equations\\
$\E_\Phi = \E \oplus \Phi$ & Einstein equations coupled to a massless scalar field
\end{tabular}
\]
\endgroup
They satisfy the axioms in Definition \ref{def:gwfobag},
but it is convenient to first
introduce more algebraic axioms in Definition \ref{def:axg}.
The translation back to Definition \ref{def:gwfobag} is in Lemma \ref{lemma:translation}.
The construction of $\E$ is
given conceptually in \cite{rtgla}, some of whose results we use,
whereas the extension to $\E_\Phi$ is new.
Conceptually, these constructions are functors
from a groupoid of rank four vector bundles
with conformal inner product to the category of real gLa,
but our presentation here will not make this clear and is practical.
In particular, we equip all objects with a distinguished basis
for later convenience.
\step
For concreteness, in this section
\[
\Omega = [0,\infty) \times M^3
\]
where $M^3$ is a parallelizable 3-dimensional manifold.
We denote by $\p_0$ the derivative on the first factor,
and by $L_1,L_2,L_3 \in \Gamma(TM^3)$ a fixed frame.
The presentation in \cite{rtgla} assumes a base manifold
diffeomorphic to $\R^4$ for simplicity.
The same construction is valid on $\Omega$
with minor modifications when invoking the Poincare lemma.
We will point out where such modifications are necessary.

\subsection{The module $W$ and its conformal inner product}

Set $\cinf = \cinf(\Omega,\R)$.
Note that $\Der(\cinf)$ is a free $\cinf$-module,
as distinguished basis take $\p_0,L_1,L_2,L_3$.
Let $\theta_0 \ldots \theta_3$
be four symbols and define the free $\cinf$-module
$W = \cinf \theta_0 \oplus \ldots \oplus \cinf \theta_3$,
the module of sections of a trivial rank four vector bundle.

Endow $W$ with a conformal inner product
by declaring that $\theta_0 \ldots \theta_3$
is a conformally orthonormal basis with signature ${-}{+}{+}{+}$. 
This means that the symmetric $\cinf$-bilinear map $\langle-,-\rangle : W \times W \to \cinf$
defined by
\begin{equation}\label{eq:nijm}
\langle \theta_i,\theta_j \rangle = \eta_{ij}
\qquad
i,j = 0 \ldots 3
\end{equation}
is a representative of the conformal inner product, where $\eta_{ij} = \diag(-1,1,1,1)_{ij}$.
Any other representative is given by $f \langle -, - \rangle$
for some $f \in \cinf$ with $f>0$.
Define the `future timelike elements' to be
\[
W_+ \;=\; \big\{\,
a^i \theta_i \in W
\mid a^i \in \cinf,
\;\; a^0 > 0,
\;\; \eta_{ij}a^ia^j < 0
\;\;
\text{on $\Omega$}
\big\}
\]

The $\cinf$-exterior algebra $\wedge W$
has rank 16. As distinguished basis we take
$\theta_{i_1} \cdots \theta_{i_m}$
where $0 \leq i_1 < \ldots < i_m \leq 3$ and $0 \leq m \leq 4$
(we omit $\wedge$ in the notation).
Let $\Der(\wedge W)$ be the graded Lie algebra of $\R$-linear graded derivations.
In particular, the degree zero derivations
comprise a Lie algebra $\Der^0(\wedge W)$, an element of which 
is determined by how it acts on the basis elements $\theta_0\ldots \theta_3$
and by how it acts as a derivation on $\cinf$, and this data is not constrained further,
hence
\begin{equation}\label{eq:derw0}
\Der^0(\wedge W) \simeq \End_\cinf(W) \oplus \Der(\cinf)
\end{equation}
More generally\footnote{This isomorphism,
which depends on the basis for $W$, arises as follows.
Let $A$ be any unital $\cinf$ graded commutative algebra.
Let $\Der(A)$ be the gLa of $\R$-linear graded derivations,
meaning the Leibniz rule is graded.
If $A$ is free over $\cinf$,
then there is a split short exact sequence
$0 \to \Der_\cinf(A) \to \Der(A) \to A \otimes_\cinf \Der(\cinf) \to 0$.
Here $\Der_\cinf(A)$ are the $\cinf$-linear graded derivations.
If $A$ is freely generated as a unital $\cinf$-gca by
a graded submodule $S \hookrightarrow A$,
then $\Der_\cinf(A) \simeq \Hom_\cinf(S,A)$ as graded $\cinf$-modules.}
\begin{equation}\label{eq:derww}
\Der^i(\wedge W) \simeq \Hom_\cinf(W,\wedge^{i+1}W)
\oplus (\wedge^i W \otimes_\cinf \Der(\cinf))
\end{equation}
and this is a free $\cinf$-module of rank $4 {4 \choose i+1} + 4 {4 \choose i}$.
The various bases chosen previously
canonically determine a distinguished basis of $\Der(\wedge W)$
via \eqref{eq:derww}.
\subsection{Algebroid structure and symmetric hyperbolic gauges}
The following will serve as a blueprint
for the definitions of $\E$ and $\E_\Phi$.
\begin{definition}[Axioms for a gLa $\gx$ with symmetric hyperbolic gauge II] \label{def:axg}
Let $\gx$ be a gLa with the following additional properties
and distinguished structures:
\begin{itemize}
\item \emph{Module structure.}
The gLa $\gx$ is a graded $\wedge W$-module,
in particular a $\cinf$-module through $\cinf \hookrightarrow \wedge W$.
As a $\cinf$-module it is free.
\item \emph{Algebroid structure.}
A $\wedge W$-linear map and gLa morphism $\anc: \gx \to \Der(\wedge W)$,
called the anchor,
such that
for all homogeneous $x,y\in \gx$ and $\omega \in \wedge W$\footnote{%
As usual,
a symbol in the exponent of $(-1)$
is a shorthand for the degree of that symbol.
}:
\[
[x,\omega y]
\;=\;
\anc(x)(\omega)y + (-1)^{x\omega} \omega[x,y]
\]
\item \emph{Symmetric hyperbolic gauge structure.}
A free graded $\cinf$-submodule $\gx_\GS \subset \gx$
such that scalar multiplication by
any $w \in W_+$ is an injective map $\gx_\GS \to \gx$
and
\[
\gx = \gx_\GS \oplus w\gx_\GS
\]
Further, $\cinf$-bilinear maps
$b^i : \gx_\GS^i \times \gx^{i+1} \to \cinf$
such that $b^i(-,w-)|_{\gx_\GS^i \times \gx_\GS^i}$ is symmetric for all $w \in W$,
positive definite if $w \in W_+$. Furthermore,
\begin{equation}\label{eq:exk}
\gx_\GS^{i+1} = \ker b^i(\gx_\GS^i,-)
\end{equation}
\item \emph{Bases and constant coefficients\footnote{%
This condition
is useful in the present application,
but one may want to drop it in other applications.}.}
Both $\gx$ and $\gx_\GS$
are equipped with an ordered $\cinf$-basis of homogeneous elements
such that the bracket
$\gx \times \gx \to \gx$,
module multiplication $\wedge W \times \gx \to \gx$,
the inclusion $\gx_\GS \hookrightarrow \gx$,
the anchor $\anc$, and the $b^i$, all have constant coefficients.
That is, applying them to $\cinf$-basis elements (or pairs of basis elements)
returns an $\R$-linear combination of basis elements.
\end{itemize}
\end{definition}
Note that $\Der(\cinf)$ and its dual
$\Der(\cinf)^\ast$ are respectively the free $\cinf$-modules of sections of the 
tangent and cotangent bundles of $\Omega$.
Consider the composition:
\begin{equation}\label{eq:sss1}
\gx^1 
\;\xrightarrow{\;\anc\;}
\Der^1(\wedge W)
\;\to\;
 W \otimes_\cinf \Der(\cinf) \\
\simeq \Hom_\cinf(\Der(\cinf)^\ast,W)
\end{equation}
where the second map from the left is the restriction
of a $\Der^1(\wedge W)$ element to a map $\cinf \to W$,
which is an element of 
$ W \otimes_\cinf \Der(\cinf)$.
(This is also a special case of \eqref{eq:derww}.)
By definition, \eqref{eq:sss1} associates to every element of $\gx^1$ its frame.
\begin{definition}[Nondegenerate frame and conformal metric] \label{def:nondeg}
An element of $\gx^1$ is said to be nondegenerate iff its frame
(via \eqref{eq:sss1}) is
an isomorphism $\Der(\cinf)^\ast \simeq W$.
Every such element of $\gx^1$ induces a conformal metric on the manifold $\Omega$
via \eqref{eq:nijm}.
\end{definition}
\begin{lemma}[%
Translation lemma] \label{lemma:translation}
Here $\Omega = [0,\infty) \times \mathbbm{T}^3$
and $L_i = \p_i$ with $i=1,2,3$
for compatibility with Definition \ref{def:gwfobag}.
Suppose $\gx$ is a gLa as in Definition \ref{def:axg}, then:
\begin{itemize}
\item $\gx$ satisfies Definition \ref{def:gwfobag},
where $\dimx = 3$ and
$\dimy_i = \rank_\cinf \gx^i$ and
$\dimz_i = \rank_\cinf \gx^i_\GS$ and where,
using the distinguished bases,
the matrix $\ginj_i$ is the injection $\gx^i_\GS \hookrightarrow \gx^i$;
the matrix $S_i$ is the bilinear pairing $b^i$
viewed as a map $\gx^{i+1} \to (\gx_\GS^i)^\ast$;
and
\begin{equation}\label{eq:ttt1}
  A_{ij}^\mu(u) \p_\mu f
 = \anc(u)(f)
\end{equation}
for all $u \in \gx^i$, $f \in \cinf$
where the right hand side, nominally in $\wedge^i W$,
is here viewed as an element of $\Hom_\cinf(\gx^j,\gx^{i+j})$
via module multiplication.
The $B_{ij}$ are given by the bracket $\gx^i \times \gx^j \to \gx^{i+j}$
on basis elements and are extended $\cinf$-bilinearly.
\item
For all $u \in \gx^1$, $f \in \cinf$,
if $\anc(u)(f) \in W_+$
then $A^\mu_i(u) \p_\mu f > 0$,
pointwise as a symmetric matrix.
If $u \in \gx^1$ is nondegenerate
then $\keroneform_\mu A_i^\mu(u) > 0$
for all $\keroneform \in \Der(\cinf)^\ast$ that are future timelike
in the sense that they are in $W_+$ via the frame of $u$.
\end{itemize}
\end{lemma}
\begin{proof}
In the first claim, the exactness of \eqref{exGSseq} is by \eqref{eq:exk}
and the rest is clear.
In the second claim, 
note that \eqref{eq:sss1} and \eqref{eq:ttt1}
imply that $\keroneform_\mu A_{1i}^\mu(u) : \gx^i \to \gx^{i+1}$
corresponds to module multiplication
by an element $w \in W_+$, then use the definition of
$S_i$ and $\ginj_i$ and the positivity property of $b^i$. 
\qed\end{proof}
\subsection{Einstein vacuum}
Let $\CDer(W) \subset \Der^0(\wedge W)$ be the sub Lie algebra of all $\delta$
that preserve the conformal inner product, meaning
there exists an $f \in \cinf$ such that
$\delta \langle -,- \rangle
= \langle \delta -,- \rangle + \langle -, \delta -\rangle + f \langle -,-\rangle$
in $W \times W \to \cinf$; this definition does not depend on the
choice of a representative of the conformal inner product.
Then $\CDer(W)$ is free of rank $7 + 4$,
and we use the following distinguished basis
(closely related to \eqref{eq:derw0}):
\begin{itemize}
\item $\sigma_0,\sigma_1,\sigma_2,\sigma_3,\sigma_{23},\sigma_{31},\sigma_{12}$ annihilate $\cinf$
and $\sigma_0(\theta_a) = \theta_a$
and $\sigma_{ab}(\theta_c) = \eta_{bc}\theta_a - \eta_{ac}\theta_b$.
We abbreviate $\sigma_1 = \sigma_{01}$ and $\sigma_2 = \sigma_{02}$
and $\sigma_3 = \sigma_{03}$.
\item
$\p_0,L_1,L_2,L_3$ which are derivations on $\cinf$
are extended to annihilate $\theta_0\ldots\theta_3$.
This is a basis-dependent injection
$\Der(\cinf) \hookrightarrow \CDer(W)$.
\end{itemize}

Set
$\L = \wedge W \otimes_\cinf \CDer(W)$
and define the bracket
\begin{equation}\label{eq:brak}
[\omega\delta,\omega'\delta'] = \omega\omega'[\delta,\delta']
+ \omega\delta(\omega')\delta' - \delta'(\omega)\omega'\delta
\end{equation}
for all $\omega,\omega' \in \wedge W$ and
$\delta,\delta' \in \CDer(W)$.
It is well-defined and a gLa.
It is a $\wedge W$ graded Lie algebroid with
anchor $\L \to \Der(\wedge W)$, $\omega \delta \mapsto (\omega' \mapsto \omega\delta(\omega'))$.

Let $\I^2 \subset \L^2$ be the $\cinf$-submodule generated by
  \begin{equation}\label{eq:i2}
  \RE
  \left[\begin{pmatrix}
    \theta_0\theta_1 + i \theta_2\theta_3\\
    \theta_0\theta_2 + i \theta_3\theta_1\\
    \theta_0\theta_3 + i \theta_1\theta_2
  \end{pmatrix}^T
  S
  \begin{pmatrix}
    \sigma_1 + i\sigma_{23}\\
    \sigma_2 + i\sigma_{31}\\
    \sigma_3 + i\sigma_{12}
  \end{pmatrix}
\right]
  \end{equation}
where $S$ runs over all symmetric traceless complex $3\times 3$ matrices.
Set $\I = (\wedge W) \I^2$.
By direct calculation \cite{rtgla}, $\I \subset \L$ is a gLa ideal,
contained in the kernel of the anchor $\L \to \Der(\wedge W)$.
So $\E = \L/\I$ is a gLa,
with induced anchor $\anc: \E \to \Der(\wedge W)$.

We must choose a distinguished basis for $\E$.
The distinguished bases for $\wedge W$ and $\CDer(W)$ yield a basis for $\L$,
and one can take a subset of this basis to be the distinguished basis for $\E$.
Alternatively
the $\E$-column of Table \ref{table:EPhiG}
yields a concrete basis for $\E$.
All these choices are equivalent for the present purpose.
\begin{definition}[Associated one-form] \label{def:aof1}
To every nondegenerate $u \in \E^1$
we associate a one-form $\alpha_u$ as follows.
Decompose $u = w \sigma_0 + u'$
where $w \in W$
whereas $u'$ does not contain $\sigma_0$.
Then $\alpha_u$ is the one-form corresponding to $w$ via the frame of $u$.
\end{definition}
An example is in Remark \ref{remark:kasner2}.
\begin{lemma}[gLa Einstein vacuum] \label{lemma:glavac}
\begin{itemize}
\item 
The gLa $\E = \L/\I$
satisfies the axioms in Definition \ref{def:axg}.
A symmetric hyperbolic gauge $\E_\GS \subset \E$
along with bilinear forms exists, by a construction in \cite{rtgla},
and we assume that one is fixed.
\item
If $u \in \MC(\E)$ is nondegenerate in the sense of Definition \ref{def:nondeg},
then the corresponding conformal metric
has a unique, up to multiplication by a positive constant,
representative metric $g$ on the universal
cover of $\Omega$ that is Ricci-flat, $\Ric_g = 0$.
Furthermore $\alpha_u$ is closed,
and $g$ descends to a metric on $\Omega$ iff $\alpha_u$ is exact.
\end{itemize}
\end{lemma}
\begin{proof}
See \cite{rtgla}.
The conformal factor 
is the exponential of the integral of $\alpha_u$,
hence the exactness requirement to get a metric on $\Omega$.
An example is in Remark \ref{remark:kasner2}.
\qed\end{proof}
\subsection{Einstein coupled to massless scalar field}

Here we define $\E_\Phi = \E \oplus \Phi$.
This extension is motivated by algebraic considerations,
and we only make the connection to \eqref{eq:esf} in Lemma \ref{lemma:gsf}.
It is determined by the axioms in Definition \ref{def:axg}
and the following additional stipulations
(it is likely that one can add
other matter fields to the formalism of \cite{rtgla}
by similar reasoning):
\begin{enumerate}[({s}1),leftmargin=10mm]
\item \label{stip:br}
$[\E,\E] \subset \E$
and $[\E,\Phi] \subset \Phi$ and $[\Phi,\Phi] \subset \E$
where the first is the bracket of $\E$.
\item \label{stip:phi}
$\Phi = \Phi^0 \oplus \cdots \oplus \Phi^4$
where\footnote{We set $\Phi^0 = 0$ since there is
no gauge freedom attached to a scalar field.
To motivate the rest, consider for simplicity the linear homogeneous wave equation
$\Box f = 0$ on Minkowski space $\R^{1+3}$.
It is equivalent to the first order system
$d \omega = 0 \in \Omega^2$ and $d(\ast \omega) = 0 \in \Omega^4$ for
a one-form $\omega \in \Omega^1$. The equivalence is up to an irrelevant
additive constant, and given by $\omega = df$ respectively $f = \int \omega$.
This is the first homology of the complex
$0 \to 0 \to \Omega^1 \to \Omega^2\otimes \Omega^4 \to \Omega^3 \to \Omega^4 \to 0$,
whose tail is the de Rham tail.
}
\begin{align*}
\Phi^0 & = 0&
\Phi^1 & = W&
\Phi^2 & = \wedge^2 W \oplus \wedge^4 W&
\Phi^3 & = \wedge^3 W&
\Phi^4 & = \wedge^4 W
\end{align*}
with distinguished $\cinf$-basis inherited from the ring $\wedge W$.
To distinguish this basis from the basis of the ring, we add underlines.
So $\inphi{\theta_0\theta_1\theta_2}$ is an element of $\Phi^3$,
rather then of the ring summand $\wedge^3 W$.
We never explicitly refer to the bases of $\wedge^4 W \subset \Phi^2$
and $\wedge^4 W \subset \Phi^4$, so do not introduce notation for them.
\item \label{stip:mod} The $\wedge W$ module structure is defined as follows.
The map $W \times \Phi^1 \to \Phi^2$ is
$(w,\phiel) \mapsto w\phiel \oplus w\phiel^\ast$
where $\phiel^\ast \in \wedge^3W$ is
the basis Hodge dual\footnote{%
Explicitly,
$\theta_0^\ast = \theta_1\theta_2\theta_3$
and $\theta_1^\ast = \theta_0\theta_2\theta_3$
and $\theta_2^\ast = \theta_0\theta_3\theta_1$
and $\theta_3^\ast = \theta_0\theta_1\theta_2$.
The construction actually does not depend
on the choice of a conformally orthonormal basis,
if the summand $\wedge^4 W \subset \Phi^2$
is understood to be multiplied by a suitable fractional density bundle associated to $W$.
Somewhat informally, this would be like defining
$\theta_0^\ast  = \theta_1\theta_2\theta_3
/ (\theta_0 \theta_1\theta_2\theta_3)^{1/2}$ etc.
}.
The map $W \times \Phi^2 \to \Phi^3$ is $(w,\phiel_1 \oplus \phiel_2) \mapsto w \phiel_1$.
And $W \times \Phi^3 \to \Phi^4$ is $(w,\phiel) \mapsto w\phiel$.
\item \label{stip:epp}
$[\E^0,\Phi^1] \subset \Phi^1$
is the canonical action of $\E^0 \simeq \CDer(W)$ on $\Phi^1 \simeq W$.
\item \label{stip:anc}
The anchor map $\anc : \E_\Phi \to \Der(\wedge W)$ coincides with the anchor for $\E$,
that is, it factors through $\E_\Phi \to \E$.
\item \label{stip:sign}
A choice of sign remains, see below. It will be motivated by energy conditions.
\end{enumerate}

By \ref{stip:br},
the full bracket is uniquely determined by
\begin{subequations}
\begin{align}
\label{b011}
\E^0 \times\Phi^1 & \to \Phi^1\\
\label{b112}
\Phi^1 \times \Phi^1 & \to \E^2
\end{align}
\end{subequations}
since $\E$ and $\Phi$ are generated by respectively
$\E^0$ and $\Phi^1$ over $\wedge W$, see \ref{stip:phi}, \ref{stip:mod}. 
\begin{itemize}
\item
\eqref{b011} is determined by \ref{stip:epp}.
This then determines $\E \times \Phi \to \Phi$ uniquely
and the Jacobi identity
$\E \times \E \times \Phi \to \Phi$ holds\footnote{%
So $\E \times \Phi \to \Phi$ is a representation of $\E$ on $\Phi$.
The action of $\E^0 \simeq \CDer(W)$ on
$\wedge^i W \subset \Phi^i$ 
is the standard one, but not the one on
$\wedge^4 W \subset \Phi^2$.
For example $[\sigma_0,-]$ acts like multiplication by $i$
on $\wedge^i W \subset \Phi^i$,
but like multiplication by $2$ on
$\wedge^4 W \subset \Phi^2$,
due to the density bundle mentioned in another footnote.
}.
\item
\eqref{b112}
must be $\cinf$-bilinear by \ref{stip:anc}; symmetric;
satisfy $\anc([\Phi^1,\Phi^1]) = 0$\footnote{%
By \ref{stip:anc} we have $\anc(\Phi) = 0$,
and 
since $\anc$ is a gLa map,
$\anc([\Phi,\Phi]) = 0$.
};
the Jacobi identity $\E^0 \times \Phi^1 \times \Phi^1 \to \E^2$ must hold\footnote{%
In particular,
\eqref{b112} must have constant coefficients relative to the distinguished bases.
};
and for all $v^i \in \cinf$ with
$\eta_{ij}v^iv^j = 0$
one needs $v^i v^j \theta_i [\inphi{\theta_j},\Phi^1] = 0$\footnote{%
Necessary for the existence of a bracket
$\Phi^2 \times \Phi^1 \to \E^3$ with
$[w-,-] = w[-,-] : \Phi^1 \times \Phi^1 \to \E^3$ for all $w \in W$.
In fact, the
$v^i v^j \theta_i \otimes \inphi{\theta_j}$
are in the kernel of 
the module multiplication
$W \otimes \Phi^1 \to \Phi^2$
in \ref{stip:mod}.}.
This implies that for all $\phiel = a^i \inphi{\theta_i} \in \Phi^1$,
\[
[\phiel,\phiel]
\;=\;
(-6 a^p a^q
+
\eta_{bc} a^ba^c \eta^{pq}
) \eta^{ij} \theta_i \theta_p \sigma_{jq}\;\in\; \E^2
\]
for all $a^i \in \cinf$.
The multiplicative constant is irrelevant,
except for the sign which, as announced in \ref{stip:sign},
was chosen consistent with the sign in \eqref{eq:esf}.
\end{itemize}
One checks that the bracket exists and defines a gLa,
and $[\Phi^1,\Phi^3] = [\Phi^2,\Phi^2] = 0$.
\begin{definition}[Associated one-form] \label{def:aof2}
For every $\eel \oplus \phiel \in \E_\Phi^1$
with nondegenerate $\eel$,
denote by $\beta_{\eel \oplus \phiel}$ the one-form
associated to $\phiel$ via the frame of $\eel$.
\end{definition}
\begin{lemma}[gLa Einstein coupled to massless scalar field] \label{lemma:gsf}
Let $\E_\GS$ be any gauge for $\E$, with bilinear forms, as in Lemma \ref{lemma:glavac}.
Then
\begin{itemize}
\item The gLa $\E_\Phi$
satisfies the axioms in Definition \ref{def:axg},
with symmetric hyperbolic gauge
$\E_{\Phi,\GS} = \E_\GS \oplus \Phi_\GS$ where $\Phi_\GS$ is given by,
with notation from \ref{stip:phi},
\begin{align*}
\Phi^1_\GS & = \Phi^1 &
\Phi^2_\GS & = \cinf \inphi{\theta_2\theta_3} \oplus
               \cinf \inphi{\theta_3\theta_1} \oplus
               \cinf \inphi{\theta_1\theta_2} \\
\Phi^3_\GS & = \cinf \inphi{\theta_1\theta_2\theta_3} &
\Phi^4_\GS & = 0
\end{align*}
and we omit obvious compatible bilinear forms as required in Definition \ref{def:axg}.
\item If $\eel \oplus \phiel \in \MC(\E_\Phi)$
and if $\eel$ is nondegenerate in the sense of Definition \ref{def:nondeg},
then the one-form $\beta_{\eel \oplus \phiel}$ is closed,
and \eqref{eq:esf} holds on the universal cover of $\Omega$
where $g$ is determined by $\eel$ as in Lemma \ref{lemma:glavac}
and $\phi$ is a nonzero constant times the integral of $\phiel$.
Furthermore $\alpha_{\eel}$ and $\beta_{\eel \oplus \phiel}$ are closed,
and we get a solution to \eqref{eq:esf} on $\Omega$
if and only if
$\alpha_{\eel}$ and $\beta_{\eel \oplus \phiel}$ are exact.
\end{itemize}
\end{lemma}
\begin{proof}
The MC equations $[\eel,\eel] + [\phiel,\phiel] = [\eel,\phiel] = 0$ imply \eqref{eq:esf}.
Module multiplication by $w \in W_+$ is injective
as a map $\Phi_\GS^i \to \Phi^{i+1}$; for $i=1$ use the Hodge dual in \ref{stip:mod}.
\qed\end{proof}

\begin{remark}[Kasner metric] \label{remark:kasner2}
On
$\Omega = [0,\infty) \times \mathbbm{T}^3$
consider the element of $\E_\Phi^1$ given by
\[
\begin{aligned}
& \theta_0 \p_0
+ e^{-(p_2+p_3)\tau} \theta_1 \p_1
+ e^{-(p_3+p_1)\tau} \theta_2 \p_2
+ e^{-(p_1+p_2)\tau} \theta_3 \p_3 &\quad& \textnormal{conformal frame}\\
& + (p_1+p_2+p_3)\theta_0\sigma_0
+ p_1 \theta_1\sigma_1
+ p_2 \theta_2\sigma_2
+ p_3 \theta_3\sigma_3 && \textnormal{connection}\\
& + \bklq \inphi{\theta_0} && \textnormal{scalar field}
\end{aligned}
\]
By direct calculation,
it is in $\MC(\E_\Phi)$ if and only if $p_2p_3 + p_3p_1 + p_1p_2 = 3 \bklq^2$.
It is nondegenerate in the sense of Definition \ref{def:nondeg},
with associated conformal metric
\[
g = A^2 \cdot
\big(
-(\dd \tau)^2
+ e^{2(p_2 + p_3)\tau} (\dd x^1)^2
+ e^{2(p_3 + p_1)\tau} (\dd x^2)^2
+ e^{2(p_1 + p_2)\tau} (\dd x^3)^2
\big)
\]
The conformal factor $A>0$ that yields a solution to \eqref{eq:esf}
is $A = \exp(-\int \alpha)$ using the one-form
$\alpha = (p_1+p_2+p_3)\dd \tau$ in Definition \ref{def:aof1}.
We get the metric in Remark \ref{remark:kasner1}.
\end{remark}
\begin{remark}[Nondegeneracy is preserved] \label{remark:nondegpreserved}
For simplicity assume $\Omega = [0,\infty) \times \mathbbm{T}^3$ here
and let $\p_\mu$ be the standard partial derivatives.
Consider an element of the form
\begin{equation}\label{eq:t0d0}
     x = \underbrace{\theta_0\p_0 + \textstyle\sum_{i=1}^3\sum_{\mu=0}^3 e_i^\mu \theta_i \p_\mu}_{\text{frame}}
           + \underbrace{
                  \textstyle\sum_{i=0}^3 a_i \theta_i \sigma_i + \text{rest}
                }_{\text{connection}}
\end{equation}
in $\E^1$ or $\E_\Phi^1$. Here `rest' is a linear combination
of the other basis elements\footnote{These are $\theta_i \sigma_j$ with $i \neq j$
and $\theta_i \sigma_{jk}$ and perhaps a scalar field piece.}.
Note that $x$ is nondegenerate iff $m = \det((e_i^\mu)_{i=1,2,3,\mu=1,2,3})$
is nonzero. We have
\[
[x,x] = 0
\qquad
\Longrightarrow
\qquad
      \p_0 m = (-3a_0 + a_1+a_2+a_3)m
\]
This can be checked for the Kasner element in Remark
\ref{remark:kasner2},
and in general follows from the vanishing of the coefficients
of $\theta_0 \theta_i \p_\mu$ in $[x,x] = 0$.
It follows that under these conditions, nondegeneracy is preserved
along the integral curves of $\p_0$.
\end{remark}
\section{The BKL filtration and its Rees algebra} \label{section:filtration}
In \cite{rtfil} a filtration of $\E$ is defined
that we call the BKL filtration.
For the reader's convenience, a complete but ad-hoc
definition is included here,
together with a new but straightforward extension to $\E_\Phi$.
For a conceptual treatment, see \cite{rtfil}.
\step
We first define an $\N^3$-grading of $\E_\Phi$.
The bracket does not respect this,
but it respects the corresponding non-decreasing filtration \eqref{eq:fusingg}.
The $\N^3$-grading is compatible with the native $\N$-grading of the gLa,
giving an $\N \times \N^3 \simeq \N^4$ grading.
\begin{definition}[The BKL filtration]\label{def:filrees}
Define an $\N^3$-grading $\E_\Phi = \bigoplus_\alpha G_\alpha \E_\Phi$
by free $\cinf$-modules $G_\alpha \E_\Phi = G_\alpha \E \oplus G_\alpha \Phi$
defined in Table \ref{table:EPhiG}. Set
\begin{align}
\label{eq:fusingg}
 \text{BKL filtration}:&&
  F_\alpha \E_\Phi & =\textstyle\bigoplus_{\beta \leq \alpha} G_\beta \E_\Phi
\intertext{%
It is non-decreasing,
and $F_\alpha \E_\Phi = 0$ if $\alpha \not\geq 0$.
Let $\spar$ be a vector of three symbols
and set $\spar^\alpha = \spar_1^{\alpha_1} \spar_2^{\alpha_2} \spar_3^{\alpha_3}$,
$\alpha \in \N^3$.
Set, with summation as formal power series,
}
\text{Rees algebra}:&&
\P & =
\{ \textstyle\sum_\alpha x_\alpha \spar^\alpha \mid x_\alpha \in F_\alpha \E_\Phi
\}\\
\text{Associated graded algebra}:&&
\A & = \P/\spar
\end{align}
They are graded Lie algebras, by Lemma \ref{lemma:fil} below.
\end{definition}
 \begin{table}
   \footnotesize{
   \[
   \begin{array}{|l|l|l|}
     \hline
     \alpha & \text{Basis for $G_\alpha \E_\GB$.}
     & \text{Basis for $G_\alpha \Phi_\GB$. By definition} \\
     & \text{By definition, $G_\alpha\E = G_\alpha \E_\GB \oplus \theta_0 G_\alpha \E_\GB$.}
     & G_\alpha\Phi = G_\alpha \Phi_\GB \oplus \theta_0 G_\alpha \Phi_\GB.\\
     \hline
     \hline
     000 & \Der(\cinf),\spx\sigma_0,\spx\theta_0 \sigma_0 + \theta_1\sigma_1,\spx
             \theta_0\sigma_0 + \theta_2\sigma_2,\spx \theta_0\sigma_0 + \theta_3\sigma_3,
         & \inphi{\theta_0} \\
             & 
             \theta_2\theta_3\sigma_{23}+\theta_3\theta_1\sigma_{31}+\theta_1\theta_2\sigma_{12}
             + 2\theta_0\theta_1\sigma_1+2\theta_0\theta_2\sigma_2+2\theta_0\theta_3\sigma_3 & \\
     \hline
     200 & -\theta_1\sigma_{23} + \theta_2\sigma_{31} + \theta_3\sigma_{12}
         & \text{none}\\
     \hline
     020 & +\theta_1\sigma_{23} - \theta_2\sigma_{31} + \theta_3\sigma_{12}
         & \text{none}\\     
     \hline
     002 & +\theta_1\sigma_{23} + \theta_2\sigma_{31} - \theta_3\sigma_{12}
         & \text{none}\\
     \hline
     011 & \sigma_1,\spx
             \sigma_{23},\spx
             \theta_1 \Der(\cinf),\spx
             \theta_0\sigma_1+\theta_1\sigma_0,\spx
             \theta_2\sigma_3+\theta_3\sigma_2,
         & \inphi{\theta_1} \\
           & \theta_3\sigma_{31},\spx
             \theta_2\sigma_{12},\spx
             \theta_0\theta_2\sigma_{12}+\theta_1\theta_2\sigma_2 & \\
             \hline
     101 & \sigma_2,\spx
             \sigma_{31},\spx
             \theta_2 \Der(\cinf),\spx
             \theta_0\sigma_2+\theta_2\sigma_0,\spx
             \theta_3\sigma_1+\theta_1\sigma_3,
         & \inphi{\theta_2} \\
           & \theta_1\sigma_{12},\spx
             \theta_3\sigma_{23},\spx
             \theta_0\theta_3\sigma_{23}+\theta_2\theta_3\sigma_3 & \\
             \hline
     110 & \sigma_3,\spx
             \sigma_{12},\spx
             \theta_3 \Der(\cinf),\spx
             \theta_0\sigma_3+\theta_3\sigma_0,\spx
             \theta_1\sigma_2+\theta_2\sigma_1,
         & \inphi{\theta_3}\\
           & \theta_2\sigma_{23},\spx
             \theta_1\sigma_{31},\spx
             \theta_0\theta_1\sigma_{31}+\theta_3\theta_1\sigma_1 & \\
             \hline
     211 & \theta_2\sigma_3 - \theta_3\sigma_2,\spx
             \theta_2\theta_3 \Der(\cinf),
         & \inphi{\theta_2\theta_3} \\
           & \theta_0\theta_2\sigma_3-\theta_0\theta_3\sigma_2-2\theta_2\theta_3\sigma_0,\spx
             \theta_0\theta_2\sigma_3+\theta_0\theta_3\sigma_2-2\theta_1\theta_2\sigma_{31} & \\
             \hline
     121 & \theta_3\sigma_1 - \theta_1\sigma_3,\spx
             \theta_3\theta_1 \Der(\cinf),
         & \inphi{\theta_3\theta_1}\\
           & \theta_0\theta_3\sigma_1-\theta_0\theta_1\sigma_3-2\theta_3\theta_1\sigma_0,\spx
             \theta_0\theta_3\sigma_1+\theta_0\theta_1\sigma_3-2\theta_2\theta_3\sigma_{12} & \\
             \hline
     112 & \theta_1\sigma_2 - \theta_2\sigma_1,\spx
             \theta_1\theta_2 \Der(\cinf),
         & \inphi{\theta_1\theta_2}\\
           & \theta_0\theta_1\sigma_2-\theta_0\theta_2\sigma_1-2\theta_1\theta_2\sigma_0,\spx
             \theta_0\theta_1\sigma_2+\theta_0\theta_2\sigma_1-2\theta_3\theta_1\sigma_{23} & \\
             \hline
     222 & \theta_0\theta_1\sigma_{23}+\theta_0\theta_2\sigma_{31}
             + \theta_0\theta_3\sigma_{12}
             - 2 \theta_2\theta_3\sigma_1
             - 2 \theta_3\theta_1\sigma_2
             - 2 \theta_1\theta_2\sigma_3,
         & \inphi{\theta_1\theta_2\theta_3} \\
           & \theta_1\theta_2\theta_3 \Der(\cinf),\spx
            \theta_0\theta_2\theta_3\sigma_1
            +\theta_0\theta_3\theta_1\sigma_2
            +\theta_0\theta_1\theta_2\sigma_3
            +3\theta_1\theta_2\theta_3\sigma_0 &
     \\
     \hline
     \text{else} & \text{none} & \text{none}\\
     \hline
 \end{array}
 \]}
 \caption{
 Definition of $G_\alpha\E_\Phi = G_\alpha \E \oplus G_\alpha \Phi$.
 All elements are in $\E$ via $\mathcal{L} \to \E$.
 In this paper we often refer to the elements in this table and
 this often implicitly includes $\theta_0$ times these elements; see the table heading.
 This is always clear from context. In particular, the elements in the $\E$-column
 (resp.~$\Phi$-column)
 and $\theta_0$ times these elements is a basis for $\E$ (resp.~$\Phi$).
 Note that the definition of $G_\alpha \E_\Phi$ is such that permuting the components of $\alpha$
 is equivalent to applying an automorphism of $W$
 that permutes the  $\theta_1$, $\theta_2$, $\theta_3$.
 }\label{table:EPhiG}
 \end{table}
\begin{lemma}\label{lemma:fil}
The filtration
$F_\alpha \E_\Phi$
is a graded Lie algebra filtration,
\[
[F_\alpha\E_\Phi,F_\beta\E_\Phi] \subset F_{\alpha+\beta}\E_\Phi
\]
Therefore
$\P$ is a gLa over $\R[[\spar]]$, namely a sub gLa of
$\E_\Phi[[\spar]] = \{ \sum_\alpha x_\alpha \spar^\alpha \mid x_\alpha \in \E_\Phi\}$.
Accordingly its associated graded $\A$ is a gLa over $\R$.
\end{lemma}
\begin{proof}
For $\E$ see \cite{rtfil}. The extension to $\E_\Phi$ is by direct calculation.
\qed\end{proof}
We use the notation
\[
		\A^i_\alpha \subset \A
\]
for elements of homological degree $i \in \N$ and filtration degree $\alpha \in \N^3$,
therefore $\A=\bigoplus_{\alpha,i}\A_{\alpha}^{i}$ 
and $[\A_{\alpha}^i,\A_{\beta}^j] \subset \A_{\alpha+\beta}^{i+j}$.
The additional subscript in $\A^i_{\alpha,\GB}$
refers to the span of the elements in Table \ref{table:EPhiG}.
We write
\[
\textstyle\A_{>0} = \bigoplus_{\alpha \neq 0} \A_{\alpha}
\]
which is a nilpotent subalgebra.
Notation such as $\A^i_{>0}$ and $\A^i_{\GB}$ is analogous.
\begin{remark}\label{remark:abbs}
The surjection $\P \to \A$ has a unique right-inverse
$\A \hookrightarrow \P$
with image 
\[
\{ \textstyle\sum_\alpha x_\alpha \spar^\alpha \mid
x_\alpha \in G_{\alpha} \E_{\Phi}\}
\]
This is not a gLa morphism.
But it induces an isomorphism 
$\P \simeq \A[[\spar]]$
of finite free $\cinf[[\spar]]$-modules
through which
the bracket on $\P$ is the $\R[[\spar]]$-bilinear extension of a map
$\A \times \A \to \A[\spar]$\footnote{%
Roughly, in $\E_\Phi$ the bracket of something in \smash{$G_{\alpha}$}
with something in \smash{$G_{\beta}$} is a sum of terms in
\smash{$G_{\gamma}$} with $\gamma \leq \alpha + \beta$,
and each summand is given a factor \smash{$\spar^{\alpha+\beta - \gamma}$} 
for the purpose of $\A \times \A \to \A[\spar]$.
}; the bracket on $\A$
is obtained by setting $\spar = 0$.
\end{remark}

\section{$\MC(\A)$ elements, leading term} \label{sec:mca}

In this section $M^3$ is a parallelizable closed 3-manifold and
\[
\Omega = [0,\infty) \times M^3
\]
The coordinate on the first factor is denoted $\tau$,
and we implicitly use the injection $C^\infty(M^3) \hookrightarrow C^\infty(\Omega)$
that corresponds to functions that are independent of $\tau$.
In Section \ref{sec:mcp} we construct $\MC(\P)$ elements.
The leading term of such an element is naturally
its image under the canonical $\MC(\P) \to \MC(\A)$.
Hence it makes sense to begin by constructing elements in $\MC(\A)$.
We begin with examples.

\begin{lemma}[Spatially homogeneous elements] \label{lemma:she}
Suppose $D_1,D_2,D_3 \in \Gamma(TM^3)$ are vector fields that satisfy
$[D_{i+1},D_{i+2}] = c_i D_i$ for all $i \in \Z/3\Z$ for some $c_i \in \R$.
Suppose further that $\bklq,p_1,p_2,p_3 \in \R$ satisfy
$3 \bklq^2 =  p_2p_3 + p_3p_1 + p_1p_2$.
Then
\begin{equation*}
\begin{aligned}
  \theta_0 \p_0 \;+\;
   \bklq\inphi{\theta_0} \;+\;
\textstyle\sum_{(i,j,k) \in C}
(\;\;
  +\,& p_i (\theta_0 \sigma_0 + \theta_i\sigma_i)\\
  +\,& \smash{\tfrac12 \sv{i}^2 c_i (\theta_i \sigma_{jk} - \theta_j \sigma_{ki} - \theta_k \sigma_{ij})}\\
  +\,& \sv{j}\sv{k} \theta_i D_i
  )
\end{aligned}
\end{equation*}
is in $\MC(\A)$.
Here we abbreviate $\sv{i} = \spar_i e^{-p_i \tau }$,
and see \eqref{eq:cycl} for the definition of $C$.
\end{lemma}
\begin{proof}
By direct calculation.
\qed\end{proof}
The elements in Lemma \ref{lemma:she}
are similar to the Kasner elements in Remark \ref{remark:kasner2},
but they are in general not in $\MC(\E_\Phi)$.
The elements in Lemma \ref{lemma:she} with $p_1,p_2,p_3 > 0$
will play the role of the leading term at $\tau \to +\infty$
of solutions as in Figure \ref{fig:anisotropic2}.
\step
We now study MC elements that are not necessarily
homogeneous, in particular $p_1,p_2,p_3$ and the structure coefficients become functions of the spatial variables.
We first construct the degree zero part,
$\A_0 \simeq F_{000}\E_\Phi$.

\begin{lemma}[Naive leading term]\label{lemma:nlt}
The elements in
\[
\MC(\A_0) \cap (\theta_0\p_0 + \A^1_{0,\GB})
\]
are precisely the elements
\begin{equation}\label{eq:naive}
\usol_0 =
  \theta_0 \p_0 +
   \bklq\inphi{\theta_0} + \textstyle\sum_{i=1}^3 p_i (\theta_0 \sigma_0 + \theta_i\sigma_i)
\end{equation}
with $(\bklq,p_1,p_2,p_3) \in C^\infty(M^3)$ and $3 \bklq^2 = 
p_2p_3 + p_3p_1 + p_1p_2$.
Below we continue to abbreviate
$\sv{i} = \spar_i e^{-p_i \tau }$ which are now functions of all variables.
\end{lemma}
\begin{proof}
Every element in $\theta_0\p_0 + \A^1_{0,\GB}$ is of the from \eqref{eq:naive}
but with the four coefficient functions taken in $C^\infty(\Omega)$.
We must check that the Maurer Cartan equation
holds iff the four functions are in $C^\infty(M^3)$ and satisfy the quadratic equation.
This follows by direct evaluation of the bracket in $\A_0$, for instance
\[
[\theta_0\p_0,f (\theta_0 \sigma_0 + \theta_i\sigma_i)]
 = (\p_0 f) \theta_0\theta_i \sigma_i
\]
for all $f \in C^\infty(\Omega)$, which is zero iff $f \in C^\infty(M^3)$.
\qed\end{proof}
\begin{remark}\label{rem:framerank0}
The restriction to $\theta_0\p_0 + \A^1_{0,\GB}$ is a gauge choice.
It is not difficult to see
that every element of $\MC(\A_0)$
whose frame has rank one 
(this is the maximal possible in $\A_0$ 
because the first row of Table \ref{table:EPhiG} only allows 
$\theta_0\Der(\cinf)$ frame components)
is locally of the form \eqref{eq:naive}
after a gauge transformation, in particular by choosing coordinates
such that $\p_0$ corresponds to that frame direction.
Below we prove such a statement at the formal power series level,
for the affine gauge subspace defined next.
\end{remark}
\begin{definition}[Affine gauge subspace] \label{def:agsub}
Define $\Aspec \subset \A^1_\GB \subset \A^1$ by the basis in Table \ref{table:aspecial}.
In particular $\rank \Aspec = 28$ versus $\rank \A^1_\GB = 37$. Set
\[
\Aspecaffine \;=\; \theta_0 \p_0 + \Aspec
\]
\end{definition}
Note that $\A^1_\GB$ has corank $11$ which is equal to the rank of the gauge Lie algebra $\A^0$,
so informally, membership in an affine shift of $\A^1_\GB$ ought to be
a gauge condition that can always be realized.
By contrast, membership in $\Aspecaffine$ is more restrictive and seems too much.
We will nevertheless see that there are enough residual gauge transformations
to realize this gauge for anisotropic MC elements.

 \begin{table}
   \footnotesize{
   \[
   \begingroup
   \renewcommand{\arraystretch}{1.1}
   \begin{array}{|c|l|}
     \hline
     \text{factor} & \text{$\Aspec \subset \A^1_{\GB} \subset \A^1$ 
                     basis elements over $\cinf(\Omega)$}\\
     \hline
     1 &   \theta_0 \sigma_0 + \theta_1\sigma_1,\spx
             \theta_0\sigma_0 + \theta_2\sigma_2,\spx
             \theta_0\sigma_0 + \theta_3\sigma_3,\spx
             \inphi{\theta_0} \\ 
     \spar_1^2 & -\theta_1\sigma_{23} + \theta_2\sigma_{31} + \theta_3\sigma_{12} \\ 
     \spar_2^2 & +\theta_1\sigma_{23} - \theta_2\sigma_{31} + \theta_3\sigma_{12} \\ 
     \spar_3^2 & +\theta_1\sigma_{23} + \theta_2\sigma_{31} - \theta_3\sigma_{12} \\ 
     \spar_2\spar_3 &   \theta_1 L_{1,2,3},\spx
             \theta_0\sigma_1+\theta_1\sigma_0,\spx
             \theta_3\sigma_{31},\spx
             \theta_2\sigma_{12},\spx
             \inphi{\theta_1} \\ 
     \spar_3\spar_1 &   \theta_2 L_{1,2,3},\spx
             \theta_0\sigma_2+\theta_2\sigma_0,\spx
             \theta_1\sigma_{12},\spx
             \theta_3\sigma_{23},\spx
             \inphi{\theta_2} \\ 
     \spar_1\spar_2 &   \theta_3 L_{1,2,3},\spx
             \theta_0\sigma_3+\theta_3\sigma_0,\spx
             \theta_2\sigma_{23},\spx
             \theta_1\sigma_{31},\spx
             \inphi{\theta_3} \\
\hline
 \end{array}
 \endgroup
 \]}
 \caption{%
   Basis for $\Aspec$.
   For $L_1,L_2,L_3$ see Section \ref{section:gla}.
 }\label{table:aspecial}
 \end{table}
\begin{definition}[Anisotropic element]\label{def:anisotropic}
We say that $\usol_0$ is anisotropic if and only if the functions $p_1,p_2,p_3$
are pairwise different at every point of $M^3$,
\[
p_1 \neq p_2 \neq p_3 \neq p_1
\]
\end{definition}
\begin{lemma}[Gauging of anisotropic elements via nilpotent automorphisms]\label{lemma:gae}
Set
$\A[[t]]^\times=\A_{>0} + t \A[[t]]$ 
and 
$\Aspec[[t]]^\times
= \A[[t]]^\times \cap \Aspec[[t]]$
and $I_n = t^{n+1} \A[[t]]$.
Given an anisotropic MC element $\usol_0$ as in \eqref{eq:naive}, 
then for all $n\ge0$
the canonical map
\begin{equation}\label{eq:surjaXXX}
\MC(\A[[t]]/I_n)\cap ((\usol_0+\Aspec[[t]]^\times)/I_n)
\to
\frac{\MC(\A[[t]]/I_n)\cap ((\usol_0 +\A[[t]]^\times)/I_n)}{\exp(\A^0[[t]]^\times/I_n)}
\end{equation}
is surjective, where the denominator is a group using Baker-Campbell-Hausdorff.
Furthermore for $n\ge1$, if
$x_0+tx_1+\dots+t^nx_n\in \MC(\A[[t]]/I_n)\cap ((\usol_0 +\A[[t]]^\times)/I_n)$
with $x_0+\dots+t^{n-1}x_{n-1}\in  \usol_0 +\Aspec[[t]]^\times$,
then the surjection \eqref{eq:surjaXXX} is realized by
an element in
 $\exp( t^n\A^0[[t]]/I_n)$.
Analogous if $t$ stands for several formal parameters.
\end{lemma}
\begin{proof}
We first prove the simplified statement with
$\Aspec$ replaced by $\A^1_{\GB}$.
This will not require anisotropy or membership in $\MC$.
For every $\alpha \in \N^3$ define a differential
\[
    d_\alpha \in \End^1(\A_\alpha)
\]
by $d_\alpha = [\usol_0,-]$. 
Define the composition (this is adapted from \cite{rtgla})
\[
K_\alpha^i: \A_{\alpha,\GB}^i \to \A_\alpha^i \xrightarrow{\;d_\alpha\;}
            \A_\alpha^{i+1} \to \A_\alpha^{i+1}/\A_{\alpha,\GB}^{i+1}
\]
The space on the right is $\simeq \A_{\alpha,\GB}^i$
using $\A_{\alpha}^{i+1} = \A_{\alpha,\GB}^{i+1} \oplus \theta_0 \A_{\alpha,\GB}^i$.
Then $K_\alpha^i$ is of the form $\mathbbm{1} \p_0$ plus terms without derivatives,
and therefore it is surjective.
This implies the simplified statement by standard arguments.

We now prove the lemma for $n=0$,
imitating
the section on nilpotent automorphisms in \cite{rtfil}.
(The case $n>0$ is analogous.)

Given $x \in \MC(\A) \cap (\usol_0 + \A_{>0})$.
By the simplified statement, we may assume $x \in \theta_0\p_0 + \A^1_\GB$. 
Decompose by degree, $x = \bigoplus_{\alpha} x_{\alpha}$.
So $x_{000} = \usol_0$ and $x_{\alpha} \in \A^1_{\alpha,\GB}$ when $\alpha \neq 0$. 
Consider $x_{011}$.
This may contain two terms
\begin{equation}\label{eq:xxx}
f\spar_2\spar_3(\theta_2 \sigma_3 + \theta_3\sigma_2) + g\sv{2}\sv{3}\theta_1 \p_0
\;\in\;\A^1_{011,\GB}
\end{equation}
for some coefficients $f,g \in C^\infty(\Omega)$.
These are the terms that fail to be in $\Aspec$.
However $x \in \MC(\A)$ implies $[\usol_0,x_{011}] = 0$ and this implies,
by evaluating the bracket, that $f,g \in C^\infty(M^3)$\footnote{%
Actually a weaker observation suffices for the following argument:
$f$ and $g$ satisfy homogeneous equations and are zero for all $\tau$ iff
they are zero at $\tau = 0$.}.
We want to remove such term by exploiting elements in the kernel of $K_{011}^0$,
which correspond to residual gauge freedom.
Concretely,  
\[
F \spar_2\spar_3 \sigma_{23} +
G \sv{2}\sv{3} \sigma_1
\;\in\;
\A^0_{011,\GB}
\]
is in the kernel of $K_{011}^0$ for all $F,G \in C^\infty(M^3)$.
By anisotropy one can choose $F,G$ such that $d_{011} = [\usol_0,-]$ applied to this element
yields the negative of \eqref{eq:xxx}\footnote{%
Actually it suffices
to do this at $\tau = 0$,
since then it will automatically be true for all $\tau$.}\textsuperscript{,}\footnote{%
That the anisotropy is needed can be seen from the following brackets in $\A$:
\begin{align*}
[\spar_j\spar_k \sigma_{jk},
 p_j(\theta_0\sigma_0 + \theta_j\sigma_j)]
 & = -\spar_j\spar_k f_ip_j (\theta_j\sigma_k + \theta_k\sigma_j)\\
[\spar_j\spar_k \sigma_{jk},
 p_k(\theta_0\sigma_0 + \theta_k\sigma_k)]
 & = +\spar_j\spar_k f_ip_k (\theta_j\sigma_k + \theta_k\sigma_j)
\end{align*}}.
By an analogous argument for $x_{101}$ and $x_{110}$,
there is an automorphism in $\exp(\A^0_{>0})$ that applied to $x$
yields a new element $x$ with
$x_{011},\, x_{101},\, x_{110} \in \Aspec$.
The terms $x_{200},\, x_{020},\, x_{002} \in \A^1_\GB$
are automatically in $\Aspec$.
Finally, $x_{211} = x_{121} = x_{112} = 0$
follows by direct calculation using $x \in \MC(\A)$\footnote{
See the lemma about synchronous frames in
\cite{rtfil}.
Considering
the coefficients of $\spar_j\spar_k\spar_i^2 \theta_j\theta_k \p_0$.
}.
So \eqref{eq:surjaXXX} is surjective for $n=0$.
\qed\end{proof}

\begin{definition} \label{def:mcsnew}
\begin{align*}
\mcs \;& =\;
  \MC(\A)
     \;\cap\; \Aspecaffine \;\cap\; \textnormal{nondegenerate} \;\cap\; \textnormal{exact}
\end{align*}
Here we use the following terminology for elements of $\A^1$:
\begin{itemize}
\item Nondegenerate:
Direct analog of nondegeneracy in Definition \ref{def:nondeg}.
Equivalently,
nondegenerate as an element of $\E^1_\Phi$
by setting $\spar_1 = \spar_2 = \spar_3=1$.
\item Exact:
The direct analogs of the one-forms in Definition \ref{def:aof1} and \ref{def:aof2}
are exact;
they are closed so this means that their integrals over the $1$-cycles on $M^3$
are zero.
\end{itemize}
\end{definition}
\begin{definition}[Structure functions]\label{def:scoeff}
If $D_1,D_2,D_3 \in \Gamma(TM^3)$ constitute a frame,
their structure functions are the $c_{ij}^k \in C^\infty(M^3)$
defined by
\[
[D_i,D_j] = \textstyle\sum_k c_{ij}^k D_k
\]
Here the bracket is the ordinary commutator of vector fields (derivations).
\end{definition}
\begin{definition}[Constraint equation]
For a tuple
\begin{equation}\label{eq:ttuu}
   (D_1,D_2,D_3,p_1,p_2,p_3,\bklq,\scl,\phiconstraint)
\end{equation}
with $D_1,D_2,D_3 \in \Gamma(TM^3)$ that constitute a frame,
and $p_1,p_2,p_3,\bklq,\scl,\phiconstraint \in \cinf(M^3)$,
the constraint equations are, by definition, the equations
\begin{equation}\label{eq:consmca}
\begin{aligned}
0 & = 3 \bklq^2 - p_2p_3 - p_3p_1 - p_1p_2\\
0 & = -\tfrac12 D_i(p_j+p_k) 
- \tfrac12 c_{ij}^j (p_i-p_j)
+ \tfrac12 c_{ki}^k (p_i-p_k)
+ p_i D_i(\scl)
+ 3\bklq D_i(\phiconstraint)
\end{aligned}
\end{equation}
where $(i,j,k) \in C$.
Beware that we do not use the summation convention.
\end{definition}
See Appendix \ref{app:constraints}
for the constraint equations, in the case 
$M^3 =  \mathbbm{T}^3$.\footnote{%
The first equation in \eqref{eq:consmca} is used
to eliminate $\bklq$.
The remaining three equations are
analyzed perturbatively 
near anisotropic spatially homogeneous elements
and give an elliptic system.
Appealing to an implicit function theorem the solution
space is a good intersection of a graph with three quadrics.}
\begin{lemma}[All $\mcs$ elements]\label{lemma:mcael}
For every tuple \eqref{eq:ttuu},
the element
\begin{equation}\label{eq:mcysw}
\begin{aligned}
  \theta_0 \p_0 \;+\;
   \bklq\inphi{\theta_0} \;+\;
\textstyle\sum_{(i,j,k) \in C}
(\;\;
  +\,& p_i (\theta_0 \sigma_0 + \theta_i\sigma_i)\\
  +\,& \tfrac12 \sv{i}^2 c_{jk}^i (\theta_i \sigma_{jk} - \theta_j \sigma_{ki} - \theta_k \sigma_{ij})\\
  +\,& \sv{j}\sv{k} \theta_i D_i\\
  -\,& \sv{j}\sv{k} (D_i(\scl) - \tau D_i(p_1+p_2+p_3))(\theta_0\sigma_i + \theta_i\sigma_0)\\
  +\,& \sv{j}\sv{k} (D_i(\scl) + c_{ki}^k
                     -\tau D_i(p_k)) \theta_k \sigma_{ki}\\
  -\,& \sv{j}\sv{k} ( D_i(\scl) - c_{ij}^j
                     -\tau D_i(p_j)) \theta_j \sigma_{ij}\\
  +\,& \sv{j}\sv{k} (D_i(\phiconstraint) + \tau D_i(\bklq)) \inphi{\theta_i}
  )
\end{aligned}
\end{equation}
is in $\mcs$ if and only if the constraints \eqref{eq:consmca} hold.
Here $\sv{i} = \spar_i e^{-p_i \tau }$,
and see \eqref{eq:cycl} for the definition of $C$.
Every element $\uzero \in \mcs$ is of this form
for a unique tuple \eqref{eq:ttuu} up to
additive constants for $\scl$ and $\phiconstraint$.
\end{lemma}
\begin{proof}
The degree zero part is by Lemma \ref{lemma:nlt}.
Use Table \ref{table:aspecial} to
make a general ansatz for the remaining degrees.
We only write out some of the remaining terms:
\[
\uzero =
\usol_0 
+ \ldots + \textstyle\sum_{i=1}^3 \spar_2\spar_3 c_i\theta_1  L_i
+ \ldots + \spar_2\spar_3 d_1\theta_2 \sigma_{12} + \ldots
\]
with $\ldots,c_1,c_2,c_3,\ldots, d_1, \ldots \in  \cinf(\Omega)$.
Evaluate
the bracket $\A^1 \times \A^1 \to \A^2$
using \eqref{eq:brak};
the ideal $\I^2$ in \eqref{eq:i2};
the filtration in Definition \ref{def:filrees}.
 For example,
\begin{align*}
[\theta_0\p_0, \spar_2\spar_3 c_i \theta_1 L_i]
& = \spar_2\spar_3 (\p_0 c_i) \theta_0\theta_1 L_i\\
[p_2 (\theta_2\sigma_2 + \theta_0\sigma_0),
 \spar_2\spar_3 c_i \theta_1 L_i] & = 
\spar_2\spar_3 p_2 c_i \theta_0\theta_1 L_i +
\spar_2\spar_3 c_i  L_i(p_2) (\theta_1\theta_2\sigma_2 + \theta_1\theta_0\sigma_0)
\end{align*}
Actually other brackets contribute to  $\spar_2\spar_3\theta_0\theta_1 L_i$
as well, and one has to write out all.
\begin{itemize}
\item Consider
$[\uzero,\uzero] = 0 \bmod \A^2_\GB$.
The vanishing of the coefficients of $\spar_2\spar_3\theta_0\theta_1 L_i$
require $\p_0 c_i + (p_2 + p_3) c_i = 0$,
so $\sum_i c_i L_i = e^{-(p_2+p_3)\tau} D_1$ for some $D_1 \in \Gamma(TM^3)$.
The vanishing of the coefficient of
$\spar_2\spar_3\theta_0\theta_2\sigma_{12} =
-\spar_2\spar_3 \theta_1\theta_2 \sigma_2 \bmod \A^2_\GB$
requires
$\p_0 d_1 + (p_2+p_3)d_1 = D_1(p_2)$
so $d_1 = \smash{e^{-(p_2+p_3)\tau}} (e_1 + \tau D_1(p_2))$
for some function $e_1 \in \cinf(M^3)$.
Continuing,
all coefficients in $\uzero$ are expressed in terms of
smooth functions on $M^3$.
Nondegeneracy requires that $D_1,D_2,D_3$ be a frame.
\item The remaining equations in $[\uzero,\uzero] = 0$
imply that
$\smash{(e_1-c_{12}^2,e_2-c_{23}^3,e_3-c_{31}^1)}$
are the components, relative to $D_1,D_2,D_3$,
of a closed one-form,
so $e_i = -D_i(\scl) + \smash{c_{ij}^j}$
for some function $\scl$ on the universal cover of $M^3$,
unique up to a constant.
Exactness in $\mcs$ requires
$\scl \in C^\infty(M^3)$.
Similarly, \eqref{eq:consmca} are forced.
\end{itemize}
All these steps are forced.
\qed\end{proof}

\section{$\MC(\P)$ elements, formal solutions} \label{sec:mcp}

The setting is that of Section \ref{sec:mca}.
Denote by $\econst \subset \E_\Phi$ 
the subspace of elements that are linear combinations
of the basis elements of $\E_\Phi$
(the entries in Table \ref{table:EPhiG}
and $\theta_0$ times these entries)
but with coefficients restricted to $\cinf(M^3)$.
It is not a subalgebra.
Accordingly $\epoly \subset \E_\Phi$
are the polynomials in $\tau$ with coefficients in $\econst$.
\begin{lemma}[A subalgebra of $\P$] \label{lemma:pp}
Given $p_1,p_2,p_3 \in \cinf(M^3)$ abbreviate $\sv{i} = \spar_i e^{-p_i \tau}$.
Then the following is a subalgebra of $\P$ over $\R$ (but not over $\R[[\spar]]$):
\begin{align*}
\P_p & = \{ \textstyle\sum_\alpha x_\alpha \sv{}^\alpha
\mid x_\alpha \in F_\alpha \E_\Phi \cap \epoly
\}
\end{align*}
\end{lemma}
\begin{proof}
This is immediate since
$[\sv{}^\alpha \epoly, \sv{}^\beta \epoly] \subset \sv{}^{\alpha + \beta} \epoly$.
\qed\end{proof}

In this subalgebra, every power of $\spar_i$ comes
with a power of $e^{-p_i\tau}$ which decays as $\tau \to \infty$ if $p_i > 0$;
only $\alpha = 0$ comes without such exponential damping.
\begin{definition}\label{def:mcsp}
Let
\[
\mcsp \subset \mcs
\]
be the elements where the naive leading term satisfies $p_1,p_2,p_3 > 0$.
\end{definition}
\begin{theorem}[Existence and uniqueness] \label{thm:ex}
Let $\uzero \in \mcsp$ as in Lemma \ref{lemma:mcael}
and $\P_p \subset \P$ the corresponding subalgebra.
Then there exists a
\[
\usol = \textstyle\sum_\alpha \usol_\alpha \sv{}^\alpha
                          \in \MC(\P_p)
\]
with $\usol \bmod \sv{} \P_p = \uzero$;
this means $\uzero$ is the leading term.
The element $\usol$
is unique up to gauge transformations in $\exp(\sv{} \P_p^0)$.
The assignment $\uzero \mapsto \usol$ can be realized as a map
\begin{equation}\label{eq:exkhke}
S_{\textnormal{formal}}\;:\;\mcsp \to \MC(\P)
\end{equation}
such that
\begin{enumerate}[({g}1),leftmargin=8mm]
\item \label{concg2}
$\usol_\alpha$ with $\alpha \neq 0$
are linear combinations of only the $F_\alpha \E^1_{\Phi,\GB}$ elements.
These are the entries in Table \ref{table:EPhiG}
but not $\theta_0$ times these entries.
\item \label{concg1}
Each coefficient of the linear combination in \ref{concg2}
is a polynomial in $\tau$
with coefficients 
that are polynomials
in (i) the
entries of the tuple \eqref{eq:ttuu}
and their derivatives of orders $\leq |\alpha|$;
(ii) the structure coefficients (Definition \ref{def:scoeff})
and their derivatives of orders $\leq |\alpha|$;
 (iii) $(n_1p_1 + n_2p_2 + n_3p_3)^{-1}$ with $n \in \N^3-0$.
\item \label{concg3} $\usol_\alpha = 0$ when $|\alpha|$ is odd.
\end{enumerate}
\end{theorem}
\begin{remark} \label{remark:sfr}
The map $S_{\text{formal}}$
 is the restriction of a smooth map
 \[
    \smash{\Aspecaffine^+} \to \P
 \]
 that is right-inverse to $\P \to \A$.
To every element in $\smash{\Aspecaffine}$
associate a tuple \eqref{eq:ttuu}.
For example, $p_i \in C^\infty(\Omega)$ are
 defined to be the coefficients of $\theta_0 \sigma_0 + \theta_i \sigma_i$.
 By definition, 
 $\smash{\Aspecaffine^+} \subset \Aspecaffine$
  is the subset where $p_1,p_2,p_3 > 0$ 
  and where $D_1,D_2,D_3$ are linearly independent everywhere on $\Omega$.
  In particular one can define the structure coefficients.
The polynomials in \ref{concg2} and \ref{concg1} give
 the desired map.
\end{remark}
\begin{proof}
Use \eqref{eq:fusingg} to decompose $\P_p = \sum_{n\beta} \P_{n\beta}$ where
\[
\P_{n\beta} = (G_\beta \E_\Phi \cap \epoly) \sv{}^{\beta + n}
\]
for $n,\beta \in \N^3$.
Let $B \subset \N^3$ be the finite subset of $\beta$ for which
$G_\beta \E_\Phi \neq 0$.
Throughout, the pairs $n\beta$ are in $\N^3 \times B$.
In addition we have
\[
[\P_{n\beta},\P_{n'\beta'}]
\subset \P_{n+n'\,
            \beta+\beta'}
+
\textstyle\sum_{\substack{n'' > n+n',\;
                \beta'' < \beta + \beta'\\
n + \beta + n' + \beta' = n'' + \beta''}}
\P_{n''\beta''}
\]
On $\N^3 \times B$ define
$n\beta < n'\beta'$ iff $(n < n')$ or $(n = n'\;\text{and}\;\beta < \beta')$,
which is a well-founded partial order; it is not total since the order on $\N^3$ is not.
It follows that
\begin{equation}\label{eq:ppp}
[\P_{\geq n\beta}, \P_{\geq n'\beta'}]
\subset
\P_{\geq n+n'\,\beta+\beta'}
\end{equation}
and similarly if on the left one replaces one $\geq$ by $>$,
and on the right one replaces $\geq$ by $>$.
We seek an MC-element
$\usol = \sum_{n\beta} \usol_{n\beta}$ with $\usol_{n\beta} \in \P_{n\beta}^1$.
The two notations are related by
$\usol_{\alpha} \sv{}^\alpha = \sum_{n + \beta = \alpha} \usol_{n\beta}$,
in particular $\usol_0 = \usol_{00}$.
Define the $\usol_{0\beta}$
by $\sum_\beta \usol_{0\beta} = \uzero$,
which one can since the pieces of $\uzero$ are in the right spaces.
Note that $\usol_0$ is given by \eqref{eq:naive}.
We will prove the following lemma:
\begin{quote}
For all $n\beta$ with $n \neq 0$,
and $y \in \P_{< n\beta}^1$ satisfying
$[y,y] \in \P_{ \not< n\beta}^2$\footnote{%
This means that in the decomposition of $[y,y]$
all components $< n\beta$ are zero.
In this sense, $y$ is an approximate MC-element.
},
there exists a $z \in \P_{n\beta}^1$
such that $x = y+z \in \P_{\leq n\beta}^1$
satisfies $[x,x] \in \P_{\not\leq n\beta}^2$.
\end{quote}
We first show that the Lemma implies the existence claim in Theorem \ref{thm:ex}.
Note that there exists an enumeration
$(n_i\beta_i)_{i \in \N}$ of $\N^3 \times B$ such that
$n_i\beta_i < n_j\beta_j$ implies $i < j$.
This is an appropriate order to construct the $\usol_{n_i\beta_i}$.
Namely, by induction on $N\in \N$ one shows:
There exist $\usol_{n_i\beta_i} \in \smash{\P_{n_i\beta_i}^1}$
for all $i \leq N$, with the $\usol_{0\beta}$ fixed in terms of $\uzero$,
such that $y = \sum_{i \leq N} \usol_{n_i\beta_i}$ satisfies
$[y,y] \in \sum_{i > N} \smash{\P_{n_i\beta_i}^2}$.
The induction step follows from the Lemma
(which is only for $n \neq 0$
but this suffices since we have fixed the $\usol_{0\beta}$ already)
together with \eqref{eq:ppp}.
The existence claim in Theorem \ref{thm:ex}
follows by choosing
the $\usol_{n\beta}$ consistently
(logically this requires the axiom of dependent choice,
but it is not actually
required given the constructive nature of the rest of the proof).

To prove the lemma, given $n\beta$, introduce
a differential on $\P_{n\beta}$ by
\[
d_{n\beta} = [\usol_0,-]
\;\in\; \End^1(\underbrace{\P_{\not< n\beta}/\P_{\not \leq n\beta}}_{\simeq \P_{n\beta}})
\]
which is well-defined and a differential using \eqref{eq:ppp}.
The lemma
amounts to finding a $z$ such that $d_{n\beta} z = e$ where $e = -[y,y]/2$.
But the Jacobi identity $[y,[y,y]]=0$
together with $y-\usol_0 \in \P_{>00}$ implies $d_{n\beta} e = 0$,
hence it suffices to show that $\smash{H^2(d_{n\beta})} = 0$.
We will actually show $H(d_{n\beta}) = 0$,
always assuming $n \neq 0$ as in the lemma.
The fact that $\smash{H^1(d_{n\beta})}=0$ implies the uniqueness part of the claim
by standard arguments.

Decompose $\P_{n\beta} = 
\bigoplus_m \P_{n\beta m}$
where $\P_{n\beta m}
= (G_\beta \E_\Phi \cap \econst \tau^m ) \sv{}^{\beta + n}$.
In the following we view $d_{n\beta}$
as a block matrix, with blocks indexed by $m$.
The differential $d_{n\beta}$ is triangular in $m$,
meaning it cannot increase $m$ because:
\begin{itemize}
\item $\usol_0 \in G_0 \E_\Phi \cap \econst$ so the coefficients do not involve $\tau$.
\item $\usol_0$ contains no $L_1,L_2,L_3$
that could generate $\tau$ when hitting an $\sv{i}$\footnote{%
Due to the structure of the filtration in degree zero,
see Remark \ref{rem:framerank0}.
}\textsuperscript{,}%
\footnote{%
The frame $L_1,L_2,L_3$ acts by differentiation through the anchor map.
}.
\end{itemize}
It suffices to show that the diagonal blocks of $d_{n\beta}$,
denoted $d_{n\beta m} \in \End^1(\P_{n\beta m})$,
satisfy $H(d_{n\beta m}) = 0$;
the vanishing of the homologies of the diagonal blocks
implies the vanishing of the homology of the block triangular $d_{n\beta}$.

The factor $\tau^m$ does not affect the differential $d_{n\beta m}$;
in particular $\p_0 \tau^m = m \tau^{m-1}$ does not contribute
to the diagonal. Hence it suffices to show that 
$H(d_{n\beta 0}) = 0$.

In summary, it suffices to show that
\[
	d_{n\beta 0} = [\usol_0,-]\;\in\;
        \End^1( (G_\beta \E_\Phi \cap \econst) \sv{}^{\beta + n})
\]
has vanishing homology, if $n \neq 0$.
Dividing and multiplying by $\sv{}^{\beta + n}$,
we get a $\cinf(M^3)$-linear differential
on the finite free $\cinf(M^3)$-module
$C = G_{\beta} \E_\Phi \cap \econst$.
For $C$ a basis is given by the
$G_\beta\E_\Phi$-elements in Table \ref{table:EPhiG}
and $\theta_0$ times these elements.
Denote $K = (G_\beta \E_\GB \oplus G_\beta \Phi_\GB) \cap \econst$,
that is,
a basis of $K$ is given by the $G_\beta\E_\Phi$-elements in Table \ref{table:EPhiG}
but not $\theta_0$ times these elements.
We now use the following homological algebra lemma:
If $(C,d)$ is a complex and $K \subset C$
is a subspace such that the composition $w: K \to C \to^d C \to C/K$ is bijective,
then $C$ is exact,
in fact if $dx = 0$ then there exists a $y \in K$ such that $x = dy$.
This is because the composition
$h: C \to C/K \smash{\to^{w^{-1}}} K \to C$
is a homotopy, $dh + hd = \one$ and $h^2 = 0$ with $\image h \subset K$\footnote{%
Concretely, decomposing $C \simeq K \oplus C/K$,
the differential has the block form $d = (\begin{smallmatrix} a & b \\ w & c \end{smallmatrix})$
and $h = (\begin{smallmatrix} 0 & w^{-1} \\ 0 & 0 \end{smallmatrix})$.
Then $dh + hd = (\begin{smallmatrix} \one & aw^{-1} + w^{-1}c \\ 0 & \one \end{smallmatrix})$.
The upper right entry vanishes since $d^2 = 0$ implies $wa + cw = 0$.
}.

Note that $C = K \oplus \theta_0 K$ where $\theta_0: K \to C$ is injective,
hence $C/K \simeq K$ with a homological degree shift by one,
so we can regard $w \in \smash{\End_{\cinf(M^3)}^0(K)}$.
Checking bijectivity amounts to checking that $\det w \in \cinf(M^3)$ is nonzero
which amounts to a concrete calculation.
Since $w$ preserves homological degree,
it is a block direct sum of matrices $w = w^0 \oplus w^1 \oplus w^2 \oplus w^3$.
One finds the following values for $\det w^i$,
up to multiplicative constants in $\R^\times$,
and analogous for permutations of $\beta$:
\[
\arraycolsep=5pt
\begin{array}{c|c c c c}
\beta & \det w^0 & \det w^1 & \det w^2 & \det w^3\\
\hline
000 & A_n^5 & A_n^4 & A_n &  \\
200 &   & A_n &   &  \\
011 & A_nA_{n+011} & A_n^8A_{n+011} & A_n &  \\
211 &   & A_n & A_n^7 &  \\
222 &   &   & A_n & A_n^6
\end{array}
\]
Here $A_n = n_1p_1 + n_2p_2 + n_3p_3 \in \cinf(M^3)$,
nonzero when $n \neq 0$ since $p_1,p_2,p_3 > 0$.
The factors $A_{n+011}$ are also nonzero.
The index $n$ comes in through $\p_0 \sv{}^n = -A_n \sv{}^n$.
The determinant is easy to calculate
since for every $\beta$ the matrix $w^i$ is upper triangular relative
to some permutation of the basis elements.

Our proof of $H^2(d_{n\beta}) = 0$ 
delivers an algorithm for constructing $z$
in $d_{n\beta} z = e$, and it generates
$\usol_\alpha$ 
as polynomial objects of the kind
listed in Theorem \ref{thm:ex}.
Furthermore $z$ only uses $\GB$ basis elements, see the definition of $K$.
\qed\end{proof}

The construction 
is complete in the sense that one can get all $\MC(\P)$
elements, up to gauge transformations,
by composing $S_{\textnormal{formal}}$ with formal curves in $\MC(\A)$.
\begin{lemma}[Formal completeness for anisotropic elements] \label{lemma:fcomp}
Define the canonical projection $\pi: \MC(\P) \to \MC(\A)$.
Let $x \in \mcsp$
with anisotropic leading term, see Definition \ref{def:anisotropic}.
Then the map
\begin{equation}\label{eq:fc}
\MC(\A[[\spar]]) \cap (x + \spar \Aspec[[\spar]]) \;\to\;
\frac{\pi^{-1}(x)}{\exp(\spar \P^0)}
\end{equation}
induced by $S_{\textnormal{formal}}$
(we use Remark \ref{remark:sfr}) is surjective.
\end{lemma}
\begin{proof}[Sketch]
Lemma \ref{lemma:mcael} and Theorem \ref{thm:ex} hold,
by the same proofs, 
if the parameters such as the elements in the tuple \eqref{eq:ttuu}
are formal power series in a formal parameter $t$,
and we work modulo $t^n$ for some $n>0$.
In particular, the map in Remark \ref{remark:sfr} has the property that,
for all $x_0 = x$ and $x_1,\ldots,x_{n-1} \in \Aspec$ we have the implication
\begin{multline}\label{eq:uuuuu}
x_0 + tx_1 + \ldots + t^{n-1} x_{n-1} \in \MC(\A[[t]]/t^n)\\
\Longrightarrow
\qquad
S_{\textnormal{formal}}(x_0 + tx_1 + \ldots + t^{n-1} x_{n-1}) \in \MC(\P[[t]]/t^n)
\end{multline}
Analogous if $t$ stands for several formal parameters.
Hence in particular the map in Lemma
\ref{lemma:fcomp} is well-defined, one does get an element in $\pi^{-1}(x)$,
by setting $t = \spar$.

Use $\P \simeq \A[[\spar]]$ as in Remark \ref{remark:abbs}.
The two are different as gLa,
and it is implicit below which bracket is used.
Let $y = x + \sum_{\alpha > 0} y_\alpha \spar^\alpha \in \pi^{-1}(x)$
with $y_\alpha \in \A^1$ be given.
We must construct a corresponding element on the left hand side of \eqref{eq:fc}.
The proof is by induction, using the following lemma
for all $\alpha \in \N^3-0$:
If there exist
$x_\beta \in \Aspec$ and $a_\beta \in \A^0$ for all $0 < \beta < \alpha$
such that
\begin{align*}
\textstyle x + \sum_{0<\beta<\alpha} \spar^\beta x_\beta \;&\in\;
 \MC(\A[[\spar]]/\spar^{\not<\alpha})\\
\textstyle\exp(\sum_{0<\beta<\alpha} \spar^\beta a_\beta)
S_{\text{formal}}(x + \sum_{0<\beta<\alpha} \spar^\beta x_\beta)
\;& =\; y \mod \spar^{\not<\alpha}
\end{align*}
then there are $x_\alpha \in \Aspec$ and $a_\alpha \in \A^0$
such that the same two statements hold when we add $\spar^\alpha x_\alpha$
respectively $\spar^\alpha a_\alpha$ and
the two statements are modulo $\spar^{\not\leq\alpha}$.
Without loss of generality,
by redefining $y$, we may assume $a_\beta = 0$ for $0 < \beta < \alpha$.
First define $x_{\alpha} \in \A^1$
(not necessarily $x_\alpha \in \smash{\Aspec}$)
 uniquely by requiring
\begin{equation}\label{eq:yy}
y = S_{\text{formal}}(x + \textstyle\sum_{0<\beta<\alpha} \spar^\beta x_\beta)
     + \spar^\alpha x_\alpha \mod \spar^{\not\leq\alpha}
\end{equation}
Using $[y,y]=0$ one checks\footnote{%
Take the $\spar^\alpha$ coefficient of $[y,y] = 0$,
where $y$ is replaced by \eqref{eq:yy}
and where $S_{\text{formal}}$ is expanded about $x$.
This is of the form
$\sum_{\beta + \gamma = \alpha} [x_\beta,x_\gamma]_0 + R = 0$
where $[-,-]_0 : \A \times \A \to \A$ is the bracket on $\A$ and $R$ is a rest term.
We must show $R = 0$. 
Expand $S_{\text{formal}}(X) = \sum_{\alpha \geq 0} S_\alpha(X) \spar^\alpha$
with $S_\alpha: \Aspec^+ \to \A$
and where $S_0(X) = X$,
and $[X,Y] = \sum_{\alpha \geq 0} [X,Y]_\alpha \spar^\alpha$
with $[-,-]_\alpha : \A \times \A \to \A$.
From \eqref{eq:uuuuu} (with $t = \spar$) one derives identities
among the $S_\alpha$ and their derivatives and the $[-,-]_\alpha$ that imply that $R = 0$.}
$x + \sum_{0 < \beta \leq \alpha} \spar^\beta x_\beta
\in \MC(\A[[\spar]]/\spar^{\not\leq\alpha})$.
Lemma \ref{lemma:gae}
yields an $a_\alpha \in \A^0$ such that
$\exp(\spar^\alpha a_\alpha) \in \exp(\spar^\alpha \A^0[[\spar]] / \spar^{\not\leq\alpha})$
puts $x_\alpha \in \Aspec$.
\qed\end{proof}
\section{$\MC(\E_\Phi)$ elements, true big-bang solutions}
\label{section:maintruesolution}
We truncate
$\MC(\P)$ elements,
replace the formal parameters
$\spar_1$, $\spar_2$, $\spar_3$ by a small constant,
and use this as an approximate solution $\amc$
for Theorem \ref{theorem:asyts}
to construct $\MC(\E_\Phi)$ elements
with nondegenerate frame
asymptotic to $\amc$ as $\tau \to \infty$.
Invoking Lemma \ref{lemma:gsf}
one gets a smooth metric and
scalar field solving the Einstein equations \eqref{eq:esf}.
The theorem below is for $M^3 = \mathbbm{T}^3$
only because this simplifies the notation for partial derivatives,
as in Theorem \ref{theorem:asyts}.
The theorem can equivalently be stated for $\spar_1 = \spar_2 = \spar_3 = 1$,
 see Remark \ref{remark:s1} below.

\newcommand{\pmin}{\tc{scolor}{H}}
\newcommand{\zmin}{Z_{\textnormal{min}}}
\newcommand{\mmin}{J_{\textnormal{minREMOVE}}}
\newcommand{\kc}{\tc{hcolor}{k}}
\newcommand{\corr}{\tc{qcolor}{\omega}}
\newcommand{\qts}{\tc{scolor}{rREMOVE}}
\newcommand{\jc}{\tc{qcolor}{jREMOVE}}

\begin{theorem}[Semiglobal existence of $\MC(\E_\Phi)$ elements] \label{theorem:sendg}
There exists
a constant $c > 0$
such that the following is true.
Given four constants
\[
\pmin > 0
\qquad
\kc > 0
\qquad
J \geq \max\{2,c \kc/\pmin\}
\qquad
K > 0
\]
with $J \in 2\N$.
Set $\Omega = [0,\infty) \times \mathbbm{T}^3$.
 Then exists a constant $\Lambda>0$ such that for all:
\begin{itemize}
\item
$\mcsp$ elements given by a tuple as in \eqref{eq:ttuu} on $\mathbbm{T}^3$, that is
\[
   (D_1,D_2,D_3,p_1,p_2,p_3,\bklq,\scl,\phiconstraint)
\]
In particular
$D_1,D_2,D_3$ constitute a frame; $p_1,p_2,p_3 > 0$;
the constraints hold.
\item constants $\lambda > 0$
\end{itemize}
if\footnote{Here $L^\infty$ means $L^\infty(\mathbbm{T}^3)$
and $\alpha \in \N^3$.
}
\begin{enumerate}[({c}1),leftmargin=10mm]
\item \label{pminx} $p_i \geq \pmin$. 
\item \label{der3}
      $\|\p^\alpha(p_i,\bklq)\|_{L^\infty} \leq \kc$
      for $|\alpha| \leq 1$.
\item \label{gen3}
      $\|\p^\alpha(D_i,p_i,\bklq,\scl,\phiconstraint,c_{ij}^k)\|_{L^\infty} \leq K$
      for $|\alpha| \leq J + 7$.
\item \label{eps3}
      $0 < \lambda \leq \Lambda$.
\end{enumerate}
for $i=1,2,3$,
then denoting by 
$\usol = \textstyle\sum_\alpha \usol_\alpha \sv{}^\alpha \in \MC(\P_p)$
the corresponding $\MC(\P)$ element defined by Theorem \ref{thm:ex}
(see also Lemma \ref{lemma:pp}),
and denoting
\begin{itemize}
\item[] $\amc = \sum_{\alpha: |\alpha| \leq J}
  \usol_\alpha \sv{}^\alpha|_{\spar_1 = \spar_2 = \spar_3 = \lambda} \in \E_\Phi^1$
\end{itemize}
there exists a
$\corr \in \E_\Phi^1$ such that, in $\E_\Phi$,
\begin{align*}
[\amc + \corr, \amc + \corr] & = 0\\
\|\p^\alpha \corr\|_{L^\infty_x} &
   = \mathcal{O}(e^{-J \pmin \tau}) \;\text{as $\tau \to \infty$}
\end{align*}
and the frame associated to $\amc + \corr \in \MC(\E_\Phi)$ is nondegenerate.
The associated smooth spacetime metric and scalar field
(by Lemma \ref{lemma:gsf})
solve the Einstein equations \eqref{eq:esf}.
The level sets of $\tau$ are spacelike.
There are particle horizons and a curvature singularity at $\tau \to \infty$.
\end{theorem}
\begin{remark}
The assumptions of Theorem \ref{theorem:sendg} can be satisfied
for all \eqref{eq:ttuu}.
Let $\pmin$ be the biggest allowed by \ref{pminx};
let $\kc$ be the smallest allowed by \ref{der3};
set $J = \max\{2,c \kc/\pmin\}$;
let $K$ be the smallest allowed by \ref{gen3}.
Note that the latter depends on $J$.
\end{remark}
\begin{remark} \label{remark:s1}
To obtain a map $S: U \subset \mcsp \to \MC(\E_\Phi)$
as in the introduction, one uses
the following observations.
\begin{itemize}
\item 
Since $\A$ is $\N^3$-graded,
replacing $D_i$ by $a_ja_k D_i$
with $(i,j,k) \in C$
with $a_i > 0$ yields another element of $\MC(\A)$.
This can be used to state the theorem differently,
where in \ref{gen3} one requires $\|\p^\alpha D_i \|_{L^\infty} \leq K \lambda^2$,
and one sets $\spar_1 = \spar_2 = \spar_3 = 1$.
\item 
It is easy to prove a uniqueness result
for a fixed gauge $\ginj_1$.
One uses the uniqueness statement in Theorem \ref{theorem:sendg}.
\end{itemize}
\end{remark}
\begin{remark}
In the informal discussion in Section \ref{sec:discussion},
it is suggested that $\amc$ is only
a good approximation in the late/IR regime and not in the early/UV regime.
This perspective could be exploited in a finer analysis.
\end{remark}
\begin{proof}
The translation lemma, Lemma \ref{lemma:translation},
is used for $\E_\Phi$ without further notice.
Fix a symmetric hyperbolic gauge for $\E_\Phi$, see Lemma \ref{lemma:gsf}.
It is given by constant matrices,
for instance $\ginj_1 \in \Hom(\R^{\dimz_1}, \R^{\dimy_1})$ with
$\dimz_1 = 33 + 4$
and $\dimy_1 = 44 + 4$.
All constants produced in this proof (including in particular $c > 0$)
may depend on this gauge.
(An admissible condition will be imposed on $\ginj_1$ when discussing nondegeneracy.)

We will use the existence part of Theorem \ref{theorem:asyts},
using the parameters in Table \ref{table:parthmxxx}
that are chosen in due course.

\begin{table}[h]
\centering
\begin{tabular}{c|c}
Theorem \ref{theorem:asyts} & Theorem \ref{theorem:sendg} and this proof\\
\hline\hline
$q$, $Q$ & $q$, $Q$ \\
$z$, $Z$ & $z$, $Z$ \\
$b$ & $b$ \\
$\dimx$ & $3$\\
$\Omega$ & $\Omega$ \\
$\gx$ & $\E_\Phi$ (including gauge) \\
\hline
$\delta$ & $\delta$ \\
$\amc$ & $\amc$ \\
\hline\hline
$\fld$ & $\fld \in C^\infty(\Omega,\R^{33+4})$ \\
\eqref{pde2} & $[\amc+\corr,\amc+\corr] = 0$
where $\corr = \ginj_1 \fld \in C^\infty(\Omega,\R^{44+4})$
\end{tabular}
\caption{%
The rows up to and including $\amc$
are parameters used to invoke Theorem \ref{theorem:asyts}.
}
\label{table:parthmxxx}
\end{table}

There exist constants $q_0 > 1$ and $c_0 > 0$ such that,
for all $\usol_{0}$ of the form \eqref{eq:naive},
\begin{subequations}
\begin{align}
\label{eq:a12}
 q_0^{-1} \one & \leq A_i^0(\usol_{0}) \leq q_0 \one\\
\label{eq:a1234}
    \| L_i(\usol_{0}) \|_{L^\infty} & \leq
      c_0 \,\textstyle\sup_{|\alpha|\leq 1} \|\p^\alpha(p_i,\bklq)\|_{L^\infty}
\end{align}
and $i=1,2$.
Here $A_i = (A_i^\mu)_{\mu = 0\ldots 3}$ and $L_i$ are as in Theorem \ref{theorem:asyts}.
Here $A_i^0(\usol_{0}) = A_i^0(\theta_0\p_0)$
is a constant symmetric matrix,
and it is positive since
$\theta_0 \in W_+$
and by the second part of Lemma \ref{lemma:translation}.
By equation \eqref{eq:defli},
$L_i(-)$ only uses up to first derivatives.
Similarly, there exists $b_0 > 0$ such that,
for all $\usol_{0}$ of the form \eqref{eq:naive}\footnote{%
The motivation for the number $6$ will only become clear later.
It arises as $\dimx + 3 = 6$.}\textsuperscript{,}\footnote{%
On the left,
$\p^\alpha$ are derivatives
in $\dimx+1 = 4$ coordinates.
On the right,
in $\dimx = 3$ coordinates.
},
\begin{equation}\label{eq:alal}
\textstyle\max_{1 \leq |\alpha|\leq 6} 
\|\p^\alpha(A_1(\usol_{0}),L_1(\usol_{0}))\|_{L^\infty}
\;\leq\; b_0 
\textstyle\sup_{|\alpha|\leq 7} \|\p^\alpha(p_i,\bklq)\|_{L^\infty}
\end{equation}
\end{subequations}
The $q_0$, $c_0$, $b_0$ depend only on $\E_\Phi$
and the gauge, so can influence our choice of $c$.

Set\footnote{%
The fact that $c \geq 2$ will be used later to get a nondegenerate frame.
The condition $J \geq 2$ will not be explicitly used,
it follows from $J \geq c \kc/\pmin \geq 2$
where $\kc \geq \pmin$ is
required for consistency of \ref{pminx} and \ref{der3}.
} $c = (c_0+1)(q_0+2)+1$.
Further, to invoke
Theorem \ref{theorem:asyts}
using Table \ref{table:parthmxxx},
set\footnote{%
The motivation for these choices is
that $A_i^0(\amc)$ and $L_i(\amc)$ will satisfy appropriate estimates.
The gaps $+1$ are
used to bridge, by choosing $\Lambda$ small,
 the difference between $\usol_{0}$
and $\amc$. One can use a smaller gap.
}
$q = q_0 + 1$ and
$Q = q + 1$ and
$z = (c_0+1) \kc$ and
$Z = (J+1)\pmin$ and
$b = b_0 K + 1$.
In particular,
$Qz < c\kc \leq J\pmin < Z$.
We have now specified
enough data to obtain the constants $\eps>0$ and $C > 0$ from
Theorem \ref{theorem:asyts}; they can influence our choice of $\Lambda$.

Set $\delta = \lambda^{J+1}$.
With a smallness condition on $\Lambda$
we have $0 < \delta \leq \eps$.
Now all data for 
Theorem \ref{theorem:asyts} is specified
by Table \ref{table:parthmxxx}
and we now check that all its assumptions are satisfied.
The dependence of $\usol_\alpha$ on the tuple
 \eqref{eq:ttuu} as spelled out in 
 Theorem \ref{thm:ex} is used below.
Statements that require
 a smallness assumption on $\Lambda$
 and that use
 \ref{gen3} and \ref{eps3} are tagged by $\star$ below
(note that $K$ 
 can influence $\Lambda$).
 Every term
 in $\amc-\usol_{0}$ comes, by \ref{concg3}, with at least a factor\footnote{%
 There can also be polynomial factors in $\tau$.}
 $\lambda^2 e^{-(p_i + p_j) \tau} \leq \Lambda^2 e^{-2\pmin \tau}$ for some $i,j = 1,2,3$.
\begin{itemize}
\item \ref{bbounded} is clear by compactness of $\mathbbm{T}^3$
and the definition of $\P_p$.
\item \ref{bk2k} uses \eqref{eq:a12}; the gap $q - q_0 = 1$; and $\star$.
\item \ref{blop} uses \eqref{eq:a1234} and \ref{der3};
the gap $z - c_0 \kc = \kc > 0$; and $\star$.
\item \ref{bcs1}, \ref{bcs2} use $\p(A_i(\usol_{0})) = 0$ for $i=1,2$; and $\star$.
\item \ref{jhgjdffd1}, \ref{jhgjdffd2}
use \eqref{eq:alal}
and \ref{gen3}
and $\dimx + 3 = 6$;
the gap $b - b_0K = 1$;
and $\star$\footnote{Note that \ref{gen3} controls $J + 7$ derivatives
of the data, hence controls $7$ derivatives of $\amc$.}.
\item \ref{hifhufn1}, \ref{hifhufn2}
use that $\amc$ is the truncation at order $J$ of a formal power series solution;
the bracket is first order;
use \ref{gen3} and $\dimx + 3 = 6$;
and $\star$.
More in detail, for \ref{hifhufn1},
note that $[\amc,\amc]$ as a polynomial in the $\sv{i}$ starts with terms of degree $J+2$
with coefficients (and their derivatives up to order $6$) controlled by \ref{gen3},
so all terms in $[\amc,\amc]$
come at least with a factor $\lambda^{J+2} e^{-(J+2)\pmin \tau}$ 
times a polynomial in $\tau$,
beating the product
of $\delta = \lambda^{J+1}$ with $e^{-Z\tau} = e^{-(J+1)\pmin \tau}$ using $\star$.
\end{itemize}
By Theorem \ref{theorem:asyts} we have existence
of a solution $\amc + \corr$ with $\corr = \ginj_1 \fld$;
see Table \ref{table:parthmxxx}.
The frame is nondegenerate by the following argument:
\begin{itemize}
\item
By assumption $D_1,D_2,D_3$ in \eqref{eq:ttuu} constitute a frame,
hence the four-dimensional frame of $\uzero$ is nondegenerate
when $\spar_i$ are replaced by nonzero real numbers.
Hence the frame of $\amc$
is nondegenerate for sufficiently large $\tau$ using \ref{concg2}\footnote{%
This is true even if the $p_i$ are of very different sizes
 say two small and one large.}.
\item 
Since $Z = (J+1)\pmin > J\pmin
   \geq c \kc \geq 2 \kc \geq 2 \sup_i \smash{\|p_i\|_{L^\infty}}$
we use \eqref{adec22} to conclude that
also $\amc + \corr$ 
with $\corr = \ginj_1 \fld$
has nondegenerate frame for large $\tau$.
\item
Note that $\amc$ has the form \eqref{eq:t0d0},
and one can choose $\ginj_1$ such that $\amc + \corr$ also has the form \eqref{eq:t0d0}.
Then Remark \ref{remark:nondegpreserved} implies
that nondegeneracy holds everywhere\footnote{%
Alternatively note that, whatever the choice of $\ginj_1$,
the $\theta_0$ part of the frame of $\amc + \corr$
is $\sum_{\mu = 0}^3 v^\mu \theta_0 \p_\mu$
for some $v^\mu$ with $v \approx (1,0,0,0)$ and a small change of coordinates
brings this into the form $\theta_0 \p_0$.}.
\end{itemize}

The associated 1-forms in Definition \ref{def:aof1} and \ref{def:aof2}
are closed, and exact because their integrals around the 1-cycles on $\mathbbm{T}^3$
vanish for large $\tau$ because they vanish for $\uzero$.
So Lemma \ref{lemma:gsf} applies.
The claims about the causal structure are
properties of the conformal metric,
see Definition \ref{def:nondeg} and \eqref{eq:nijm} and \eqref{eq:sss1},
and follow from $\uzero$.
\qed\end{proof}

\section{Discussion and outlook} \label{sec:discussion}

We have constructed  smooth big bang type solutions.
Formal completeness
for anisotropic elements makes it plausible
that these solutions yield an open set of smooth initial data,
up to gauge transformations.
Here we discuss informally the difficulties that one might encounter,
and a plan for addressing them, for $M^3 = \mathbbm{T}^3$.
\step
Recall that \cite{rs} obtain an open set of initial data,
in the near-spatially-flat-FLRW case.
We only consider anisotropic elements $p_1 \neq p_2 \neq p_3 \neq p_1$
see Remark \ref{remark:anisotropic}.
\step
The obvious strategy is to apply an inverse function theorem to 
a suitably gauged version of the map $U \subset \MC(\A) \to \MC(\P)$
to obtain local surjectivity. Subtly different ways of defining the map can
make a big difference in carrying this out, and one must be open to making modifications.
With this in mind,
recall our construction of the true solution as a truncation of a formal series solution
plus a correction, $\amc + \corr$.
The truncation $\amc$ is accurate at late times $\tau \to \infty$,
and the correction $\corr$ becomes more relevant at earlier times.

This perspective can be refined in an essential way by
analyzing both in time $\tau \in [0,\infty)$ and in spatial frequency (or momentum) $k \in \Z^3$.
There are two qualitatively different regimes, see Figure \ref{fig:regs}.
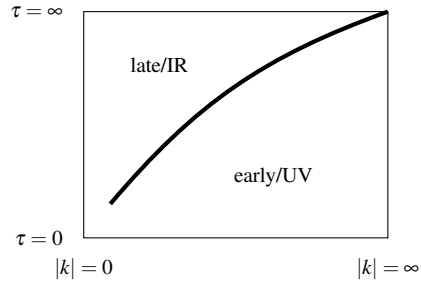
\begin{figure}[h]
\centering
\[
  \begin{tikzpicture}
  \draw (0,0) rectangle (4,3);
  \draw [ultra thick] (0.35,0.45) to[bend angle = 15, bend left] (4,3);
  \draw (0,0) node[anchor=east, xshift=-5] {$\tau=0$};
  \draw (0,3) node[anchor=east, xshift=-5] {$\tau=\infty$};
  \draw (1,2.3) node {late/IR};
  \draw (2.5,0.8) node {early/UV};
  \draw (0,0) node[anchor=north, yshift=-5] {$|k|=0$};
  \draw (4,0) node[anchor=north, yshift=-5] {$|k|=\infty$};
  \end{tikzpicture}
\]
\caption{Two regimes where IR stands for infrared or low frequency,
and UV stands for ultraviolet or high frequency.
The transition is not sharp, but is roughly
at $\tau = \text{const} \cdot \log |k| + \text{const}$
with constants that are also not sharply defined.
}\label{fig:regs}
\end{figure}
This follows from the causal structure,
schematically given by the conformal orthonormal frame
$\p_\tau$ and $e^{-\text{const}\cdot \tau} \p_x$, by which
$e^{-\text{const}\cdot \tau} |k| \approx 1$ is an important heuristic threshold,
and we conjecture that the following heuristics can be made rigorous:
\begin{itemize}
\item \emph{Late/IR regime.} Time derivatives dominate.
The behavior is essentially determined, in the sense of scattering,
by the associated graded of the BKL filtration.
One can use a gauge adapted to the filtration.
Note that MC elements in the associated graded are explicit and do not display wave behavior.
\item \emph{Early/UV regime.} 
Here one has typical wave behavior,
not easy to describe by an expansion.
This must be analyzed accordingly,
using gauges that yield a symmetric hyperbolic system and energy estimates.
\end{itemize}
One should be ready to use different gauges in the two regimes;
the word gauge here refers comprehensively
to gauge fixing conditions
and choosing a splitting into dynamic and constraint equations.
Note that we have done something very similar in this paper,
by constructing the truncation $\amc$ and the correction $\corr$
using two different gauges.
But while the splitting into $\amc$ and $\corr$ is similar to the late/IR versus early/UV splitting,
it is not as clear-cut,
and further analysis may benefit from using
a definite gauge depending only on $(\tau,k)$.

However large $\tau$,
there is always a threshold for $|k|$ above which one is in the early/UV
regime. But using energy estimates
to do what they are good for, namely controlling high frequencies,
one would effectively show that the late/IR behavior
as, governed by the associated graded, ultimately dominates.

For a nonlinear system,
or a linear system with non-constant coefficients,
different spatial frequencies $k$ interact,
and therefore the two regimes late/IR and early/UV are in permanent interaction,
and this must be controlled.
But the propagator of a \emph{linear} symmetric hyperbolic system
with sufficiently smooth coefficients admits Sobolev $H^s$ estimates
for positive and negative $s$,
which bounds transfer from low-to-high respectively high-to-low frequencies.
Such estimates for linear systems 
are useful when, and good motivation for, invoking an inverse function theorem.
While it seems difficult to obtain estimates that do not lose derivatives,
more precisely to find two Banach spaces such that linearization and inverse are bounded,
toy calculations suggest that one can get estimates that lose only finitely
many derivatives, hence the Nash-Moser inverse function theorem
in Frechet spaces is a natural candidate.

We propose that one should prove scattering, between 
true MC elements and the simpler associated graded MC elements.
To clarify what we mean by scattering,
suppose we have two systems, described by linear or nonlinear propagators $P_A$ and $P_B$,
so the composition law $P_A(a,\tau)P_A(\tau,b) = P_A(a,b)$ must hold and similar for $P_B$.
Scattering is the existence of the limits
$S_{AB} = \lim_{\tau \to +\infty} P_A(0,\tau) P_B(\tau,0)$
and $S_{BA} = \lim_{\tau \to +\infty} P_B(0,\tau) P_A(\tau,0)$
that must be inverses of one another. An important technique is to define
$P_X(a,b) = P_B(0,a)P_A(a,b)P_B(b,0)$, which is itself a system in the sense that
the composition law holds.
Then
$S_{AB} = \lim_{\tau \to +\infty} P_X(0,\tau)$ and $S_{BA} = \lim_{\tau \to +\infty} P_X(\tau,0)$.
The limits for $P_X$ are often easier to establish,
if one can calculate the infinitesimal generator of the linear or nonlinear system $X$.
This ought to be an effective technique in the case at hand,
where $A$ would be a dynamical system for $\MC(\E_\Phi)$
and $B$ would be a dynamical system for $\MC(\A)$.

Unfortunately,
we anticipate that scattering holds but with a nasty technical complication.
Namely $S_{AB}$ and $S_{BA}$ will only exist on, and will only establish a bijection between,
certain nontrivial Frechet submanifolds, but this will be sufficient
because directions transversal to these submanifolds
are equivalent to applying $\Orth(3)$ gauge transformations
that rotate the spatial frame.
These $\Orth(3)$ rotations depend on the spatial coordinates.
Intuitively, scattering requires that the frame be properly rotated
at each spatial point, for the filtration to control the late/IR behavior.
This kind of complication
should perhaps not be surprising in the context of a gauge theory.

If this plan can be implemented,
it would prove stability under perturbations of anisotropic initial data,
for big bang type singularities,
with very precise control over the behavior at $\tau \to \infty$ via scattering.
More importantly,
we expect that similar techniques 
will be relevant for the BKL conjectures for vacuum (cf.~Subsection \ref{subsec:heuristics}).

\section{Acknowledgements}
M.R.~is grateful to have received funding from ERC through grant agreement No.~669655.

\appendix

\setcounter{MaxMatrixCols}{14}
\newcommand{\kfreq}{k}
\section{Analysis of the $\mcs$ constraint equations}\label{app:constraints}

In this logically self-contained appendix we solve the 
$\mcs$
constraint equations.
The discussion is limited to the following situation:
\begin{itemize}
\item We are on the torus $\mathbbm{T}^3 = (\R/2\pi \Z)^3$.\\
This simplification allows one to use Fourier series,
other cases should be similar
(see Remark \ref{remark:sphere}).
\item We only consider solutions close to spatially homogeneous
and totally anisotropic.
\item We only consider smooth solutions.\\
One can construct finitely differentiable solutions in just the same way.
\end{itemize}
We show that locally the space of solutions
almost looks like a graph
of a smooth map of Frechet spaces.
Almost, because there are in addition a finite number $q$ of independent nonlinear
conditions that cut out a nonempty codimension $q$ real variety,
to leading order described by $q$ homogeneous quadratic equations.
Here $q=3$.
The main tool is to observe that the principal symbol of this constraint system
has full rank and therefore can be reduced to an elliptic problem.
 We do not make contact with literature
 about the Einstein constraint equations.
 This appendix is logically self-contained.
\step
Let $\mathcal{S} = C^\infty(\mathbbm{T}^3)$ be the Frechet space of smooth real
functions on the torus $\mathbbm{T}^3$.
Fourier series establish an explicit isomorphism to the Frechet space
of rapidly decreasing sequences.
The three standard partial derivatives will be denoted
$\p_1,\p_2,\p_3$.
\begin{definition}[Constraint equations]
Define the smooth nonlinear map
\begin{equation}\label{eq:csx}
\begin{aligned}
\mathbf{C} : N \subset \mathcal{S}^{15} & \to \mathcal{S}^4\\
   (D_1,D_2,D_3,p_1,p_2,p_3,\bklq,\scl,\phiconstraint)
    & \mapsto
   (A,B_1,B_2,B_3)
\end{aligned}
\end{equation}
where $D_1,D_2,D_3$ are three vector fields on $\mathbbm{T}^3$
that for analysis purposes we view as nine component functions,
and the rest are six functions.
By definition,
the domain $N$ corresponds to the open subset where
$D_1,D_2,D_3$ are linearly independent at every point of $\mathbbm{T}^3$.
We have abbreviated
\begin{align*}
A & = 3 \bklq^2 - p_2p_3 - p_3p_1 - p_1p_2\\
B_i & = -\tfrac12 D_i(p_j+p_k) 
- \tfrac12 c_{ij}^j (p_i-p_j)
+ \tfrac12 c_{ki}^k (p_i-p_k)
+ p_i D_i(\scl)
+ 3\bklq D_i(\phiconstraint)
\end{align*}
for all cyclic $(i,j,k) \in C$, see \eqref{eq:cycl}, and the structure functions are as in Definition \ref{def:scoeff}.
\end{definition}
\begin{definition}[Good intersection of real quadrics]\label{def:ndegb}
Let $U$ be a real Banach space and $\beta = (\beta^i)_{i=1\ldots n}$ with $\beta^i: U \times U \to \R$
a collection of $n$ continuous symmetric bilinear maps. We say $\beta$ defines a good intersection
if there exists a decomposition into closed subspaces
$U = U_1 \oplus \ldots \oplus U_n \oplus W$
with $\dim_\R U_i < \infty$ such that the
$\beta_{jk}^i = \beta^i|_{U_j \times U_k}$
satisfy:
\begin{itemize}
\item $\beta^i_{jk} = 0$ unless $i = j = k$.
\item $\beta^i_{ii}$ is an indefinite form, meaning it is of mixed signature.
\end{itemize}
\end{definition}
\begin{lemma}[Local structure of the solution space]
Suppose $p_{10},p_{20},p_{30}>0$ are three constants 
that are pairwise different, $p_{10} \neq p_{20} \neq p_{30} \neq p_{10}$,
interpreted as anisotropy.
Set
\[
x_0 = (\p_1,\p_2,\p_3,p_{10},p_{20},p_{30},p_{00},0,0) \in N
\]
where $p_{00} > 0$ is chosen so that $\mathbf{C}(x_0) = 0$.
Then there exist:
\begin{itemize}
\item An internal direct sum decomposition of Frechet spaces
\[
    \mathcal{S}^{15} = U \oplus V
\]
\item An open neighborhood $\Omega_U \subset U$ of the origin of $U$.
\item An open neighborhood $\Omega_V \subset V$ of the origin of $V$.
\item A smooth map $\varphi : \Omega_U \subset U \to \Omega_V \subset V$ 
with $\varphi(0) = 0$.
\item A smooth map $B : \Omega_U \subset U \to \R^3$ with
$B(0) = 0$ and $DB(0) = 0$
and whose second derivative $D^2B(0)$
defines a good intersection in the sense of Definition \ref{def:ndegb}.
\end{itemize}
such that
\[
	\mathbf{C}^{-1}(\{0\}) \cap (x_0 + \Omega_U + \Omega_V)
        \;=\;
           \big\{
             x_0 + (u \oplus \varphi(u)) \mid u \in \Omega_U,\; B(u) = 0
           \big\}
\]
\end{lemma}
\begin{proof}[Sketch]
Solving $A=0$ yields $\bklq = \bklq(p_1,p_2,p_3) > 0$.
Substituting this, we get a map
\begin{align*}
\mathbf{C}' : N' \subset \mathcal{S}^{14} & \to \mathcal{S}^3\\
   x = (D_1,D_2,D_3,p_1,p_2,p_3,\scl,\phiconstraint)
    & \mapsto
   (B_1,B_2,B_3)
\end{align*}
where $N'$
is the set where $D_1,D_2,D_3$ are nondegenerate and $p_1,p_2,p_3 > 0$.
The operator $\mathbf{C}'$ is quasilinear first order.
Parsing yields three unique nonlinear maps
$A^i: M' \subset \R^{14} \to \Hom_\R(\R^{14},\R^3)$
such that
\[
\mathbf{C}'(x) = \textstyle\sum_{i=1}^3 A^i(x) \p_i x
\]
The open subset $M' = \text{GL}(\R^3) \times (0,\infty)^3 \times \R^2$ is
simply the pointwise version of $N'$.
So $N'$ coincides with the elements in $\mathcal{S}^{14}$
with values in $M'$.
The principal symbol at $x_0$ has constant coefficients,
\begin{multline*}
\sigma(\kfreq) = A^i(x_0) \kfreq_i=\\
\begin{pmatrix}
0 & \Delta_{12}\kfreq_2 & \Delta_{13} \kfreq_3 &
0 & \Delta_{21}\kfreq_1 & 0 &
0 & 0 & \Delta_{31}\kfreq_1 &
0 & -\tfrac12 \kfreq_1 & -\tfrac12 \kfreq_1 &
p_{10}\kfreq_1 & 3p_{00} \kfreq_1\\
\Delta_{12}\kfreq_2 & 0 & 0 &
\Delta_{21}\kfreq_1 & 0 & \Delta_{23}\kfreq_3 &
0 & 0 & \Delta_{32}\kfreq_2 &
-\tfrac12 \kfreq_2 & 0 & -\tfrac12 \kfreq_2 &
p_{20}\kfreq_2 & 3p_{00} \kfreq_2\\
\Delta_{13}\kfreq_3 & 0 & 0 &
0 & \Delta_{23}\kfreq_3 & 0 &
\Delta_{31}\kfreq_1 & \Delta_{32}\kfreq_2 & 0 &
-\tfrac12 \kfreq_3 & -\tfrac12 \kfreq_3 & 0 &
p_{30}\kfreq_3 & 3p_{00} \kfreq_3
\end{pmatrix}
\end{multline*}
where $\Delta_{ij} = \tfrac12(p_{i0}-p_{j0})$
and $p_{00} = \bklq(p_{10},p_{20},p_{30})$.
The assumption $\Delta_{ij} \neq 0$ when $i \neq j$
implies that $\sigma(\kfreq) \in \Hom_\R(\R^{14},\R^3)$
has full rank equal to three
for all $\kfreq \in \R^3-0$.
For every $\kfreq \in \R^3$ decompose
\[
\R^{14} = U_\kfreq \oplus V_\kfreq
\]
where $U_\kfreq = \ker \sigma(\kfreq)$ and $V_\kfreq$ is some complement.
Explicitly, let $V_\kfreq$ be the orthogonal complement relative to the standard inner product on $\R^{14}$.
Note that $U_0 = \R^{14}$ and $V_0 = 0$.
If $\kfreq \neq 0$
then $\dim_\R U_\kfreq = 11$ respectively $\dim_\R V_\kfreq = 3$
and they are smooth vector bundles.
Note that $U_{-\kfreq} = U_{\kfreq}$ and $V_{-\kfreq} = V_{\kfreq}$.
Set\footnote{%
Recall that $\mathcal{S}$ are real valued functions.}
\begin{align*}
U' & = \{ y \in \mathcal{S}^{14} \mid \forall \kfreq:\;\widehat{y}(\kfreq) \in U_\kfreq \oplus i U_\kfreq \}\\
V' & = \{ y \in \mathcal{S}^{14} \mid \forall \kfreq:\;\widehat{y}(\kfreq) \in V_\kfreq \oplus i V_\kfreq \}
\end{align*}
where $\widehat{y}(\kfreq) \in \C^{14}$
denotes the Fourier coefficient for $\kfreq \in \Z^3$.
Note that $\mathcal{S}^{14} = U' \oplus V'$ is an internal
sum decomposition of Frechet spaces.
Every $x \in \mathcal{S}^{14}$ has a unique decomposition
 $x = x_0 + u + v$ where $u \in U'$ and $v \in V'$.
If $u,v$ are small then $x \in N'$.
Then
\[
\mathbf{C}'(x) = A(x)\p x = 0
\qquad
\Longleftrightarrow
\qquad
A(x_0)\p v = f_u(v)
\]
where by definition $f_u(v) = - (A(x_0+u+v)-A(x_0))\p (u+v)$.
We have used $A(x_0)\p u = 0$,
which holds by definition of $U'$.
For fixed $u$ the operator $f_u$ is a quasilinear first order
differential operator and we decompose
$f_u(v) = P_u + Q_uv + R_u(v)$ where
$R_u(v) = \mathcal{O}(v^2)$ as $v \to 0$.
Here $P_u \in \mathcal{S}^3$
satisfies $P_u = \mathcal{O}(u^2)$ as $u \to 0$.
Here $Q_u$ is a linear and $R_u$ a quasilinear first order differential operator.
Note that $Q_0 = 0$ so that $Q_u$ is small if $u$ is small;
this means that a Sobolev operator norm
losing one derivative is small.
Let $\Pi: \mathcal{S}^3 \to \R^3$ be the map that extracts
the zeroth Fourier coefficients, which are real. Note that
\[
\mathbf{C}'(x) = A(x)\p x = 0
\qquad
\Longleftrightarrow
\qquad
\begin{aligned}
A(x_0)\p v &= (1-\Pi) f_u(v)\\
\Pi f_u(v) &= 0
\end{aligned}
\]
since $\Pi A(x_0)\p v = 0$ for all $v \in V'$. We solve the first equation on the right, then the second equation.
\begin{itemize}
\item
Since $V_0=0$
and 
$A(x_0)\p$ is elliptic on $V'$ by construction,
the first equation admits,
using standard Sobolev estimates,
a solution $v = \varphi'(u)$ for small $u$ obtained by a fixed point iteration
starting at $0$.
In particular, inverting the (constant coefficient) elliptic operator gains one derivative
 compensating the loss of one derivative in $v \mapsto f_u(v)$.
 One first constructs a solution by iteration in a Sobolev space for a fixed but sufficiently large index,
 then one shows that the same iterates converge in the Sobolev space of index $k$ for all $k$, by induction on $k$.
The map $\varphi'$ is a smooth map of Frechet spaces,
and the space of solutions to the first equation is the graph of
$\varphi': \Omega_{U'} \subset U' \to \Omega_{V'} \subset V'$
where $\Omega_{U'}$ and $\Omega_{V'}$ are open sets containing the origin,
and furthermore $\varphi'(0) = 0$ and $D \varphi'(0) = 0$.
The derivative vanishes because
$P_u = \mathcal{O}(u^2)$ as $u \to 0$ by construction of $U'$.
\item
The second equation is equivalent to
$B'(u) = 0$
where
\[
B' : \Omega_{U'} \subset U' \to \R^3,\;u \mapsto \Pi f_u(\varphi'(u))
\]
which is a smooth nonlinear map
\[
B'(u) = -\Pi (A(x_0+u)-A(x_0))\p u + \mathcal{O}(u^3)
\]
Clearly $B'(0) = DB'(0) = 0$.
The second derivative $D^2B'(0)$
defines a good intersection in the sense of Definition \ref{def:ndegb},
for instance, one can take the three subspaces
$U_1 \oplus U_2 \oplus U_3 \subset U'$ where
\begin{align*}
U_1 & \;=\; \{\text{only $\kfreq = (\pm 1,0,0)$ Fourier coefficients}\}\\
U_2 & \;=\; \{\text{only $\kfreq = (0,\pm 1,0)$ Fourier coefficients}\}\\
U_3 & \;=\; \{\text{only $\kfreq = (0,0,\pm 1)$ Fourier coefficients}\}
\end{align*}
with\footnote{The restriction to $U_i \times U_j$ with $i \neq j$ is zero.
The signature of the restriction to $U_1 \times U_1$ is indefinite
already when restricting to the last 5 of 14 components of the
Fourier coefficients in $U_{(\pm 1,0,0)} \oplus i U_{(\pm 1,0,0)}$.
} $\dim_\R U_i = 22$.
\end{itemize}
To finish the proof, set $U \simeq U'$
and $V \simeq V' \oplus \mathcal{S}$
where the second summand is for $\bklq$;
the map $\varphi$ is a direct sum of $\varphi'$
and the map 
$\bklq = \bklq(p_1,p_2,p_3) > 0$
introduced at the beginning of this proof;
and $B \simeq B'$.
\qed\end{proof}
\begin{remark}[The sphere]\label{remark:sphere}
The argument above should also work for the sphere $S^3$.
We sketch how this works.
Let
$L_1,L_2,L_3$
and
$R_1,R_2,R_3$
be the usual frames of left- and right-invariant vector fields on $S^3$,
\[
[L_i,L_j] = \eps_{ijk} L_k
\qquad
[R_i,R_j] = \eps_{ijk} R_k
\qquad
[L_i,R_j] = 0
\]
Expand $D_i = {D_i}^j L_j$
so in \eqref{eq:csx} we can keep working with $\mathcal{S}^{14}$.
Use $x_0 =
(L_1,L_2,L_3,p_{10},p_{20},p_{30},p_{00},0,0)$.
Denote $L_{\pm} = L_1 \pm i L_2$.
Expand $\mathbf{C}'(x) = A^+(x) L_+ x + A^-(x) L_-x + A^3(x) L_3x$
where $A^+,A^-,A^3: M' \subset \R^{14} \to \Hom_{\C}(\C^{14},\C^{3})$
satisfying appropriate reality conditions.
The `Fourier coefficients' of an element of
$\C \otimes \mathcal{S}$
are elements of $V_m \otimes_\C V_m$ where
$\dim_\C V_m = 2m+1$
and where the left (resp.~right) $V_m$
is the irreducible unitary representation for the $L_i$ (resp.~$R_i$).
The ellipticity statement is that the linear map
\[
	\sigma(m) = A^+(x_0) L_+ + A^-(x_0) L_- + A^3(x_0)L_3
        \;\in\; \Hom_{\C}(\C^{14} \otimes V_m, \C^3 \otimes V_m) 
\]
has full rank for all $m > 0$. Let $V_{m0} \subset V_m$ be the lowest weight
using the Cartan subalgebra spanned by $L_3$. Using the anisotropy,
one can check that $A^3(x_0)$
has full rank when restricted to the joint kernel of $A^+(x_0)$ and $A^-(x_0)$,
which implies that $\C^3 \otimes V_{m0} \subset \image \sigma(m)$ when $m \neq 0$.
One can check that $A^+(x_0)$ has full rank which implies by an induction
on the weight filtration that $\sigma(m)$ has full rank.
\end{remark}


\end{document}